\documentclass{amsart}
\usepackage[all]{xy}
\usepackage{color}
\usepackage{amsthm}
\usepackage{amssymb}
\usepackage[colorlinks=true]{hyperref}

\numberwithin{equation}{section}

\newtheorem{theorem}[equation]{Theorem}
\newtheorem*{theorem*}{Theorem}
\newtheorem{lemma}[equation]{Lemma}

\newtheorem*{conjecture*}{Mamma Conjecture}
\newtheorem*{conjecture1*}{Mamma Conjecture (revisited)}
\newtheorem{proposition}[equation]{Proposition}
\newtheorem{corollary}[equation]{Corollary}
\newtheorem*{corollary*}{Corollary}

\theoremstyle{remark}
\newtheorem{definition}[equation]{Definition}

\newtheorem{example}[equation]{Example}

\newtheorem{notation}[equation]{Notation}

\theoremstyle{remark}
\newtheorem{remark}[equation]{Remark}

\newcommand{\cA}{{\mathcal A}}
\newcommand{\cB}{{\mathcal B}}
\newcommand{\cC}{{\mathcal C}}

\newcommand{\cH}{{\mathcal H}}

\newcommand{\cO}{{\mathcal O}}
\newcommand{\cP}{{\mathcal P}}
\newcommand{\cQ}{{\mathcal Q}}

\newcommand{\cS}{{\mathcal S}}

\newcommand{\cV}{{\mathcal V}}
\newcommand{\cW}{{\mathcal W}}

\newcommand{\bbC}{\mathbb{C}}

\newcommand{\bbF}{\mathbb{F}}
\newcommand{\bbG}{\mathbb{G}}

\newcommand{\bbN}{\mathbb{N}}

\newcommand{\bbR}{\mathbb{R}}

\newcommand{\bbV}{\mathbb{V}}
\newcommand{\bbQ}{\mathbb{Q}}
\newcommand{\bbZ}{\mathbb{Z}}

\DeclareMathOperator{\Id}{Id}
\DeclareMathOperator{\id}{id}

\newcommand{\Hom}{\mathrm{Hom}}

\newcommand{\too}{\longrightarrow}

\newcommand{\ie}{\textsl{i.e.}\ }
\newcommand{\eg}{\textsl{e.g.}}

\newcommand{\Z}{{\mathbb Z}}
\newcommand{\Q}{{\mathbb Q}}

\newcommand{\N}{{\mathbb N}}
\newcommand{\C}{{\mathbb C}}

\newcommand{\R}{{\mathbb R}}

\title[BC-systems, categorification, QSM-systems, and Weil numbers]{Bost--Connes systems, categorification, quantum statistical mechanics, and Weil numbers}
\author{Matilde Marcolli and Gon{\c c}alo~Tabuada}

\address{Matilde Marcolli, Mathematics Department, Mail Code 253-37, Caltech, 1200 E.~California Blvd. Pasadena, CA 91125, USA}
\email{matilde@caltech.edu} 
\urladdr{http://www.its.caltech.edu/~matilde}

\address{Gon{\c c}alo Tabuada, Department of Mathematics, MIT, Cambridge, MA 02139, USA}
\email{tabuada@math.mit.edu}
\urladdr{http://math.mit.edu/~tabuada}
\thanks{M.~Marcolli was partially supported by the NSF grants  
DMS-1007207, DMS-1201512, and PHY-1205440. G.~Tabuada was partially 
supported by a NSF CAREER Award.} 
\date{\today}


\begin{document}

\begin{abstract}
In this article we develop a broad generalization of the classical Bost-Connes system, where roots of unit are replaced by an algebraic datum consisting of an abelian group and a semi-group of endomorphisms. Examples include roots of unit, Weil restriction, algebraic numbers, Weil numbers, CM fields, germs, completion of Weil numbers, etc. Making use of the Tannakian formalism, we categorify these algebraic data. For example, the categorification of roots of unit is given by a limit of orbit categories of Tate motives while the categorification of Weil numbers is given by Grothendieck's category of numerical motives over a finite field. To some of these algebraic data (\eg\ roots of unity, algebraic numbers, Weil numbers, etc), we associate also a quantum statistical mechanical system with several remarkable properties, which  generalize those of the classical Bost--Connes system. The associated partition function, low temperature Gibbs states, and Galois action on zero-temperature states are then studied in detail. For example, we show that in the particular case of the Weil numbers the partition function and the low temperature Gibbs states can be described as series of polylogarithms.
\end{abstract}
\maketitle
\vskip-\baselineskip
\vskip-\baselineskip
\tableofcontents
\section{Introduction}
\subsection*{Bost-Connes systems}
Let $\bbQ/\bbZ$ be the abelian group of roots of unit. In \cite{BC}, Bost and Connes introduced the group algebra $\bbQ[\bbQ/\bbZ]$, the algebra endomorphisms
\begin{eqnarray}\label{eq:morph-1}
\bbQ[\bbQ/\bbZ] \too \bbQ[\bbQ/\bbZ] & s\mapsto ns & n \in \bbN\,,
\end{eqnarray} 
and also the $\bbQ$-linear additive maps 
\begin{eqnarray}\label{eq:map-1}
\bbQ[\bbQ/\bbZ] \too \bbQ[\bbQ/\bbZ] & s \mapsto \sum_{s'|n s'=s}s' & n \in \bbN\,.
\end{eqnarray}
The (continuous) action of the absolute Galois group $\mathrm{Gal}(\overline{\bbQ}/\bbQ)$ on $\bbQ/\bbZ$ extends to $\bbQ[\bbQ/\bbZ]$ making \eqref{eq:morph-1}-\eqref{eq:map-1} $\mathrm{Gal}(\overline{\bbQ}/\bbQ)$-equivariant. The above piece of structure is very important since it encodes all the arithmetic information about roots of unit. Making use of it, Bost and Connes constructed in {\em loc. cit.} a quantum statistical mechanical (=QSM) system $(\cA_{\bbQ/\bbZ}, \sigma_t)$ with the following remarkable properties:
\begin{itemize}
\item[(i)] The partition function $Z(\beta)$ agrees with the Riemman zeta function.
\item[(ii)] The low temperature Gibbs states are given by polylogarithms evaluated at roots of unit.
\item[(iii)] The group $\mathrm{Gal}(\overline{\bbQ}/\bbQ)^{\mathrm{ab}} \simeq \widehat\Z^\times$ acts by symmetries on the QSM-system. This action induces an action on the set of low temperature Gibbs states and on their zero temperature limits (=ground states), where it recovers the Galois action on $\bbQ^{\mathrm{ab}}$.   
\end{itemize}
The foundational article \cite{BC} was the starting point in the study of the Riemann zeta function via noncommutative geometry tools, see \cite{CoRH,CCM,CCM3,CoMa}. The main goal of this article is to extend the above technology from $\bbQ/\bbZ$ to other interesting arithmetic settings such as the Weil numbers, for example.

\subsection*{Statement of results}
Let $k$ be a field of characteristic zero, $G:=\mathrm{Gal}(\overline{k}/k)$, $\Sigma$ an abelian group equipped with a continuous $G$-action, and $\sigma_n: \Sigma \to \Sigma, n \in \bbN$, $G$-equivariant homomorphisms such that $\sigma_{nm}=\sigma_n \circ \sigma_m$. We denote by $\alpha(n)$ the cardinality of the kernel of $\sigma_n$. Out of these data, we can construct the $k$-algebra $k[\Sigma]$, the $k$-algebra endomorphisms
\begin{eqnarray}\label{eq:morph-2}
\underline{\sigma_n}:k[\Sigma] \too k[\Sigma] & s \mapsto \sigma_n(s) & n \in \bbN\,,
\end{eqnarray}
and also the $k$-linear additive maps (defined only when $\alpha(n)$ is finite)
\begin{eqnarray}\label{eq:map-2}
\underline{\rho_n}:k[\Sigma_n] \too k[\Sigma] & s \mapsto \sum_{s'|\sigma_n(s')=s}s' & n \in \bbN\,,
\end{eqnarray}
here $\Sigma_n$ stands for the image of $\sigma_n$. The $G$-action on $\Sigma$ extends to $k[\Sigma]$ making \eqref{eq:morph-2}-\eqref{eq:map-2} $G$-equivariant. In the particular case where $k=\bbQ$, $\Sigma=\bbQ/\bbZ$, and $\sigma_n$ is multiplication of $n$, we recover the original construction of Bost-Connes. Other examples include Weil restriction, algebraic numbers, Weil numbers, CM fields, germs, completion of Weil numbers, etc; consult \S\ref{sec:BC-systems}. Making use of the $G$-action, we can also consider the $k$-algebras $\overline{k}[\Sigma]^G$ and $k[\Sigma]^G$. They carry canonical Hopf structures and consequently give rise to affine group $k$-schemes $\mathrm{Spec}(\overline{k}[\Sigma]^G)$ and $\mathrm{Spec}(k[\Sigma]^G)$. Our first main result, which summarizes the content of \S\ref{sec:categorification}, is the~following:
\begin{theorem}[Categorification]\label{thm:main1}
\begin{itemize}
\item[(i)] The affine group $k$-scheme $\mathrm{Spec}(\overline{k}[\Sigma]^G)$ agrees with the Galois group of the neutral Tannakian category $\mathrm{Vect}_\Sigma^{\overline{k}}(k)$ of pairs $(V, \bigoplus_{s \in \Sigma}\overline{V}^s)$, where $V$ is a finite dimensional $k$-vector space and $\bigoplus_{s \in \Sigma}\overline{V}^s$ an appropriate $\Sigma$-grading of $\overline{V}:=V \otimes \overline{k}$; see Definition \ref{def:categorification-1}. 
\item[(ii)] The $k$-algebra endomorphisms $\underline{\sigma_n}: \overline{k}[\Sigma]^G \to \overline{k}[\Sigma]^G$ are induced from $k$-linear additive symmetric monoidal functors $\underline{\bf \sigma_n}: \mathrm{Vect}_\Sigma^{\overline{k}}(k) \to \mathrm{Vect}_\Sigma^{\overline{k}}(k)$. Similarly, the $k$-linear additive maps $\underline{\rho_n}: \overline{k}[\Sigma_n]^G \to \underline{k}[\Sigma]^G$ are induced from $k$-linear additive functors $\underline{\bf \rho_n}: \mathrm{Vect}_{\Sigma_n}^{\overline{k}}(k) \to \mathrm{Vect}_\Sigma^{\overline{k}}(k)$.
\item[(iii)] Items (i)-(ii) hold {\em mutatis mutandis} with $\overline{k}[\Sigma]^G$ replaced by $k[\Sigma]^G$ and $\mathrm{Vect}_\Sigma^{\overline{k}}(k)$ replaced by the subcategory $\mathrm{Vect}_\Sigma^k(k)$ of pairs $(V,\bigoplus_{s \in \Sigma} \overline{V^s})$, where $\bigoplus_{s \in \Sigma} V^s$ is an appropriate $\Sigma$-grading of $V$; see Definition \ref{def:smaller}.
\item[(iv)] When $\Sigma$ admits a $G$-equivariant embedding into $\overline{k}^\times$, $\mathrm{Vect}_\Sigma^{\overline{k}}(k)$ is equivalent to the neutral Tannakian category $\mathrm{Aut}_\Sigma^{\overline{k}}(k)$ of pairs $(V, \Phi)$, where $V$ is a finite dimensional $k$-vector space and $\Phi: \overline{V} \stackrel{\sim}{\to} \overline{V}$ a $G$-equivariant diagonalizable automorphism whose eigenvalues belong to $\Sigma$. Moreover, when $\sigma_n$ is given by multiplication by $n$, the functor $\underline{\bf \sigma_n}$, resp. $\underline{\bf \rho_n}$, reduces to the Frobenius, resp. Verschiebung, endofunctor of $\mathrm{Aut}_\Sigma^{\overline{k}}(k)$.
\end{itemize}
\end{theorem}
\begin{remark}
When the $G$-action on $\Sigma$ is trivial (\eg\ $k=\overline{k}$), we have $\overline{k}[\Sigma]^G=k[\Sigma]^G=k[\Sigma]$. In this case, $\mathrm{Vect}_\Sigma^{\overline{k}}(k)=\mathrm{Vect}_\Sigma^k(k)$ reduces to the category $\mathrm{Vect}_\Sigma(k)$ of $\Sigma$-graded $k$-vector spaces. Similarly, $\mathrm{Aut}_\Sigma^{\overline{k}}(k)$ reduces to the category $\mathrm{Aut}_\Sigma(k)$ of diagonalizable automorphisms $\Phi: V \stackrel{\sim}{\to} V$ whose eigenvalues belong to $\Sigma$.
\end{remark}
Intuitively speaking, Theorem \ref{thm:main1} provides two simple ``models'' of the Tannakian categorification of the $k$-algebras $\overline{k}[\Sigma]^G, k[\Sigma]^G, k[\Sigma]$ and of the maps \eqref{eq:morph-2}-\eqref{eq:map-2}. This categorification is often related to the theory of (pure) motives. For example, in the original case of Bost and Connes, the category $\mathrm{Vect}_{\bbQ/\bbZ}(\bbQ)$ can be described as the limit $\mathrm{lim}_{n \geq 1} \mathrm{Tate}(\bbQ)_\bbQ/_{\!-\otimes \bbQ(n)}$ of orbit categories of Tate motives; see Proposition \ref{prop:lim}. More interestingly, in the case where $\Sigma$ is the abelian group of Weil numbers $\cW(q)$, $\mathrm{Vect}_{\cW(q)}^{\overline{\bbQ}}(\overline{\bbQ})$ agrees with the category of numerical motives over $\bbF_q$ with $\overline{\bbQ}$-coefficients; see Theorem \ref{thm:Milne1}.

Our second main result, which summarizes the content of \S\ref{sec:QSM}, is the following:
\begin{theorem}[QSM-systems]\label{thm:main22}
Assume that $k \subseteq \bbC$. Given a pair $(\Sigma, \sigma_n)$ as above and an appropriate set of $G$-equivariant embeddings $\mathrm{Emb}_0(\Sigma, \overline{\bbQ}^\times)$ (see Definition \ref{goodBC}), we construct a QSM-system $(\cA_{(\Sigma,\sigma_n)}, \sigma_t)$ with the remarkable properties:
\begin{itemize}
\item[(i)] Given an embedding $\iota \in \mathrm{Emb}_0(\Sigma, \overline{\bbQ}^\times)$, let $N_\iota(\Sigma):= \{|\iota(s)|\,\mathrm{such}\,\mathrm{that} \in \Sigma\}$ be the associated countable multiplicative subgroup of $\bbR^\times_+$. When $N_\iota(\Sigma)$ is the union of finitely (resp. infinitely) many geometric progressions, the choice of a semi-group homomorphism $g: \bbN \to \bbR^\times_+$ allows us to describe the partition functor $Z(\beta)$ of the QSM-system as follows
\begin{eqnarray*}
\sum_{\eta\leq 1 \in N_\iota(\Sigma)}  \sum_{n\geq 1} \eta^{\alpha(n) \beta} g(n)^{-\beta} && 
(\mathrm{resp.}\,\,\sum_{n\geq1} g(n)^{-\beta} \zeta(\beta \alpha(n)))\,,
\end{eqnarray*}
where $\zeta$ stands for the Riemann zeta function. The left (resp.~right) hand side series converges for all $\beta>\beta_0$ (resp. $\beta > \max\{ \beta_0, 3/2 \}$), 
where $\beta_0$ denotes the exponent of convergence of $\sum g(n)^{-\beta}$.
\item[(ii)] The low temperature Gibbs states $\varphi_{\beta,\iota}(s)$, evaluated at $s\in \Sigma$ with $N_\iota(s)=1$, are given by the following expressions:
\begin{eqnarray*}
Z(\beta)^{-1} \sum_{\eta\leq1\in N_\iota(\Sigma)}\sum_{n\geq 1} \iota(s)^n \eta^{\alpha(n) \beta} g(n)^{-\beta} && (\mathrm{resp.}\,\, Z(\beta)^{-1} \sum_{n\geq 1} \iota(s)^n \zeta(n\beta) g(n)^{-\beta}) \,.
\end{eqnarray*}
\item[(iii)] The group $\tilde Z(G_\Sigma)$ (see Definition
\ref{def:GSigma}) acts on the set of low temperature Gibbs states and on their ground states, where it 
agrees with the Galois $\tilde Z(G_\Sigma)$-action on $\iota(\Sigma)$.
\end{itemize}
\end{theorem}

In the particular case where $k=\bbQ$, $\Sigma=\bbQ/\bbZ$, $\sigma_n$ is multiplication by $n$, and $g(n)=n$, we have $\tilde Z(G_\Sigma)={\rm Gal}(\bar\Q/\Q)$ and the remarkable properties (i)-(iii) of Theorem \ref{thm:main22} reduce to those of the classical Bost-Connes system. Therefore, the QSM-system of Theorem \ref{thm:main22} greatly generalizes the original one of Bost-Connes. Here are some examples (which summarize the content of \S\ref{sec:examples}):

\begin{example}[Algebraic numbers]
When $k=\bbQ$, $\Sigma=\overline{\bbQ}^\times$, and $\sigma_n$ is 
raising to the $n^\mathrm{th}$ power, the choice of an appropriate set of $G$-equivariant embeddings 
$\mathrm{Emb}_0(\overline{\bbQ}^\times, \overline{\bbQ}^\times)$ gives rise to a QSM-system $(\cA_{(\overline{\bbQ}^\times, \sigma_n)}, \sigma_t)$ with the following properties: (i) the partition function $Z(\beta)$ is given by $\sum_{n\geq1} g(n)^{-\beta} \zeta(\beta \alpha(n))$; (ii) the low temperature Gibbs states, evaluated at $s\in \overline{\Q}^\times$ with $|s| = 1$, are given by 
$$ \varphi_{\beta,\iota}(s) = \frac{\sum_{n\geq 1} \iota(s)^n \zeta(\beta n) 
g(n)^{-\beta}}{\sum_{n\geq 1} \zeta(\beta n) g(n)^{-\beta}}\,;$$
(iii) the group $\tilde Z(G_\Sigma)$ is trivial.
\end{example}
\begin{example}[Weil numbers of weight zero]
Let $q=p^r$ be a prime power. When $k=\bbQ$, $\Sigma$ is the group of Weil numbers of weight zero $\cW_0(q)$, 
and $\sigma_n$ is raising to the $n^\mathrm{th}$ power, the choice of an appropriate set of $G$-equivariant embeddings $\mathrm{Emb}_0(\cW_0(q), \overline{\bbQ}^\times)$ gives rise to a QSM-system $(\cA_{(\cW_0(q), \sigma_n)}, \sigma_t)$ with the following properties (which are very similar to those of the classical Bost-Connes system): (i) the partition function $Z(\beta)$ is given by $\sum_{n\geq 1} g(n)^{-\beta}$. When $g(n)=n$, it reduces to the Riemann zeta function; (ii) the low temperature Gibbs states, evaluated at $s \in \cW_0(q)$, are given by 
$$  \varphi_{\beta,\iota}(s) =\frac{{\rm Li}_\beta(\iota(s))}{\zeta(\beta)}\,;$$
(iii) the group $\tilde Z(G_{\cW_0(q)})$ acts on the set of low temperature Gibbs states and on their ground states, where it agrees with the Galois $\tilde Z(G_{\cW_0(q)})$-action on $\iota(\cW_0(q))$.
\end{example}
\begin{example}[Weil numbers]
When $k=\bbQ$, $\Sigma$ is the group of Weil numbers $\cW(q)$, and $\sigma_n$ is raising to
the $n^\mathrm{th}$ power, the choice of an appropriate set of $G$-equivariant embeddings $\mathrm{Emb}_0(\cW(q), \overline{\bbQ}^\times)$ gives rise to a QSM-system $(\cA_{(\cW(q), \sigma_n)}, \sigma_t)$ with the following properties: (i) the partition function $Z(\beta)$ is given by a series of polylogarithm functions $\sum_{k \geq 0}  {\rm Li}_\beta(q^{-k\beta/2})$; (ii) the low temperature Gibbs states, evaluated at $s\in \cW_0(q)$, are given by
$$ \varphi_{\beta,\iota}(s)= \frac{\sum_{k \geq 0} {\rm Li}_\beta(\iota(s) q^{-k\beta/2})}{\sum_{k \geq 0} 
{\rm Li}_\beta(q^{-k\beta/2})}\,; $$
(iii) the group $\tilde Z(G_{\cW(q)})$ acts on the set of low temperature Gibbs states and on their ground states, where it agrees with the Galois $\tilde Z(G_{\cW(q)})$-action on $\iota(\cW_0(q))$.
\end{example}

Finally, in \S\ref{sec:Weil+completion} we outline the construction of ``diagonal'' QSM-systems associated to Weil restriction and completion. The details will appear in a forthcoming article.

\smallskip

\begin{notation}
Throughout the article $k$ will be a field of characteristic zero. We will denote by $G$ its absolute Galois group $\mathrm{Gal}(\overline{k}/k)$. Unless stated differently, all tensor products will be taken over $k$.
\end{notation}

\section{Bost-Connes systems}\label{sec:BC-systems}

In this section we introduce the general notion of an abstract/concrete Bost-Connes system and describe several examples. Given a set $S$, we denote by $\cP(S)$ the set of subsets of $S$ and by $\cP(S)^n$ the set of subsets of $S$ with cardinality $n \in \bbN$.
\begin{definition}\label{def:BC-data}
An {\em abstract Bost-Connes datum $(\Sigma, \sigma_n)$} consists of:
\begin{itemize}
\item[(i)] An abelian group $\Sigma$ equipped with a $G$-action $G \to \mathrm{Aut}(\Sigma)$. We assume that the $G$-action is {\em continuous}, \ie that $\Sigma=\bigcup_l \Sigma^{\mathrm{Gal}(\overline{k}/l)}$ where $l$ runs through the finite Galois field extensions of $k$ contained in $\overline{k}$.
\item[(ii)] Group homomorphisms $\sigma_n:\Sigma \to \Sigma, n \in \bbN$. We assume that $\sigma_n$ is $G$-equivariant and that $\sigma_{nm}=\sigma_n \circ \sigma_m$ for every $n,m \in \bbN$.
\end{itemize}
\end{definition}
\begin{notation}
Let us write $\alpha(n)$ for the cardinality of the kernel of $\sigma_n$.
\end{notation}
\begin{definition}\label{def:BC-data2}
An abstract Bost-Connes datum $(\Sigma,\sigma_n)$ is called {\em concrete} if $\alpha(n)$ is finite for every $n \in \bbN$ and the assignment $\alpha: \N \to \N,\, n \mapsto \alpha(n)$, is a non-trivial multiplicative semi-group 
homomorphism, \ie $\alpha(nm) =\alpha(n)\cdot\alpha(m)$.
\end{definition}
By definition, every concrete Bost-Connes datum is also an abstract Bost-Connes datum. 
The converse is false; see Examples \ref{ex:6}-\ref{ex:7}, \ref{ex:0L}--\ref{ex:0L11}, and \ref{ex:last} below.

\begin{notation}
Let $\Sigma_n$ be the image of $\sigma_n$ and $\rho_n: \Sigma_n \to \cP(\Sigma)$ the map that sends an element of $\Sigma_n$ to its pre-image under $\sigma_n$. Note that $\Sigma_n$ is a subgroup of $\Sigma$ which is stable under the $G$-action, that $\rho_n$ is $G$-equivariant (the $G$-action on $\Sigma$ extends to a $G$-action on $\cP(\Sigma)$), and that we have the following composition:
\begin{eqnarray*}
\Sigma_n \stackrel{\rho_n}{\too} \cP(\Sigma)\stackrel{\cP(\sigma_n)}{\too} \cP(\Sigma) && s \mapsto \{s\}\,.
\end{eqnarray*}
\end{notation}

\begin{definition}\label{def:BC-sys}
Let $(\Sigma,\sigma_n)$ be an abstract Bost-Connes datum. The associated {\em abstract Bost-Connes system} consists of the following data:
\begin{itemize}
\item[(i)] The $k$-algebra $\overline{k}[\Sigma]^G$.
\item[(ii)] The $k$-algebra homomorphisms $
\underline{\sigma_n}: \overline{k}[\Sigma]^G \to \overline{k}[\Sigma]^G, s \mapsto \sigma_n(s)$.
\item[(iii)] The $k$-linear additive maps (defined only when $\alpha(n)$ is finite)
\begin{eqnarray*}
\underline{\rho_n}: \overline{k}[\Sigma_n]^G \to \overline{k}[\Sigma]^G && s \mapsto \sum_{s' \in \rho_n(s)}s'\,.
\end{eqnarray*}
\end{itemize}
\end{definition}

In \S\ref{sec:categorification} we will categorify the abstract Bost-Connes systems.
\begin{definition}\label{def:BC-sys2}
Let $(\Sigma, \sigma_n)$ be a concrete Bost-Connes system. The associated {\em concrete Bost-Connes system} consists of the following data:
\begin{itemize}
\item[(i)] The $k$-algebra $k[\Sigma]$ equipped with the induced $G$-action.
\item[(ii)] The $G$-equivariant $k$-algebra homomorphisms $\underline{\sigma_n}: k[\Sigma]\to
k[\Sigma]$, $s\mapsto \sigma_n(s)$.
\item[(iii)] The $G$-equivariant $k$-linear additive maps 
\begin{eqnarray*}
\underline{\rho_n}: k[\Sigma_n]\to k[\Sigma] && s\mapsto \sum_{s'\in \rho_n(s)} s'\,.
\end{eqnarray*}
\end{itemize}
\end{definition}
In \S\ref{sec:QSM} we will associate to every concrete Bost-Connes system a quantum statistical mechanical system.
\begin{remark}
Let $(\Sigma,\sigma_n)$ be a concrete Bost-Connes datum with trivial $G$-action. Since $\overline{k}[\Sigma]^G=k[\Sigma]$, the associated abstract and concrete Bost-Connes systems~agree.
\end{remark}

We now describe several examples of abstract/concrete Bost-Connes systems.
\subsection*{Example 1 - Original Bost-Connes system}
Let $k:=\bbQ$, $\Sigma:=\bbQ/\bbZ$ equipped with the $G$-action induced by the identification of $\bbQ/\bbZ$ with the roots of unit in $\overline{\bbQ}^\times$, and $\sigma_n$ the homomorphism $n\cdot -: \bbQ/\bbZ \to \bbQ/\bbZ$. This defines a concrete Bost-Connes datum. The associated concrete Bost-Connes system agrees with the one introduced originally by Bost and Connes in \cite{BC}; consult \cite{CM} for its reformulation. If we forget about the $G$-action, then the associated abstract=concrete Bost-Connes system is the arithmetic subalgebra of the Bost-Connes algebra; see \cite{CM}. More generally, $k=\bbQ$ can be replaced by any subfield of $\overline{\bbQ}$. 

\subsection*{Example 2 - Weil restriction}
Let $k:=\bbR$, $\Sigma:= \bbQ/\bbZ \times \bbQ/\bbZ$ equipped with the switch $\bbZ/2$-action, and $\sigma_n$ the homomorphism $(n\cdot -, n \cdot -): \bbQ/\bbZ \times \bbQ/\bbZ\to \bbQ/\bbZ \times \bbQ/\bbZ$. In this case, $\alpha(n)=n^2$. This defines a concrete Bost-Connes datum. The associated abstract Bost-Connes system is morally speaking the Weil restriction along $\bbC/\bbR$ of the arithmetic subalgebra of the Bost-Connes algebra; see \S\ref{sec:categorification}.
\subsection*{Example 3 - Algebraic numbers}
Let $k:=\bbQ$, $\Sigma:= \overline{\bbQ}^\times$ equipped with the canonical $G$-action, and $\sigma_n$ the homomorphism $(-)^n: \overline{\bbQ}^\times \to \overline{\bbQ}^\times$. This defines a concrete Bost-Connes datum. The associated abstract/concrete Bost-Connes system contains the one of Example 1. More generally, $k=\bbQ$ can be replaced by any subfield of $\overline{\bbQ}$.
\subsection*{Example 4 - Weil numbers}
\begin{definition}\label{def:Weil}
Let $q=p^r$ be a prime power. An algebraic number $\pi$ is called a {\em Weil $q$-number of weight $m \in \bbZ$} if the following holds:
\begin{itemize}
\item[(i)] For every embedding $\varrho:\bbQ[\pi] \hookrightarrow \bbC$ we have $|\varrho(\pi)| = q^{\frac{m}{2}}$.
\item[(ii)] There exists an integer $s$ such that $q^s \pi$ is an algebraic integer.
\end{itemize}
Let $\cW_m(q)$ be the set of Weil $q$-number of weight $m$ and $\cW(q):=\bigcup_{m \in \bbZ} \cW_m(q)$.
\end{definition}
Note that $\cW_0(q)$ and $\cW(q)$ are subgroups of $\overline{\bbQ}^\times$ and that $\cW_m(q)$ is stable under the canonical $\mathrm{Gal}(\overline{\bbQ}/\bbQ)$-action.

\begin{proposition}\label{lem:aux1}
The following holds:
\begin{itemize}
\item[(i)] The group homomorphism $(-)^n: \cW_0(q) \to \cW_0(q)$ is surjective.
\item[(ii)] We have the following group isomorphism 
\begin{eqnarray}\label{eq:isomorphism-decomp}
\cW(q) \stackrel{\sim}{\too} \cW_0(q) \times \bbZ && \pi \mapsto (\frac{\pi}{|\varrho(\pi)|},w(\pi))\,,
\end{eqnarray}
where $w(\pi)$ stands for the weight of $\pi$.
\item[(iii)] Under \eqref{eq:isomorphism-decomp}, the $\mathrm{Gal}(\overline{\bbQ}/\bbQ)$-action on $\cW(q)$ identifies with
\begin{equation*}
\gamma((\pi,m)) = \left\{ 
  \begin{array}{ll}
    (\gamma(\pi),m) & \text{if} \quad \gamma \in \mathrm{Ker}(\mathrm{Gal}(\overline{\bbQ}/\bbQ) \twoheadrightarrow \mathrm{Gal}(\bbQ[\sqrt{q}]/\bbQ))\\
   ((-1)^m \gamma(\pi),m) & \text{otherwise}\,.
  \end{array} \right.
\end{equation*}
\item[(iv)] Under \eqref{eq:isomorphism-decomp}, the homomorphism $(-)^n:\cW(q) \to \cW(q)$ identifies with 
\begin{eqnarray*}
\cW_0(q) \times \bbZ \too \cW_0(q) \times \bbZ && (\pi, m) \mapsto (\pi^n, nm)\,.
\end{eqnarray*}
\end{itemize}
\end{proposition}

\begin{remark}
Note that when $q$ is a even power of $p$, we have $\gamma((\pi,m))=(\gamma(\pi),m)$.
\end{remark}

\begin{proof}
(i) Given $\pi \in \cW_0(q)$, we need to show that $\sqrt[n]{\pi}$ also belongs to $\cW_0(q)$. Condition (i) of Definition \ref{def:Weil} is clear. Condition (ii) follows from the equality $q^s \sqrt[n]{\pi} = q^n \cdot q^{\frac{s}{n}}\sqrt[n]{\pi}$ and from the fact that $q^n$ and $q^{\frac{s}{n}}\sqrt[n]{\pi}$ are algebraic integers.

(ii) The inverse of \eqref{eq:isomorphism-decomp} is given by the group homomorphism
\begin{eqnarray*}
\cW_0(q) \times \bbZ \too \cW(q) && (\pi,m) \mapsto \pi q^{\frac{m}{2}}\,.
\end{eqnarray*}

(iii) Let $(\pi,m) \in \cW_0(q) \times \bbZ$. The action of $\gamma \in \mathrm{Gal}(\overline{\bbQ}/\bbQ)$ on $\pi q^{\frac{m}{2}}$ is given by $\gamma(\pi q^{\frac{m}{2}})=\gamma(\pi)\gamma(q^{\frac{m}{2}}) \in \cW(q)$. When $\gamma$ belongs to the kernel of the homomorphism $\mathrm{Gal}(\overline{\bbQ}/\bbQ) \twoheadrightarrow \mathrm{Gal}(\bbQ[\sqrt{q}]/\bbQ)$, we have $\gamma(q^{\frac{m}{2}})=q^{\frac{m}{2}}$. Otherwise, we have $\gamma(q^{\frac{m}{2}})=(-1)^m q^{\frac{m}{2}}$. The proof follows now from the above isomorphism \eqref{eq:isomorphism-decomp}.

(iv) The proof is by now clear.
\end{proof}
\begin{example}\label{ex:4}
Let $k:=\bbQ$, $\Sigma:= \cW_0(q)$ equipped with the canonical $G$-action, and $\sigma_n$ the homomorphism $(-)^n:\cW_0(q) \to \cW_0(q)$. Thanks to Lemma \ref{lem:aux1}(i), we have $\Sigma_n =\cW_0(q)$. This defines a concrete Bost-Connes datum. The associated abstract/concrete Bost-Connes system is contained in the one of Example 3 
and contains the one of Example 1. More generally, $k=\bbQ$ can be replaced by any subfield of $\overline{\bbQ}$.
\end{example}
\begin{example}\label{ex:5}
Let $k:=\bbQ$, $\Sigma:= \cW(q)$ equipped with the canonical $G$-action, and $\sigma_n$ the homomorphism $(-)^n: \cW(q) \to \cW(q)$. Thanks to Lemma \ref{lem:aux1}(iv), we have $\Sigma_n= \bigcup_{m \in \bbZ} \cW_{nm}(q)$. This defines a concrete Bost-Connes datum. The associated abstract/concrete Bost-Connes system is contained in the one of Example $3$ and contains the one of Example \ref{ex:4}. More generally, $k=\bbQ$ can be replaced by any subfield of $\overline{\bbQ}$.
\end{example}
\subsection*{Example 5 - CM fields}
Let $L \subset \overline{\bbQ}$ be a CM field which is Galois over $\bbQ$. We denote by $\mathfrak{P}$ its set of places and by $\|\text{-}\|_{\mathfrak{p}}, \mathfrak{p} \in \mathfrak{P}$, the normalized valuations.
\begin{definition}\label{def:CM}
Let $\cW_m^L(q)$ be the subset of those Weil $q$-numbers $\pi$ of weight $m$ such that $\pi \in L$ and $ \| \pi\|_{\mathfrak{p}} \in q^\bbZ$ for every $\mathfrak{p} \in \mathfrak{P}$.
\end{definition}
Note that $\cW_0^L(q)$ and $\cW^L(q):= \bigcup_{m \in \bbZ} \cW_m^L(q)$ are subgroups of $\cW_0(q)$ and $\cW(q)$, respectively, and that $\cW_m^L(q)$ is stable under the canonical $\mathrm{Gal}(\overline{\bbQ}/\bbQ)$-action.
\begin{remark}\label{rk:square}
When $\sqrt{q} \in L$, items (ii)-(iv) of Proposition \ref{lem:aux1} hold {\em mutatis mutandis} with $\cW(q)$ and $\cW_0(q)$ replaced by $\cW^L(q)$ and $\cW_0^L(q)$, respectively.
\end{remark}
\begin{example}\label{ex:6}
Let $k:=\bbQ$, $\Sigma:= \cW^L_0(q)$ equipped with the canonical $G$-action, and $\sigma_n$ the homomorphism $(-)^n:\cW^L_0(q) \to \cW^L_0(q)$. This defines an abstract Bost-Connes datum which is not concrete! Since the field extension $L/\bbQ$ is finite, it contains solely a finite number of roots of unit. Therefore, there exists a natural number $N \gg 0$ such that $\alpha(N)=\alpha(N^2)$, which implies that the assignment $\alpha: \bbN \to \bbN$ is not a semi-group homomorphism. The associated abstract Bost-Connes system is contained in the one of 
Example \ref{ex:4}. More generally, $k$ can be replaced by any subfield of $\overline{\bbQ}$.
\end{example}
\begin{example}\label{ex:7}
Let $k:=\bbQ$, $\Sigma:= \cW^L(q)$ equipped with the canonical $G$-action, and $\sigma_n$ the homomorphism $(-)^n: \cW^L(q) \to \cW^L(q)$. Similarly to Example \ref{ex:6}, this defines an abstract Bost-Connes datum which is not concrete. The associated abstract Bost-Connes system is contained in the one of Example \ref{ex:5} and contains the one of Example \ref{ex:6}. More generally, $k$ can be replaced by any subfield of $\overline{\bbQ}$.
\end{example}

\subsection*{Example 6 - Germs}
When $r |r'$ we have the group homomorphism
\begin{eqnarray}\label{eq:transaction}
\cW(p^r) \too \cW(p^{r'}) && \pi \mapsto \pi^{\frac{r'}{r}}\,. 
\end{eqnarray}
\begin{definition}
Let $\cW(p^\infty)$ be the colimit $\mathrm{colim}_r \cW(p^r)$. Note that $\cW(p^\infty)$ comes equipped with an induced $\mathrm{Gal}(\overline{\bbQ}/\bbQ)$-action.
\end{definition}
\begin{proposition}\label{lem:germs}
The following holds:
\begin{itemize}
\item[(i)] The above homomorphism \eqref{eq:transaction} sends $\cW_m(p^r)$ to $\cW_m(p^{r'})$. Consequently, we obtain the abelian group $\cW_0(p^\infty):= \mathrm{colim}_r \cW_0(p^r)$.
\item[(ii)] Given a CM field $L \subset \overline{\bbQ}$ which is Galois over $\bbQ$, the above homomorphism \eqref{eq:transaction} sends $\cW_m^L(p^r)$ to $\cW_m^L(p^{r'})$. Consequently, the obtain the abelian groups $\cW_0^L(p^\infty):= \mathrm{colim}_r \cW_0^L(p^r)$ and $\cW^L(p^\infty):= \mathrm{colim}_r \cW^L(p^r)$.
\item[(iii)] The group homomorphism $(-)^n: \cW(p^\infty) \to \cW(p^\infty)$ is surjective.
\item[(iv)] The group homomorphism $(-)^n$ is also injective.
\end{itemize}
\end{proposition}
\begin{proof}
(i) Let $\pi$ be a Weil $p^r$-number of weight $m$. Making use of the equalities
$$ | \varrho(\pi^{\frac{r'}{r}})|= |\varrho(\pi)|^{\frac{r'}{r}} = (p^{\frac{rm}{2}})^{\frac{r'}{r}}= p^{\frac{r'm}{2}}\,,$$
we conclude that $\pi^{\frac{r'}{r}}$ satisfies condition (i) of Definition \ref{def:Weil}(i) (with $q=p^{r'}$). Condition (ii) follows from the equality $(p^{r'})^s \pi^{\frac{r'}{r}} = (\sqrt[r]{(p^r)^s\pi})^{r'}$
and from the fact that $(p^r)^s\pi$ is an algebraic integer.

(ii) The proof follows automatically from the equality $\| \pi^{\frac{r'}{r}}\|_{\mathfrak{p}}=\|\pi\|_{\mathfrak{p}}^{\frac{r'}{r}}$.

(iii) Every element of $\cW(p^\infty)$ can be represented by a pair $(\pi,r)$ with $\pi \in \cW(p^r)$ and $r \in \bbN$, and two pairs $(\pi,r), (\pi',r')$ represent the same element of $\cW(p^\infty)$ if and only if $\pi^{r' N} = \pi'^{rN}$ for some $N \in \bbN$. Therefore, $(\sqrt[n]{\pi},r)^n=(\pi,r)$. 

(iv) Note that $(\pi,r)^n=(1,1)$ if and only if $\pi^N=1$ for some $N \in \bbN$. Given another pair $(\pi',r')$ such that $\pi'^{N'}=1$ for some $N' \in \bbN$, the equality $\pi^{r'NN'}= \pi'^{r N N'} =1$ allows us to conclude that $(\pi,r)=(\pi',r')$. This achieves the proof. 
\end{proof}
\begin{example}\label{ex:0L}
Let $k:=\bbQ$, $\Sigma := \cW_0(p^\infty)$ equipped with the induced $G$-action, and $\sigma_n$ the homomorphism $(-)^n: \cW_0(p^\infty) \to \cW_0(p^\infty)$. Thanks to Lemma \ref{lem:germs}, we have $\Sigma_n= \cW_0(p^\infty)$ and $\alpha(n)=1$ for every $n \in \bbN$. Therefore, this defines an abstract Bost-Connes datum and consequently an abstract Bost-Connes system. It is not a concrete Bost--Connes system, because $\alpha(n)=1$ is the trivial semi-group homomorphism.
More generally, $k=\bbQ$ can be replaced by any subfield of $\overline{\bbQ}$.
\end{example}
\begin{example}\label{ex:0L2}
Let $k:=\bbQ$, $\Sigma := \cW(p^\infty)$ equipped with the induced $G$-action, and $\sigma_n$ the homomorphism $(-)^n: \cW(p^\infty) \to \cW(p^\infty)$. Thanks to Lemma \ref{lem:germs}, we have $\Sigma_n= \cW(p^\infty)$ and $\alpha(n)=1$ for every $n \in \bbN$. Therefore, this defines an abstract Bost-Connes datum that is not concrete. The associated abstract 
Bost--Connes system contains the one of Example \ref{ex:0L}. More generally, $k=\bbQ$ can be replaced by any subfield of $\overline{\bbQ}$.
\end{example}
\begin{example}\label{ex:0L1}
Let $k:=\bbQ$, $\Sigma:= \cW_0^L(p^\infty)$ equipped with the induced $G$-action, and $\sigma_n$ the homomorphism $(-)^n: \cW_0^L(p^\infty) \to \cW_0^L(p^\infty)$. Thanks to Lemma \ref{lem:germs}, we have $\alpha(n)=1$ for every $n \in \bbN$. Therefore, this defines an abstract Bost-Connes datum that is not concrete. The associated abstract 
Bost--Connes system is contained in the one of Example \ref{ex:0L}.
\end{example}
\begin{example}\label{ex:0L11}
Let $k:=\bbQ$, $\Sigma:= \cW^L(p^\infty)$ equipped with the induced $G$-action, and $\sigma_n$ the homomorphism $(-)^n: \cW^L(p^\infty) \to \cW^L(p^\infty)$. Similarly to Example \ref{ex:0L}, this defines an abstract Bost-Connes datum that is not concrete. The associated abstract Bost-Connes system is contained in the one of Example \ref{ex:0L2} 
and contains the one of Example \ref{ex:0L1}. 
\end{example}
\subsection*{Example 7 - Completion}
Let $\widehat{\cW}(q)$ be the limit of the following diagram:
$$\cdots \twoheadrightarrow \cW(q)/_{\!q^{n+1}=1} \twoheadrightarrow \cW(q)/_{\!q^{n}=1} \twoheadrightarrow \cdots \twoheadrightarrow \cW(q)/_{\!q^2=1} \twoheadrightarrow \cW(q)/_{\!q=1}\,.$$ 
Note that $\widehat{\cW}(q)$ comes equipped with an induced $\mathrm{Gal}(\overline{\bbQ}/\bbQ)$-action.
\begin{proposition}\label{lem:aux3}
\begin{itemize}
\item[(i)] We have a group isomorphism $\widehat{\cW}(q) \simeq \cW_0(q) \times \bbQ/2\bbZ$.
\item[(ii)] Under (i), the $\mathrm{Gal}(\overline{\bbQ}/\bbQ)$-action on $\widehat{\cW}(q)$ identifies with \begin{equation*}
\gamma ((\pi,m)) = \left\{  \begin{array}{ll}   (\gamma(\pi),m) & \text{if} \quad \gamma \in \mathrm{Ker}(\mathrm{Gal}(\overline{\bbQ}/\bbQ) \twoheadrightarrow \mathrm{Gal}(\bbQ[\sqrt{q}]/\bbQ))\\  ((-1)^{|m|} \gamma(\pi),m) & \text{otherwise}  \end{array} \right.\,,
\end{equation*}
where $|m|$ stands for the parity of $m \in \bbQ/2\bbZ$.
\item[(iii)] Under (i), the group homomorphism $(-)^n: \widehat{\cW}(q) \to \widehat{\cW}(q)$ identifies with 
\begin{equation*}
((-)^n, n \cdot -): \cW_0(q) \times \bbQ/2\bbZ \too \cW_0(q) \times \bbQ/2\bbZ \,.
\end{equation*}
\end{itemize}
\end{proposition}
\begin{proof}
Note that \eqref{eq:isomorphism-decomp} induces to an isomorphism between $\cW(q)/_{\!q^n=1}$ and the product $\cW_0(q) \times \bbZ/2n$. The proof of item (i) follows from the fact that $\bbQ/2\bbZ \simeq \mathrm{lim}_{n\geq 1} \bbZ/2n$. In what concerns the proof of item (ii), resp. item (iii), it follows from the combination of item (i) with Lemma \ref{lem:aux1}(ii), resp. Lemma \ref{lem:aux1}(iv).
\end{proof}
\begin{example}\label{ex:9}
Let $k:=\bbQ$, $\Sigma:= \widehat{\cW}(q)$ equipped with the induced $G$-action, and $\sigma_n$ the homomorphism $(-)^n:\widehat{\cW}(q) \to \widehat{\cW}(q)$. In this case, $\alpha(n)=n^2$. Thanks to Lemma \ref{lem:aux3}(iii), we have $\Sigma_n= \cW_0(q) \times n \bbQ/2\bbZ$. This defines a concrete Bost-Connes datum and consequently an abstract/concrete Bost-Connes system. More generally, $k=\bbQ$ can be replaced by any subfield of $\overline{\bbQ}$.
\end{example}
Let $L \subset \overline{\bbQ}$ be a CM-field which is Galois over $\bbQ$. As above, we can define the abelian group $\widehat{\cW}^L(q):= \mathrm{lim}_{n \geq 1} \cW^L(q)/_{\!q^n=1}$ equipped with the $\mathrm{Gal}(\overline{\bbQ}/\bbQ)$-action. 
\begin{remark}
When $\sqrt{q} \in L$, Remark \ref{rk:square} allow us to conclude that Proposition \ref{lem:aux3} holds {\em mutatis mutandis} with $\cW_0(q)$ and $\widehat{\cW}(q)$ replaced by $\cW_0^L(q)$ and $\widehat{\cW}^L(q)$.
\end{remark}
\begin{example}\label{ex:last}
Let $k:=\bbQ$, $\Sigma:= \widehat{\cW}^L(q)$ equipped with the induced $G$-action, and $\sigma_n$ the homomorphism $(-)^n:\widehat{\cW}^L(q) \to \widehat{\cW}^L(q)$. Similarly to Example \ref{ex:0L11}, this defined an abstract Bost-Connes datum which is not concrete. More generally, $k=\bbQ$ can be replaced by any subfield of $\overline{\bbQ}$.
\end{example}

\section{Categorification}\label{sec:categorification}
In this section we categorify the abstract Bost-Connes systems. Given an abstract Bost-Connes datum $(\Sigma, \sigma_n)$, we start by categorifying the $k$-algebra $\overline{k}[\Sigma]^G$. Note that $\overline{k}[\Sigma]^G$ becomes an Hopf $k$-algebra when we set $\Delta(s):= s\otimes s$, $\epsilon(s):=1$, and $\mathrm{inv}(s):= s^{-1}$ for every $s \in \Sigma$. As explained in \cite[XIV Thm.~5.3]{AGS}, the assignment $\Sigma \mapsto \mathrm{Spec}(\overline{k}[\Sigma]^G)$ gives rise to a contravariant equivalence between the category of abelian groups equipped with a continuos $G$-action and the category of affine group $k$-schemes of multiplicative type.
\begin{definition}\label{def:categorification-1}
Let $\mathrm{Vect}_\Sigma^{\overline{k}}(k)$ be the category of pairs $(V, \bigoplus_{s \in \Sigma} \overline{V}^s)$, where $V$ is a finite dimensional $k$-vector space and $\bigoplus_{s \in \Sigma} \overline{V}^s$ is a $\Sigma$-grading on $\overline{V}:= V \otimes \overline{k}$. We assume that $\overline{V}^{\gamma(s)}={}^\gamma \overline{V}^s$ for every $s \in \Sigma$ and $\gamma \in G$, where ${}^\gamma \overline{V}^s$ stands for the $\overline{k}$-vector space obtained from $\overline{V}^s$ by restriction of scalars along the automorphism $\gamma^{-1}: \overline{k} \stackrel{\sim}{\to} \overline{k}$. The morphisms are the $k$-linear maps $f: V \to V'$ such that $\overline{f}:\overline{V} \to \overline{V'}$ preserves the $\Sigma$-grading. The tensor product of $k$-vector spaces and of $\Sigma$-graded $\overline{k}$-vector spaces gives rise to a symmetric monoidal structure on $\mathrm{Vect}_\Sigma^{\overline{k}}(k)$.
\end{definition}
\begin{theorem}\label{thm:categorification-1}
The following holds:
\begin{itemize}
\item[(i)] The $k$-linear category $\mathrm{Vect}_\Sigma^{\overline{k}}(k)$ is neutral Tannakian. A fiber functor is given by the forgetful functor $\omega: \mathrm{Vect}_\Sigma^{\overline{k}}(k) \to \mathrm{Vect}(k)$.
\item[(ii)] The category $\mathrm{Vect}_\Sigma^{\overline{k}}(k)$ is semi-simple. Moreover, the isomorphism classes of simple objects are in one-to-one correspondence with the $G$-orbits of $\Sigma$.
\item[(iii)] There is an isomorphism of affine group $k$-schemes $\mathrm{Aut}^\otimes(\omega) \simeq \mathrm{Spec}(\overline{k}[\Sigma]^G)$.
\end{itemize}
\end{theorem}
\begin{proof}
Since $\mathrm{Spec}(k[\Sigma]^G)$ is an affine group $k$-scheme of multiplicative type, its group of characters identifies with $\Sigma$ equipped with the $G$-action; see \cite[Thm.~5.3]{AGS}. Consequently, items (i) and (iii) follow from \cite[\S2.32]{Deligne-Milne}. In what concerns item (ii), it follows from \cite[Thm.~5.10]{AGS}. 
\end{proof}
\begin{remark}\label{rk:simple}
When the $G$-action $G\to \mathrm{Aut}(\Sigma)$ is trivial, we have $\overline{k}[\Sigma]^G=k[\Sigma]$ and $\mathrm{Vect}_\Sigma^{\overline{k}}(k)$ reduces to the category $\mathrm{Vect}_\Sigma(k)$ of finite dimensional $\Sigma$-graded $k$-vector spaces. In this case, the simple objects $S_{s'}, s' \in \Sigma$, are given by $S_{s'}:=(k,\bigoplus_{s\in \Sigma} k^s)$ where $k^s=k$ when $s=s'$. 
\end{remark}
\begin{definition}\label{def:smaller}
Let $\mathrm{Vect}^k_\Sigma(k)$ be the full subcategory of $\mathrm{Vect}^{\overline{k}}_\Sigma(k)$ consisting of the pairs $(V,\bigoplus_{s \in \Sigma} \overline{V^s})$, where $\bigoplus_{s \in \Sigma} V^s$ is a $\Sigma$-grading of $V$ such that $V^{\gamma(s)}=V^s$ for every $s \in \Sigma$ and $\gamma \in G$. Note that $\mathrm{Vect}^k_\Sigma(k)$ is a neutral Tannakian subcategory of $\mathrm{Vect}_\Sigma^{\overline{k}}(k)$. A fiber functor is given by the forgetful functor $\omega: \mathrm{Vect}^k_\Sigma(k) \to \mathrm{Vect}(k)$.
\end{definition}

\begin{proposition}\label{prop:smaller}
We have an isomorphism $\mathrm{Aut}^\otimes(\omega) \simeq \mathrm{Spec}(k[\Sigma]^G)$.
\end{proposition}
\begin{proof}
Note first that $\mathrm{Vect}^k_\Sigma(k)$ naturally identifies with the neutral Tannakian category $\mathrm{Vect}_{\Sigma/G}(k)$. Therefore, making use of Theorem \ref{thm:categorification-1} and of Remark \ref{rk:simple}, we conclude that $\mathrm{Aut}^\otimes(\omega) \simeq \mathrm{Spec}(k[\Sigma/G])$. The proof follows now from the canonical isomorphism of Hopf $k$-algebras $k[\Sigma]^G \stackrel{\sim}{\to} k[\Sigma/G]$.
\end{proof}
\begin{remark}
By combining Theorem \ref{thm:categorification-1} with Proposition \ref{prop:smaller}, we observe that the inclusion of Hopf $k$-algebras $k[\Sigma]^G \hookrightarrow \overline{k}[\Sigma]^G$, or equivalently the quotient of affine group $k$-schemes $\mathrm{Spec}(\overline{k}[\Sigma]^G) \twoheadrightarrow \mathrm{Spec}(k[\Sigma]^G)$, is induced by the inclusion of neutral Tannakian categories $\mathrm{Vect}^k_\Sigma(k) \subset \mathrm{Vect}^{\overline{k}}_\Sigma(k)$.
\end{remark}

Theorem \ref{thm:categorification-1} provides a Tannakian categorification of the $k$-algebra $\overline{k}[\Sigma]^G$, as well of its canonical Hopf structure. The Tannakian categorification of the Hopf $k$-algebra homomorphism $\underline{\sigma_n}: \overline{k}[\Sigma]^G \to \overline{k}[\Sigma]^G$ is provided by the following $k$-linear additive symmetric monoidal functor (note that the direct sum is finite):
\begin{eqnarray*}
\underline{\bf \sigma_n}: \mathrm{Vect}_\Sigma^{\overline{k}}(k) \to \mathrm{Vect}_\Sigma^{\overline{k}}(k) & {\underline{\bf \sigma_n}}(V):=V & \overline{\underline{\bf \sigma_n}(V)}^s:=\left\{ 
  \begin{array}{l l}
  \bigoplus_{s' \in \rho_n(s)} \overline{V}^{s'} & \text{if}\, s \in \Sigma_n \\
     0 & \text{otherwise}\,. \\
  \end{array} \right.
\end{eqnarray*}
\begin{theorem}\label{thm:categorification-2}
The following holds:
\begin{itemize}
\item[(i)] We have natural equalities $\underline{\bf \sigma_{nm}}= \underline{\bf \sigma_n} \circ \underline{\bf \sigma_m}$ and $\omega \circ \underline{\bf \sigma_n} = \omega$.
\item[(ii)] The morphism of Hopf $k$-algebras $\underline{\bf \sigma_n}: \overline{k}[\Sigma]^G \to \overline{k}[\Sigma]^G$, corresponding to $\underline{\bf \sigma_n}: \mathrm{Spec}(\overline{k}[\Sigma]^G) \to \mathrm{Spec}(\overline{k}[\Sigma]^G)$, agrees with $\underline{\sigma_n}$.  
\end{itemize}
\end{theorem}
\begin{proof}
(i) Since $\underline{\bf \sigma_n}(V):=V$, equality $\omega \circ \underline{\bf \sigma_n}=\omega$ is clear. In what concerns $\underline{\bf \sigma_{nm}}= \underline{\bf \sigma_n} \circ \underline{\bf \sigma_m}$, it follows from the assumption $\sigma_{nm}=\sigma_n \circ \sigma_m$; see Definition \ref{def:BC-data}.

(ii) By base-change along $\overline{k}/k$, it suffices to show that the following morphisms
\begin{eqnarray}\label{eq:morphisms-2}
\underline{\bf \sigma_n} \otimes \overline{k} : \overline{k}[\Sigma] \too \overline{k}[\Sigma] && \underline{\sigma_n}: \overline{k}[\Sigma] \stackrel{s \mapsto \sigma_n(s)}{\too} \overline{k}[\Sigma]
\end{eqnarray}
agree. Thanks to Lemma \ref{lem:key-difficult} below, $\underline{\bf \sigma_n}\otimes \overline{k}$ can be replaced by the morphism induced by the functor $\underline{\bf \sigma_n}: \mathrm{Vect}_\Sigma(\overline{k})\to \mathrm{Vect}_\Sigma(\overline{k})$. Given an arbitrary $\overline{k}$-algebra $A$, let us then describe the induced group homomorphisms
\begin{eqnarray*}
\underline{\bf \sigma_n}(R)^\ast: \mathrm{Aut}^\otimes(\omega)(A) = \Hom(\Sigma, A^\times) &\too & \Hom(\Sigma, A^\times)=\mathrm{Aut}^\otimes(\omega)(A)\\
\underline{\sigma_n}^\ast: \Hom(\overline{k}[\Sigma],A)= \Hom(\Sigma, A^\times) & \too & \Hom(\Sigma, A^\times) = \Hom(\overline{k}[\Sigma],A)\,.
\end{eqnarray*}
As explained in Remark \ref{rk:simple}, there is a one-to-one correspondence $s \mapsto S_s$ between elements of $\Sigma$ and simple objects of $\mathrm{Vect}_\Sigma(\overline{k})$. Making use of the equality $\underline{\bf \sigma_n}(S_s)=S_{\sigma_n(s)}$, we hence conclude that the homomorphism $\underline{\bf \sigma_n}(R)^\ast$ is given by pre-composition with $\sigma_n: \Sigma \to \Sigma$. The group homomorphism $\underline{\sigma_n}^\ast$ is also given by pre-composition with $\sigma_n$. Since the $k$-algebra $A$ is arbitrary, this implies that the above morphisms \eqref{eq:morphisms-2} agree, and so the proof is finished.
\end{proof}
\begin{lemma}\label{lem:key-difficult}
The above morphism $\underline{\bf \sigma_n}\otimes \overline{k}$ agrees with the one induced by the functor $\underline{\bf \sigma_n}: \mathrm{Vect}_\Sigma(\overline{k}) \to \mathrm{Vect}_\Sigma(\overline{k})$.
\end{lemma}
\begin{proof}
Consider the following commutative diagram
$$
\xymatrix{
\mathrm{Vect}_\Sigma(\overline{k}) \ar[r]^-{\underline{\bf \sigma_n}} & \mathrm{Vect}_\Sigma(\overline{k}) \ar[r]^-\omega & \mathrm{Vect}(\overline{k}) \\
\mathrm{Vect}_\Sigma^{\overline{k}}(k) \ar[r]_-{\underline{\bf \sigma_n}} \ar[u]^-\Psi & \mathrm{Vect}_\Sigma^{\overline{k}}(k) \ar[r]_-\omega \ar[u]^-\Psi & \mathrm{Vect}(k) \ar[u]_-{-\otimes \overline{k}}\,,
}
$$
$\Psi(V,\bigoplus_{s \in \Sigma} \overline{V}^s) := (\overline{V},\bigoplus_{s \in \Sigma} \overline{V}^s)$. On one hand, $\mathrm{Aut}^\otimes(\omega \otimes \overline{k}) \simeq \mathrm{Aut}^\otimes(\omega)_{\overline{k}} \simeq \mathrm{Spec}(\overline{k}[\Sigma]^G)_{\overline{k}}$. On the other hand, since the affine group $k$-scheme $\mathrm{Spec}(\overline{k}[\Sigma]^G)$ is of multiplicative type, $\Psi$ induces an isomorphism $\mathrm{Spec}(\overline{k}[\Sigma])\simeq \mathrm{Spec}(\overline{k}[\Sigma]^G)_{\overline{k}}$. Consequently, we obtain the following commutative square
$$
\xymatrix{
\mathrm{Spec}(\overline{k}[\Sigma])\ar[d]_-\sim \ar[rr]^-{\underline{\bf \sigma_n}} && \mathrm{Spec}(\overline{k}[\Sigma]) \ar[d]^-\sim \\
\mathrm{Spec}(\overline{k}[\Sigma]^G)_{\overline{k}} \ar[rr]_-{\underline{\bf \sigma_n}\otimes \overline{k}} && \mathrm{Spec}(\overline{k}[\Sigma]^G)_{\overline{k}}\,.
}
$$
This achieves the proof.
\end{proof}
\begin{remark}
Since the functor $\underline{\bf \sigma_n}$ restricts to the category $\mathrm{Vect}_\Sigma^k(k)$, the Tannakian categorification of the Hopf $k$-algebra homomorphism $\underline{\sigma_n}: k[\Sigma]^G \to k[\Sigma]^G$ is also provided by $\underline{\bf \sigma_n}$.
\end{remark}
The $k$-linear additive map $\underline{\rho_n}: \overline{k}[\Sigma_n]^G \to \overline{k}[\Sigma]^G$ does {\em not} preserve the algebra structure. Consequently, it does not admit a Tannakian categorification. Nevertheless, we have the following $k$-linear additive functor (defined when $\alpha(n)$ is finite):
\begin{eqnarray*}
\underline{\bf \rho_n}: \mathrm{Vect}_{\Sigma_n}^{\overline{k}}(k) \too \mathrm{Vect}_\Sigma^{\overline{k}}(k) &\underline{\bf \rho_n}(V):=\bigoplus_{i=1}^{\alpha(n)} V & \overline{\underline{\bf \rho_n}(V)}^s:= \overline{V}^{\sigma_n(s)}\,.
\end{eqnarray*}
Recall that $\underline{\sigma_n}\circ \underline{\rho_n} = \alpha(n)\cdot \Id$.  The following result categorifies this equality:
\begin{proposition}
We have a natural equality $\underline{\bf \sigma_n} \circ \underline{\bf \rho_n} = \Id^{\oplus \alpha(n)}$.
\end{proposition}
\begin{proof}
The proof follows automatically from the definition of $\underline{\bf \sigma_n}$ and $\underline{\bf \rho_n}$.
\end{proof}
\begin{remark}
Similarly to $\underline{\bf \sigma_n}$, the functor $\underline{\bf \rho_n}$ restricts to the category $\mathrm{Vect}_\Sigma^k(k)$.
\end{remark}
\subsection*{Automorphisms, Frobenius, and Verschiebung}
Let $(\Sigma, \sigma_n)$ be an abstract Bost-Connes datum. Assume that we have a $G$-equivariant embedding $\Sigma \hookrightarrow \overline{k}^\times$ and that $\sigma_n: \Sigma \to \Sigma$ is given by $s \mapsto s^n$. These conditions are verified by Examples 1 and 3-6 in \S\ref{sec:BC-systems}.
\begin{definition}
Let $\mathrm{Aut}_\Sigma^{\overline{k}}(k)$ be the category of pairs $(V, \Phi)$, where $V$ is a finite dimensional $k$-vector space and $\Phi: \overline{V} \stackrel{\sim}{\to} \overline{V}$ a $G$-equivariant diagonalizable automorphism whose eigenvalues belong to $\Sigma$. The morphisms $(V,\Phi) \to (V',\Phi')$ are the $k$-linear homomorphisms $f: V \to V'$ such that $\overline{f} \circ \Phi = \Phi' \circ \overline{f}$. The tensor product of vector spaces gives rise to a symmetric monoidal structure on $\mathrm{Aut}_\Sigma^{\overline{k}}(k)$. 
\end{definition}
\begin{remark}
When The $G$-action $G \to \mathrm{Aut}(\Sigma)$ is trivial, the embedding $\Sigma \hookrightarrow \overline{k}^\times$ factors through $k^\times \subset \overline{k}^\times$. Consequently, $\mathrm{Aut}_\Sigma^{\overline{k}}(k)$ reduces to the category $\mathrm{Aut}_\Sigma(k)$ of pairs $(V,\Phi)$, where $V$ is a finite dimensional $k$-vector space and $\Phi: V \stackrel{\sim}{\to} V$ a diagonalizable automorphism whose eigenvalues belong to $\Sigma$.
\end{remark}
Note that $\mathrm{Aut}_\Sigma^{\overline{k}}(k)$ comes equipped with the forgetful functor $(V,\Phi) \mapsto V$ to finite dimensional $k$-vector spaces. We have also a ``Frobenius'' functor
\begin{eqnarray}\label{eq:Frobenius}
\mathrm{Aut}_\Sigma^{\overline{k}}(k) \too \mathrm{Aut}_\Sigma^{\overline{k}}(k) && (V, \Phi) \mapsto (V, \Phi^n)
\end{eqnarray}
and a ``Verschiebung'' functor
\begin{eqnarray}\label{eq:Verschiebung}
\mathrm{Aut}_{\Sigma_n}^{\overline{k}}(k) \too \mathrm{Aut}_\Sigma^{\overline{k}}(k) && (V, \Phi) \mapsto (V^{\oplus n}, \bbV_n(\Phi))\,,
\end{eqnarray}
where 
\begin{equation}\label{eq:matrix-V}
\bbV_n(\Phi):= \left( \begin{array}{ccccc} 0 & \cdots & \cdots &  0 & \Phi \\
\Id & \ddots & \ddots & \ddots & 0 \\
0 & \ddots & \ddots & \ddots &  \vdots \\
\vdots & \ddots & \ddots & \ddots & \vdots \\
0 & \cdots & 0 & \Id & 0 \end{array} \right)\,.
\end{equation}
\begin{theorem}\label{thm:equivalence}
The following holds:
\begin{itemize}
\item[(i)] There is an equivalence of neutral Tannakian categories $\mathrm{Aut}_\Sigma^{\overline{k}}(k)$ and $\mathrm{Vect}_\Sigma^{\overline{k}}(k)$.
\item[(ii)] Under equivalence (i), the above functor \eqref{eq:Frobenius} corresponds to $\underline{\bf \sigma_n}$.
\item[(iii)] Assume that the $G$-action $G \to \mathrm{Aut}(\Sigma)$ is trivial (\eg\ $k=\overline{k}$). Under equivalence (i), the functor \eqref{eq:Verschiebung} corresponds to $\underline{\bf \rho_n}$.
\end{itemize}
\end{theorem}
\begin{proof}
(i) Consider the following additive symmetric monoidal functor
\begin{eqnarray*}
\psi: \mathrm{Aut}^{\overline{k}}_\Sigma(k) \too \mathrm{Vect}^{\overline{k}}_\Sigma(k) && (V, \Phi) \mapsto (V, \bigoplus_{s \in \Sigma}\overline{V}^s)\,,
\end{eqnarray*}
where $\overline{V}^s$ stands for the eigenspace of $\Phi$ associated to the eigenvector $s \in \Sigma$. Note that since $\Phi$ is $G$-equivariant, we have $\overline{V}^{\gamma(s)}= {}^\gamma \overline{V}^s$ for every $s \in \Sigma$ and $\gamma \in G$. The (quasi-)inverse of $\psi$ is given by the following additive symmetric monoidal functor
\begin{eqnarray*}
\varphi: \mathrm{Vect}_\Sigma^{\overline{k}}(k) \too \mathrm{Aut}^{\overline{k}}_\Sigma(k)&& (V, \bigoplus_{s \in \Sigma} \overline{V}^s) \mapsto (V, \Phi)\,,
\end{eqnarray*}
where $\Phi$ stands for diagonal automorphism whose restriction to $\overline{V}^s$ is given by multiplication by $s$. The proof follows now from the fact that $\psi$ and $\varphi$ are compatible with the forgetful functors to finite dimensional $k$-vector spaces.

(ii) Given an object $(V,\Phi)$ of $\mathrm{Aut}_\Sigma^{\overline{k}}(k)$, we claim that $\underline{\bf \sigma_n}(\psi(V,\Phi)) \simeq \psi(V, \Phi^n)$. Note that the eigenspace of $\Phi^n$ associated to $s \in \Sigma$ can be expressed as a direct sum, indexed by the elements $s'$ such that $(s')^n=s$, of the different eigenspaces of $\Phi$ associated to the elements of $s'$. Therefore, since $\sigma_n(s')=(s')^n$, our claim follows from the definition of $\underline{\bf \sigma_n}$. This achieves the proof.

(iii) Given an object $(V, \bigoplus_{s \in \Sigma}V^s)$ of $\mathrm{Vect}_\Sigma(k)$, note that 
\begin{eqnarray}\label{eq:111}
&& \varphi(\underline{\bf \rho_n}(V, \bigoplus_{s \in \Sigma}V^s)) =\left(V^{\oplus n}, \bigoplus_{s \in \Sigma}\bigoplus_{s' \in \rho_n(s)} V^s \stackrel{\oplus_s \oplus_{s'} (s'\cdot -)}{\too} \bigoplus_{s \in \Sigma}\bigoplus_{s' \in \rho_n(s)} V^s\right)\,.
\end{eqnarray}
Note also that we have the following equality
\begin{eqnarray}\label{eq:222}
\varphi(V, \bigoplus_{s \in \Sigma}V^s) = \left(V, \bigoplus_{s \in \Sigma} V^s  \stackrel{\oplus_s(s\cdot -)}{\too} \bigoplus_{s \in \Sigma} V^s\right)\,.
\end{eqnarray}
Thanks to Lemma \ref{lem:key-last} below, we observe that the ``Verschiebung'' of the pair \eqref{eq:222} is isomorphic in $\mathrm{Aut}_\Sigma(k)$ to the pair \eqref{eq:111}. This achieves the proof.
\end{proof}
\begin{lemma}\label{lem:key-last}
The following holds:
\begin{itemize}
\item[(i)] Given $k$-linear automorphisms $\Phi$ and $\Phi'$, we have a canonical isomorphism between $\bbV_n(\Phi) \oplus \bbV_n(\Phi')$ and $\bbV_n(\Phi \oplus \Phi')$.
\item[(ii)] Given a $k$-linear automorphism $V \stackrel{s\cdot -}{\to} V, s \in k^\times$, we have an isomorphism between $\bbV_n(s \cdot -)$ and the automorphism $\bigoplus_{s' \in \rho_n(s)} V \stackrel{\oplus_{s'}(s' \cdot -)}{\to} \bigoplus_{s' \in \rho_n(s)} V$. 
\end{itemize}
\end{lemma}
\begin{proof}
(i) Let $V$, resp. $V'$, be the source (= target) of $\Phi$, resp. of $\Phi'$. The searched isomorphism is given by the permutation $V^{\oplus n} \oplus V'^{\oplus n} \simeq (V \oplus V')^{\oplus n}$.

(ii) The characteristic polynomial $p(\lambda)$ of $\bbV(s\cdot -)$ is given by $\lambda^n-s=0$. Therefore, the proof follows from the diagonalization of this latter automorphism.
\end{proof}

\subsection*{Relation with Motives}
In this subsection we relate the categorification of some of our examples of abstract Bost-Connes systems with the theory of motives. 

Example 1 makes use of the notion of orbit categories. Let $(\cC,\otimes, {\bf 1})$ be a $k$-linear, additive, rigid symmetric monoidal category and $\cO \in \cC$ a $\otimes$-invertible object. The associated {\em orbit category} $\cC/_{\!\!-\otimes \cO}$ has the same objects as $\cC$ and morphisms
$$ \Hom_{\cC/_{\!\!-\otimes \cO}}(a,b):= \bigoplus_{i \in \bbZ} \Hom_\cC(a,b\otimes \cO^{\otimes i})\,.$$
Given objects $a, b, c$ and morphisms
\begin{eqnarray*}
\mathrm{f}=\{f_i\}_{i \in \bbZ} \in \bigoplus_{i \in \bbZ} \Hom_\cC(a,b\otimes \cO^{\otimes i}) && \mathrm{g}=\{g_i\}_{i \in \bbZ} \in \bigoplus_{i \in \bbZ} \Hom_\cC(b,c\otimes \cO^{\otimes i})
\end{eqnarray*}
the $i'^{\mathrm{th}}$-component of $\mathrm{g} \circ \mathrm{f}$ is given by $\sum_i((g_{i'-i}\otimes \cO^{\otimes i}) \circ f_i)$. The canonical functor 
\begin{eqnarray*}
\tau: \cC \to \cC/_{\!\!-\otimes \cO} & a \mapsto a & f \mapsto \mathrm{f}=\{f_i\}_{i \in \bbZ}\,,
\end{eqnarray*}
where $f_0=f$ and $f_i=0$ for $i \neq 0$, is endowed with an isomorphism $\tau \circ (-\otimes \cO) \stackrel{\sim}{\Rightarrow} \tau$ and is $2$-universal among all such functors; see \cite[\S7]{CvsNC}. By construction, $\cC/_{\!\!-\otimes \cO}$ is $k$-linear and additive. Moreover, as proved in \cite[Lem.~7.3]{CvsNC}, $\cC/_{\!\!-\otimes \cO}$ inherits from $\cC$ a rigid symmetric monoidal structure making the functor $\tau$ symmetric monoidal.
\subsection*{Example 1 - Original Bost-Connes system}
Let us first forget about the $G$-action. Thanks to Theorem \ref{thm:categorification-1}, we obtain the affine group $k$-scheme $\mathrm{Spec}(k[\bbQ/\bbZ])$ and the neutral Tannakian category $\mathrm{Vect}_{\bbQ/\bbZ}(k)$. Under the assignment $\Sigma \mapsto \mathrm{Spec}(k[\Sigma])$, $\bbZ/n$ corresponds to the affine group $k$-scheme $\mu_n$ of $n^{\mathrm{th}}$ roots of unit. Therefore, making use of $\bbQ/\bbZ \simeq \mathrm{lim}_{n\geq 1} \bbZ/n$, we conclude that 
$$\mathrm{Spec}(k[\bbQ/\bbZ]) \simeq \mathrm{colim}_{n \geq 1}\mathrm{Spec}(k[\bbZ/n])\simeq \mathrm{colim}_{n \geq 1}\mu_n \simeq \bbG_{m, \mathrm{tors}}\,.$$
Recall from \cite[\S4.1.5]{Andre} the construction of the category of Tate motives $\mathrm{Tate}(k)_k$ with $k$-coefficients. This is a neutral Tannakian category which comes equipped with $\otimes$-invertible objects $k(n), n \in \bbZ$.
\begin{proposition}\label{prop:lim}
There is an equivalence of neutral Tannakian categories
$$ \mathrm{Vect}_{\bbQ/\bbZ}(k) \simeq \mathrm{lim}_{n \geq 1} \mathrm{Tate}(k)_k/_{\!\!-\otimes k(n)}\,.$$
\end{proposition}
\begin{proof}
Thanks to the isomorphism $\bbQ/\bbZ \simeq \mathrm{lim}_{n\geq 1}  \bbZ/n$, we have an induced equivalence of neutral Tannakian categories $\mathrm{Vect}_{\bbQ/\bbZ}(k) \simeq \mathrm{lim}_{n\geq 1} \mathrm{Vect}_{\bbZ/n}(k)$. Hence, it is enough to construct an equivalence between $\mathrm{Tate}(k)_k/_{\!\!-\otimes k(n)}$ and $\mathrm{Vect}_{\bbZ/n}(k)$. The category $\mathrm{Tate}(k)_k$ identifies with $\mathrm{Vect}_\bbZ(k)$. Under such identification, the objects $k(n)$ correspond to the simple objects $S_n$; see Remark \ref{rk:simple}. Therefore, we obtain a symmetric monoidal equivalence between $\mathrm{Tate}(k)_k/_{\!\!-\otimes k(n)}$ and $\mathrm{Vect}_\bbZ(k)/_{\!\!-\otimes S_n}$. In order to conclude the proof, it suffices then to construct a symmetric monoidal equivalence between $\mathrm{Vect}_\bbZ(k)/_{\!\!-\otimes S_n}$ and $\mathrm{Vect}_{\bbZ/n}(k)$.

Let $\sigma:\bbZ \twoheadrightarrow \bbZ/n$ be the projection homomorphism and $\rho:\bbZ/n \to \cP(\bbZ)$ the map that sends an element of $\bbZ/n$ to its pre-image under $\sigma$. Under such notations, consider the following $k$-linear additive symmetric monoidal functor
\begin{eqnarray*}
\underline{\bf \sigma}: \mathrm{Vect}_\bbZ(k) \too \mathrm{Vect}_{\bbZ/n}(k) &\underline{\bf \sigma}(V):=V& \underline{\bf \sigma}(V)^s:= \bigoplus_{s' \in \rho_n(s)}V^{s'}\,.
\end{eqnarray*}
Clearly, $\underline{\bf \sigma}$ sends the simple object $S_n$ to the $\otimes$-unit of $\mathrm{Vect}_{\bbZ/n}(k)$. Consequently, by the universal property of orbit categories, we obtain an induced functor
\begin{equation}\label{eq:aux-2}
\mathrm{Vect}_\bbZ(k)/_{\!\!-\otimes S_n} \too \mathrm{Vect}_{\bbZ/n}(k)\,.
\end{equation}
Since $\underline{\bf \sigma}$ is $k$-linear, additive, symmetric monoidal, and essentially surjective, \eqref{eq:aux-2} is also $k$-linear, additive, symmetric monoidal, and essentially surjective. Let us now show that \eqref{eq:aux-2} is moreover fully-faithful. The homomorphisms in $\mathrm{Vect}_\bbZ(k)/_{\!\!-\otimes S_n}$ from $(V,\bigoplus_{s \in \bbZ} V^s)$ to $(V', \bigoplus_{s \in \bbZ}V'^s)$ are given by 
\begin{equation}\label{eq:identification-11}
\bigoplus_{i \in \bbZ} \bigoplus_{s \in \bbZ} \Hom(V^s, V'^{ni+s})\,.
\end{equation}
On the other hand, the homomorphisms in $\mathrm{Vect}_{\bbZ/n}(k)$ from $\underline{\bf \sigma}(V,\bigoplus_{s \in \bbZ} V^s)$ to $\underline{\bf \sigma}(V', \bigoplus_{s \in \bbZ}V'^s)$ can be written as
\begin{equation}\label{eq:identification-22}
\bigoplus_{s \in \bbZ/n} \bigoplus_{s',s'' \in \rho_n(s)}\Hom(V^{s'}, V^{s''})\,.
\end{equation}
The above functor \eqref{eq:aux-2} identifies \eqref{eq:identification-11} with \eqref{eq:identification-22}, and it is therefore fully-faithful. This achieves the proof. 
\end{proof}
When we consider the $G$-action, Theorem \ref{thm:categorification-1} furnish us the affine group $k$-scheme $\mathrm{Spec}(\overline{\bbQ}[\bbQ/\bbZ]^{\mathrm{Gal}(\overline{\bbQ}/k)})$ and the neutral Tannakian category $\mathrm{Vect}_{\bbQ/\bbZ}^{\overline{\bbQ}}(k)$. Note that this affine group $k$-scheme is a twisted form of $\bbG_{m, \mathrm{tors}}$. Definition \ref{def:smaller} and Proposition \ref{prop:smaller} furnish us also the quotient affine group $k$-scheme $\mathrm{Spec}(k[\bbQ/\bbZ]^{\mathrm{Gal}(\overline{\bbQ}/\bbZ)})$ and the neutral Tannakian subcategory $\mathrm{Vect}_{\bbQ/\bbZ}^k(k)$.
\subsection*{Example 2 - Weil restriction}
Thanks to Theorem \ref{thm:categorification-1}, we obtain the following affine group $\bbR$-scheme and neutral Tannakian category:
\begin{eqnarray*}
\mathrm{Spec}(\bbC[\bbQ/\bbZ\times \bbQ/\bbZ]^{\bbZ/2})&& \mathrm{Vect}^\bbC_{\bbQ/\bbZ \times \bbQ/\bbZ}(\bbR)\,.
\end{eqnarray*} 
Under the assignment $\Sigma \mapsto \mathrm{Spec}(\bbC[\Sigma]^{\bbZ/2})$, $\bbZ/n \times \bbZ/n$ equipped with the switch $\bbZ/2$-action, corresponds to the Weil restriction $\mathrm{Res}_{\bbC/\bbR}(\mu_n)$. Therefore, making use of the isomorphism $\bbQ/\bbZ \times \bbQ/\bbZ \simeq \mathrm{lim}_{n\geq 1} (\bbZ/n \times \bbZ/n)$, we conclude that 
\begin{equation}\label{eq:scheme}
\mathrm{Spec}(\bbC[\bbQ/\bbZ \times \bbQ/\bbZ]^{\bbZ/2})\simeq \mathrm{colim}_{n\geq 1} \mathrm{Res}_{\bbC/\bbR}(\mu_n) \simeq \mathrm{Res}_{\bbC/\bbR}(\bbG_{m, \mathrm{tors}})\,.
\end{equation}
Recall from \cite[\S2.1]{Deligne-HodgeII} that the category of real Hodge structures $\mathrm{Hod}(\bbR)$ is defined as $\mathrm{Vect}^\bbC_{\bbZ\times \bbZ}(\bbR)$. The associated affine group $\bbR$-scheme is $\mathrm{Res}_{\bbC/\bbR}(\bbG_m)$. Moreover, the closed immersion $\mathrm{Res}_{\bbC/\bbR}(\bbG_{m,\mathrm{tors}}) \hookrightarrow \mathrm{Res}_{\bbC/\bbR}(\bbG_m)$ is induced by the functor
$$ \mathrm{Hod}(\bbR) = \mathrm{Vect}_{\bbZ \times \bbZ}^\bbC(\bbR) \too \mathrm{Vect}^\bbC_{\bbQ/\bbZ \times \bbQ/\bbZ}(\bbR)\,.$$
\begin{remark}
As in Example 1, we have also the quotient affine group $\bbR$-scheme $\mathrm{Spec}(\bbR[\bbQ/\bbZ\times \bbQ/\bbZ]^{\bbZ/2})$ and the neutral Tannakian subcategory $\mathrm{Vect}^\bbR_{\bbQ/\bbZ \times \bbQ/\bbZ}(\bbR)$.
\end{remark}
\subsection*{Examples 4,5 - Weil numbers and CM fields}
Thanks to Theorem \ref{thm:categorification-1}, we obtain the affine group $k$-schemes
\begin{eqnarray}
\mathrm{Spec}(\overline{\bbQ}[\cW_0(q)]^{\mathrm{Gal}(\overline{\bbQ}/k)}) && \mathrm{Spec}(\overline{\bbQ}[\cW(q)]^{\mathrm{Gal}(\overline{\bbQ}/k)}) \label{eq:schemes-1}\\ 
\mathrm{Spec}(\overline{\bbQ}[\cW^L_0(q)]^{\mathrm{Gal}(\overline{\bbQ}/k)}) && \mathrm{Spec}(\overline{\bbQ}[\cW^L(q)]^{\mathrm{Gal}(\overline{\bbQ}/k)})  \label{eq:schemes-2}
\end{eqnarray}
as well as the neutral Tannakian categories 
\begin{equation}\label{eq:categories}
\mathrm{Vect}^{\overline{\bbQ}}_{\cW_0(q)}(k) \quad \mathrm{Vect}^{\overline{\bbQ}}_{\cW(q)}(k) \quad \mathrm{Vect}^{\overline{\bbQ}}_{\cW^L_0(q)}(k) \quad \mathrm{Vect}^{\overline{\bbQ}}_{\cW^L(q)}(k)\,. 
\end{equation}
Recall from \cite[\S1]{Milne} the construction of the (semi-simple) Tannakian categories of numerical motives over $\bbF_q$ with $k$-coefficients:
\begin{equation}\label{eq:categories-1}
\mathrm{Mot}_0(\bbF_q)_k \quad \mathrm{Mot}(\bbF_q)_k \quad \mathrm{Mot}^L_0(\bbF_q)_k \quad \mathrm{Mot}^L(\bbF_q)_k\,.
\end{equation}
The under-script $0$ stands for numerical motives of weight zero and the upper-script $L$ for numerical motives whose Frobenius numbers belong to $\cW^L(q)$. 
\begin{theorem}{(Milne \cite[\S2]{Milne})}\label{thm:Milne1}
Assuming the Tate conjecture, the following holds:
\begin{itemize}
\item[(i)] The above Tannakian categories \eqref{eq:categories-1} admit $\overline{\bbQ}$-valued fiber functors. The associated affine group $k$-schemes agree with \eqref{eq:schemes-1}-\eqref{eq:schemes-2}.
\item[(ii)] In the particular case where $k=\overline{\bbQ}$, the Tannakian categories \eqref{eq:categories-1} become neutral and moreover equivalent to \eqref{eq:categories}.
\end{itemize}
\end{theorem}

\subsection*{Example 6 - Germs}
The preceding Examples 4 and 5 hold {\em mutatis mutandis} with $\cW(q)$ (and all its variants) replaced by $\cW(p^\infty)$ and $\bbF_q$ replaced by $\overline{\bbF_p}$.

\section{Quantum statistical mechanics}\label{sec:QSM}
In this section we associate to certain Bost-Connes data (see Notation \ref{not:embedding} below) a quantum statistical mechanical system (=QSM-system). In what follows, $k \subseteq \bbC$.

\begin{definition}\label{QSMdef} 
A {\em quantum statistical mechanical system} $(\cA,\sigma_t)$ consists of:
\begin{itemize}
\item[(i)] {\bf Observables:} a separable $C^*$-algebra $\cA$ and a family $R_\iota: \cA \to B(\cH_\iota)$ 
of representations of $\cA$ in the algebra of bounded operators $B(\cH_\iota)$ on separable Hilbert spaces $\cH_\iota$.
\item[(ii)] {\bf Time evolution and Hamiltonian}: a continuous $1$-parameter family of automorphisms $\sigma: \R \to {\rm Aut}(\cA), t \mapsto \sigma_t$. We assume that the representations $R_\iota$ are {\em covariant},  \ie that there exists linear operators $H_\iota$ on $\cH_\iota$ such that, for every $t \in \R$ and $a \in \cA$, 
the following equality holds:
\begin{equation*}
R_\iota(\sigma_t(a)) = e^{itH_\iota} R_\iota(a) e^{-itH_\iota}\,.
\end{equation*}
\end{itemize}
We assume moreover the following:
\begin{itemize}
\item[(iii)] {\bf Partition function}: there exists a real number $\beta_\iota>0$ such that for every $\beta>\beta_\iota$ the operator $e^{-\beta H_\iota}$
is a trace class operator. The associated convergent function $Z_\iota(\beta):={\rm Tr}(e^{-\beta H_\iota}) <\infty$ is called the {\em partition function}.
\item[(iv)] {\bf Symmetries}: there exists a $G$-action $G \to {\rm Aut}(\cA), \gamma \mapsto \tau_\gamma$, which is compatible with the time evolution in the sense that $\sigma_t \circ \tau_\gamma = \tau_\gamma \circ \sigma_t$ for every $\gamma \in G$ and $t \in \bbR$.
\end{itemize}
\end{definition}
Definition \ref{QSMdef} is more restrictive than the classical one \cite{BR,CoMa} since we require that $e^{-\beta H_\iota}$ is a trace class operator for $\beta \gg 0$ and also that the
 group $G$ acts by automorphisms (as opposed to the more general actions by endomorphisms 
 considered in \cite{CM}). These extra assumptions ensure the existence of interesting partition functions and are satisfied by the classical 
 Bost-Connes QSM-system.

\begin{notation}\label{not:embedding}
In what follows, $(\Sigma, \sigma_n)$ is a concrete Bost-Connes datum; see Definition \ref{def:BC-data2}. We assume that all the {\em $G$-equivariant embeddings} $\iota: \Sigma \hookrightarrow \overline{\bbQ}^\times$, \ie $G$-equivariant injective group homomorphisms, contain the roots of unit. This holds in Examples 1,3,4 of \S\ref{sec:BC-systems}. Let $\mathrm{Emb}(\Sigma, \overline{\bbQ}^\times)$ be the set of $G$-equivariant embeddings.
\end{notation}

\begin{proposition}\label{alphaLem}
Up to pre-composition with automorphism of $\Sigma$, the homomorphisms $\sigma_n: \Sigma \to \Sigma, n \in \bbN$, are of the form $\sigma_n(s)=s^{\alpha(n)}$. 
\end{proposition}
\begin{proof}
Consider the induced homomorphism $\eta_n = \iota \circ \sigma_n \circ \iota^{-1} : \iota(\Sigma)
\to \iota(\Sigma)$. Since $\iota(\Sigma)$ contains the group of roots of unit (which is abstractly isomorphic to $\bbQ/\bbZ$) 
and $\Hom(\Q/\Z,\overline{\Q}^\times)=\Hom(\bbQ/\bbZ, \bbQ/\bbZ)\simeq \widehat{\bbZ}$, the induced homomorphism $\eta_n$ restricts~to
\begin{eqnarray}\label{eq:restriction}
\eta_n: \bbQ/\bbZ \too \bbQ/\bbZ && \zeta \mapsto \zeta^{u_n}
\end{eqnarray}
with $u_n \in \widehat{\bbZ}$. Now, recall that $\widehat\Z = X_1 \cup X_2$, where $X_1 := \N \widehat\Z^\times = \cup_{m\in \N} m \, \widehat\Z^\times$ and
$X_2=\cup_p \{ u \in \widehat\Z \,|\, u^{(p)} =0 \}$; the union is over the prime numbers and
$u=(u^{(p)})$ are the coordinates of $u$ in the decomposition $\widehat\Z =\prod_p \widehat\Z_p$. If $u_n \in \widehat{\bbZ}^\times$, then \eqref{eq:restriction} is an isomorphism. This implies that $\alpha(n)=1$. If $u_n \in m \widehat{\bbZ} \subset X_1$, with $m >1$, then the kernel of \eqref{eq:restriction} consists of the roots of unit of order $m$. This implies that $\alpha(n)=m$. If $u_n \in X_2$, then there exists a prime number $p$ such that the kernel of \eqref{eq:restriction} contains all the roots of unit whose order is a power of $p$. This implies that $\alpha(n)=\infty$. Since by assumption, $(\Sigma, \sigma_n)$ is a concrete Bost-Connes system, we hence conclude that $u_n \in \bbN \widehat{\bbZ}$. 
Thus, we can write $u_n$ as a product $u_n= \alpha_\iota(n) v_\iota(n)$. The semi-group property $\sigma_{nm}=\sigma_n \circ \sigma_m$ implies that $u_{nm}=u_n u_m$ and consequently that $\alpha_\iota(nm) \cdot v_\iota(nm)=
\alpha_\iota(n) \alpha_\iota(m) v_\iota(n) v_\iota(m)$. We obtain in this way two semi-group homomorphisms $\alpha_\iota: \bbN \to \bbN$ and 
$v_\iota: \bbN \to \widehat{\bbZ}^\times$. Now, consider the following homomorphisms
\begin{eqnarray*}
\Upsilon_n: \iota(\Sigma) \stackrel{\iota(s) \mapsto \iota(s)^{v_\iota(n)}}{\too} \iota(\Sigma) && \sigma_{\iota, n}: \iota(\Sigma) \stackrel{\iota(s) \mapsto \iota(s)^{\alpha_\iota(n)}}{\too} \iota(\Sigma)\,.
\end{eqnarray*}
The assignment $n \mapsto \xi_{\iota,n}:=\iota^{-1}\circ \Upsilon_n \circ \iota$ gives rise to a semi-group automorphism of $\Sigma$, and we have the equalities $\iota(\sigma_n(s))= \Upsilon_n (\iota(s))^{\alpha_\iota(n)} = \sigma_{\iota,n}(\iota(\xi_{\iota,n}(s)))$. Therefore, it remains only to show that $\alpha_\iota(n)=\alpha(n)$ is independent of $\iota$. This follows from the equalities $\alpha_{\iota}(n)=\#{\rm Ker}(\sigma_{\iota,n})= 
\# \iota^{-1} {\rm Ker}(\sigma_{\iota,n})=\# {\rm Ker}(\sigma_n) =\alpha(n)$.
\end{proof}

\begin{remark}
The automorphisms $\xi_{\iota,n}$ of $\Sigma$, introduced in the proof of Proposition \ref{alphaLem}, do not play any significant role in our construction. In what follows, we will therefore restrict ourselves to the case where $\sigma_n(s)=s^{\alpha(n)}$.
\end{remark}

\begin{corollary}\label{zerosumcor}
For every $\iota \in \mathrm{Emb}(\Sigma, \overline{\bbQ}^\times)$ and $s \in \Sigma_n$, we have the equality:
\begin{equation*}
\frac{1}{\alpha(n)} \sum_{s' \in \rho_n(s)} \iota(s')^{\alpha(m)} = \left\{  \begin{array}{ll} \iota(s)^{\frac{\alpha(m)}{\alpha(n)}} & \mathrm{when}\,\,\alpha(n)|\alpha(m)\\
0& \mathrm{otherwise}\,. \end{array} \right.
\end{equation*}
\end{corollary}

\proof The proof follows automatically from the fact that $\sigma_n(s)=s^{\alpha(n)}$; 
note that when $\alpha(n) \nmid \alpha(m)$, we have $\sum_{s' \in \rho_n(s)} \iota(s')^{\alpha(m)}=0$.
\endproof

\begin{definition}\label{NgammaDef}
Given an embedding $\iota \in \mathrm{Emb}(\Sigma, \overline{\bbQ}^\times)$, consider the homomorphism
\begin{eqnarray*}
N_\iota: \Sigma \too \bbR^\times_+ && s \mapsto |\iota(s)|\,,
\end{eqnarray*}
where $|\cdot|$ is the absolute value of complex numbers; we are implicitly using a fixed embedding $\overline{\bbQ}^\times \subset \bbC$. Note that $N_\iota(\Sigma)$ is a countable multiplicative subgroup of~$\bbR^\times_+$. 
\end{definition}

Given $\iota \in \mathrm{Emb}(\Sigma, \overline{\bbQ}^\times)$, consider the injective group homomorphism
\begin{eqnarray*}
\Sigma \too U(1) \times \R^\times_+ && s \mapsto (\theta_\iota(s), N_\iota(s))\,,
\end{eqnarray*}
where $\theta_\iota(s):=\frac{\iota(s)}{|\iota(s)|}$. It clearly gives rise to the following group decomposition
\begin{eqnarray}\label{eq:decomp-new}
\iota(\Sigma) \stackrel{\sim}{\too} \Sigma_{0,\iota} \times N_\iota(\Sigma) && 
\iota(s) \mapsto (\theta_\iota(s), N_\iota(s))\,,
\end{eqnarray}
where $\Sigma_{0,\iota}:=\theta_\iota(\Sigma)$. Note that the above isomorphism \eqref{eq:decomp-new} generalizes \eqref{eq:isomorphism-decomp}. In the latter case, the group decomposition is independent of the embedding $\iota$.

\begin{corollary}\label{Ngamma}
Given an element $s\in \Sigma$, the following conditions are equivalent:
\begin{itemize} 
\item[(i)] $s$ belongs to the domain of $\rho_n$, with $\# \rho_n(s)=\alpha(n)$.
\item[(ii)] $\frac{\iota(s)}{|\iota(s)|}$ admits  an $\alpha(n)^{\mathrm{th}}$-root in $\Sigma_{0,\iota}$
and $N_\iota(s)= \eta^{\alpha(n)}$ for some $\eta \in N_\iota(\Sigma)$.
\end{itemize} 
\end{corollary}

\proof The proof follows automatically from the above group isomorphism \eqref{eq:decomp-new} and from the fact that $\iota(\Sigma)$ contains roots of unity of all orders; the injectivity of $\iota$ then implies that $s$ has as 
many $\alpha(n)^{\mathrm{th}}$ roots in $\Sigma$.
\endproof 

The following result will be used in \S\ref{Zsec}.
\begin{lemma}\label{Niotaprogr}
The group $N_\iota(\Sigma)\subset \R^\times_+$ can always be decomposed into a union of
countably many geometric progressions $N_\iota(\Sigma)=\bigcup_{r\geq 1} \lambda_r^\Z$.
\end{lemma}
\proof The group $N_\iota(\Sigma)$ is countable. Therefore, it is at most countably generated. 
Let $\{ \lambda_r \}_{r\in \N}$ be a set of generators. Without loss of generality, we can assume that $\lambda_r >1$ for every $r\geq 1$. Hence, we conclude that $N_\iota(\Sigma)=\bigcup_{r\geq 1} \lambda_r^\Z$.
\endproof

We now construct the QSM-system $(\cA_{(\Sigma, \sigma_n)}, \sigma_t)$ associated to $(\Sigma, \sigma_n)$.

\subsection{Observables}\label{sub:observables}
We start by introducing two auxiliary $k$-algebras. The first one, denoted by $\cB_{(\Sigma, \sigma_n)}$, is a generalization of the Bost--Connes algebra \cite{BC}. 
The second one,  denoted by $\cB'_{(\Sigma, \sigma_n, \iota)}$, 
is an extension of a subalgebra\footnote{This subalgebra depends on the choice of an embedding $\iota: \Sigma \hookrightarrow \overline{\bbQ}^\times$.} of the first one by the multiplicative group.  
Making use of $\cB_{(\Sigma, \sigma_n)}$, resp. of $\cB'_{(\Sigma, \sigma_n, \iota)}$, 
and of an embedding $\iota \in \mathrm{Emb}(\Sigma, \overline{\bbQ}^\times)$, 
we then construct the $C^*$-algebra $\cA_{(\Sigma, \sigma_n)}$, resp. $\cA'_{(\Sigma, \sigma_n, \iota)}$, 
of observables of the QSM-system associated to $(\Sigma,\sigma_n)$.

\begin{definition}\label{Akalg}
Let $\cB_{(\Sigma, \sigma_n)}$ be the $k$-algebra generated by the elements $s \in \Sigma$ and by the partial isometries $\mu_n$, $\mu_n^*$,  $n\in \N$. Besides the relations between the elements of the abelian group $\Sigma$, we impose that $\mu_n \mu_m =\mu_{nm}$ and
\begin{eqnarray*}
\mu_n \mu_m^\ast = \mu_m^* \mu_n & \mathrm{when} &(\alpha(m),\alpha(n))=1
\end{eqnarray*}
\begin{equation}\label{catBCrho}
\mu_n s\mu_n^* = \left\{ \begin{array}{ll} \frac{1}{\alpha(n)} \sum_{s' \in \rho_n(s)}s' & \mathrm{when}\,\,s \in \Sigma_n\\
0& \mathrm{otherwise} \end{array} \right.
\end{equation}
\begin{eqnarray*}
\mu_n \mu_n^\ast \mu_n \mu_n^\ast = \mu_n \mu_n^\ast && \mu_n^\ast \mu_n \mu_n^\ast \mu_n = \mu_n^\ast \mu_n\,.
\end{eqnarray*}
Finally, let us write $\cB_{(\Sigma,\sigma_n,\bbC)}$ for the $\bbC$-algebra $\cB_{(\Sigma,\sigma_n)}\otimes_k \bbC$.
\end{definition}

\begin{remark}\label{projectors}
\begin{itemize}
\item[(i)] Unlike the original case \cite{BC}, we do not require that $\mu_n^* \mu_n =1$. When this holds, $\cB_{(\Sigma, \sigma_n)}$ reduces to the semi-group crossed product 
$k[\Sigma]\rtimes \N$ where the semi-group action of $\bbN$ is given by $n \mapsto (s\mapsto \mu_n s \mu_n^*)$.
\item[(ii)] The $k$-algebra of Definition \ref{def:BC-sys2} embeds in $\cB_{(\Sigma,\sigma_n)}$ as the subalgebra generated by the elements $s \in \Sigma$. Moreover, by acting trivially on $\mu_n$ and $\mu_n^\ast$, the $G$-action on $k[\Sigma]$ extends to $\cB_{(\Sigma,\sigma_n)}$.
\end{itemize}
\end{remark}

\begin{definition}\label{tildealg}
Let $\cB'_{(\Sigma,\sigma_n)}$ be the $k$-algebra defined similarly to $\cB_{(\Sigma,\sigma_n)}$ 
but with additional generators $W(\lambda), \lambda \in k^\times$, and additional relations: 
\begin{eqnarray*}
W(\lambda_1 \lambda_2) = W(\lambda_1) W(\lambda_2) && W(\lambda^{-1})= W(\lambda)^{-1}
\end{eqnarray*}
\begin{eqnarray*}
W(\lambda)s=sW(\lambda) & W(\lambda) \mu_n =\mu_n W(\lambda)^{\alpha(n)} & \mu_n^\ast W(\lambda) = W(\lambda)^{\alpha(n)} \mu_n^\ast\,.
\end{eqnarray*}
The new generators $W(\lambda)$ are called the {\em weight operators}. Finally, let us write $\cB'_{(\Sigma, \sigma_n, \bbC)}$ for the $\bbC$-algebra $\cB'_{(\Sigma,\sigma_n)}\otimes_k \C$.
\end{definition}

\begin{remark}{\rm The introduction of the weight operators is motivated by the need to obtain a time evolution whose partition function is convergent for $\beta \gg 0$; see \S\ref{tevol} below. In the particular cases where $N_\iota(\Sigma)=\{ 1 \}$, such as in the original Bost-Connes datum, the weight operators are not necessary and we can work solely with the $k$-algebra $\cB_{(\Sigma,\sigma_n)}$}.
\end{remark}

\begin{notation}
Let $\cH$ be the Hilbert space $\ell^2(\N)$ equipped with the canonical
orthonormal basis $\{ \epsilon_n \}_{n\in \N}$. We write $\cV$ for the $k$-vector space spanned by the $\epsilon_n$'s. Note that the $\bbC$-vector space $\cV_\C:=\cV\otimes_k \C$ is dense in $\cH$.
\begin{itemize}
\item[(i)] Let $\cH_\alpha$ be the Hilbert space $\ell^2(\alpha(\N))$ equipped with the orthonormal basis $\{ \epsilon_{\alpha(n)}\}_{n\in \N}$. We write $\cV_\alpha$ for the $k$-vector space spanned by the $\epsilon_{\alpha(n)}$'s and $\cV_{\alpha, \bbC}$ for the $\bbC$-vector space $\cV_\alpha \otimes_k \bbC$.
\item[(ii)] Given $\iota \in \mathrm{Emb}(\Sigma, \overline{\bbQ}^\times)$, let $\cH_\iota$ be the Hilbert space $\ell^2(N_\iota(\Sigma))$ equipped with the 
orthonormal basis $\{ \epsilon_\eta\}_{\eta\in N_\iota(\Sigma)}$. As above, we write $\cV_\iota$ for the $k$-vector space spanned by the $\epsilon_\eta$'s and $\cV_{\iota, \bbC}$ for the $\bbC$-vector space $\cV_\iota \otimes_k \bbC$.
\item[(iii)] Let $\cH_{\alpha, \iota}$ be the tensor product $\cH_\alpha \otimes \cH_\iota = \ell^2(\alpha(\N)\times N_\iota(\Sigma))$ and $\cV_{\alpha, \iota, \bbC}$ the tensor product $\cV_{\alpha, \bbC} \otimes \cV_{\iota, \bbC}$. We write $\cH_{\alpha,\iota}^{\leq}$, resp. $\cV_{\alpha,\iota,\C}^{\leq}$, for the Hilbert subspace of $\cH_{\alpha, \iota}$, resp. $\bbC$-linear subspace of $\cV_{\alpha,\iota, \bbC}$, spanned by the elements $\epsilon_{\alpha(n),\eta}$ with $\eta \leq 1$. Under these notations, we obtain the splittings $\cH_{\alpha,\iota}=\cH_{\alpha,\iota}^{\leq} \oplus
\cH_{\alpha,\iota}^{>}$ and $\cV_{\alpha,\iota,\C}=
\cV_{\alpha,\iota,\C}^{\leq} \oplus \cV_{\alpha,\iota,\C}^{>}$. 
\end{itemize}
\end{notation}

In the case where $N_\iota(\Sigma)=\bigcup_{r\geq 1} \lambda_r^\Z$ is a union of {\em infinitely} many geometric progressions we will consider also the following Hilbert space; see Lemma \ref{reptildeH}.

\begin{notation}\label{tildeHalphaiota}
Given an embedding $\iota \in \mathrm{Emb}(\Sigma, \overline{\bbQ}^\times)$ and $r \geq 1$, let $\cH_{\iota,r}^\leq$ be the Hilbert subspace of $\ell^2(N_\iota(\Sigma))$ spanned by the orthonormal vectors $\epsilon_\eta$ such that $|\eta|\leq 1$ and $\eta=\lambda_r^k$ for some $k\in \Z_{\leq 0}$. In the same vein, let $\widetilde\cH_{\alpha,\iota}^\leq$ be the Hilbert space $\ell^2(\alpha(\N)) \otimes \bigotimes_{r\geq 1} \cH_{\iota,r}^\leq$ equipped with the standard 
orthonormal basis $\{ \epsilon_{\alpha(n), \lambda_r^{k_r}} \}$ indexed by  
$r \geq 1, \alpha(n) \in \alpha(\bbN)$ and $k_r \in \bbZ_{\leq 0}$.
\end{notation}

\begin{proposition}\label{Hreps}
Given an embedding $\iota \in \mathrm{Emb}(\Sigma, \overline{\bbQ}^\times)$, the assignments
\begin{eqnarray}\label{Hrep:eq}
R_\iota(s)\, \epsilon_{\alpha(n),\eta} & := & \iota(s)^{\alpha(n)} \, \epsilon_{\alpha(n),N_\iota(s) \eta} \nonumber \\
R_\iota(\mu_m)\, \epsilon_{\alpha(n),\eta} & := & \left\{ \begin{array}{ll} \epsilon_{\alpha(mn),\xi} & \mathrm{when}\,
\eta=\xi^{\alpha(m)}\\
0& \mathrm{otherwise} \end{array} \right. \label{Rgammamu} \\
R_\iota(W(\lambda))\, \epsilon_{\alpha(n),\eta}
& := & \lambda^{\alpha(n)} \epsilon_{\alpha(n),\eta} \nonumber
\end{eqnarray}
define a representation $R_\iota$ of the $k$-algebra $\cB'_{(\Sigma,\sigma_n)}$ (and hence of the $\C$-algebra
$\cB'_{(\Sigma, \sigma_n, \bbC)}$) on the $\bbC$-vector space $\cV_{\alpha, \iota, \C}$. 
\end{proposition}
\begin{remark}
By forgetting the action of the weight operators, Proposition \ref{Hreps} gives rise to an action of 
$\cB_{(\Sigma, \sigma_n)}$ (and hence of $\cB_{(\Sigma, \sigma_n, \bbC)}$) on $\cV_{\alpha, \iota, \C}$.
\end{remark}
\proof 
We need to verify that the operators $R_\iota(s), R_\iota(\mu_m),R_\iota(W(\lambda))$ satisfy the relations of Definitions \ref{Akalg} and \ref{tildealg}. Clearly, $R_\iota(s_1s_2)=R_\iota(s_1) R_\iota(s_2)$. Similarly, $R_\iota(\mu_n) R_\iota(\mu_m) = R_\iota(\mu_{nm})$. From the above definitions \eqref{Rgammamu}, we observe that operators $R_\iota(\mu_m^\ast)=R_\iota(\mu_m)^\ast$ are given as follows:
$$ R_\iota(\mu_m^\ast)\epsilon_{\alpha(n),\eta}=\left\{\begin{array}{ll} \epsilon_{\frac{\alpha(n)}{\alpha(m)},\eta^{\alpha(m)}} & \mathrm{when} \,\, \alpha(m)|\alpha(n) \,\,\mathrm{in}\,\, \alpha(\N)  \\
0 & \text{otherwise.}
\end{array}\right. $$
Hence, we have the following identifications:
\begin{eqnarray}
R_\iota(\mu_n)R_\iota(\mu_m^\ast) \epsilon_{\alpha(r),\eta} & = & 
\left\{ \begin{array}{ll} R_\iota(\mu_n) \epsilon_{\frac{\alpha(r)}{\alpha(m)},\eta^{\alpha(m)}} 
& \mathrm{when}\, \alpha(m)|\alpha(r) \\ 0 & \mathrm{otherwise} \end{array} \right. \nonumber \\[3mm]
& = & \left\{ \begin{array}{ll}  \epsilon_{\frac{\alpha(nr)}{\alpha(m)},\xi} & \mathrm{when}\,\, \alpha(m)|\alpha(r) \, 
  \mathrm{and}\,\, \eta^{\alpha(m)}=\xi^{\alpha(n)} \\ 0 & \mathrm{otherwise}\,. \end{array} \right. \nonumber
  \end{eqnarray}
On the other hand, we have
\begin{eqnarray}
R_\iota(\mu_m^\ast) R_\iota(\mu_n) \epsilon_{\alpha(r),\eta} & = & 
\left\{ \begin{array}{ll}  R_\iota(\mu_m^\ast) \epsilon_{\alpha(nr), \zeta} 
& \mathrm{when}\, \eta=\zeta^{\alpha(n)}\, \mathrm{for\, some}\, \zeta \in N_\iota(\Sigma) 
 \\ 0 & \mathrm{otherwise} \end{array} \right. \nonumber \\[3mm]
& = &   \left\{ \begin{array}{ll}  \epsilon_{\frac{\alpha(nr)}{\alpha(m)}, \zeta^{\alpha(m)}} 
& \mathrm{when}\, \eta=\zeta^{\alpha(n)}\,  \mathrm{and}\,\, \alpha(m)|\alpha(nr) 
 \\ 0 & \mathrm{otherwise}\,. \end{array} \right. \nonumber
  \end{eqnarray}
  When $(\alpha(n),\alpha(m))=1$, we have $\alpha(m)|\alpha(nr)\Leftrightarrow \alpha(m)|\alpha(r)$. Moreover, condition
$\eta=\zeta^{\alpha(n)}$ becomes equivalently $\eta^{\alpha(m)}=\zeta^{\alpha(nm)}$. Therefore,
by setting $\xi=\zeta^{\alpha(m)}$, we conclude from above that $R_\iota(\mu_n)R_\iota(\mu_m^\ast)=
R_\iota(\mu_m^\ast) R_\iota(\mu_n)$ when $(\alpha(m),\alpha(n))=1$. In what concerns the operator $R_\iota(\mu_m^\ast\mu_m)$, it corresponds to the projection onto the subspace of $\cV_{\alpha,\iota,\C}$
spanned by the basis elements $\epsilon_{\alpha(n),\eta}$ such that $\eta$ 
admits an $\alpha(m)^{\mathrm{th}}$ root in $N_\iota(\Sigma)$. Similarly, $R_\iota(\mu_m \mu_m^*)$ is the
projection onto the subspace of $\cV_{\alpha,\iota,\C}$ spanned by the basis elements $\epsilon_{\alpha(n),\eta}$ such that $\alpha(m)|\alpha(n)$. Let us now show the following equality:
$$ R_\iota(\mu_m) R_\iota(s)R_\iota(\mu_m^\ast) \epsilon_{\alpha(n),\eta} = \frac{1}{\alpha(m)} \sum_{s' \in \rho_m(s)} R_\iota(s') \epsilon_{\alpha(n), \eta}\,.$$
The left-hand side identifies with 
\begin{equation}\label{eq:identification-1}
\iota(s)^{\alpha(n)/\alpha(m)} \, \epsilon_{\alpha(n),N_\iota(s)^{1/\alpha(m)} \eta}
\end{equation}
when $\alpha(m)|\alpha(n)$ in $\alpha(\N)$ and  
$|\iota(s)|$ has an $\alpha(m)^{\mathrm{th}}$ root in $N_\iota(\Sigma)$. Otherwise, it is zero. In what concerns the right-hand side, it identifies with  
\begin{equation}\label{eq:identification-2}
\frac{1}{\alpha(m)} \sum_{s'\in \rho_m(s)} \iota(s')^{\alpha(n)} \, \epsilon_{\alpha(n),N_\iota(s')\eta}\,.
\end{equation}
Making use of Corollary \ref{zerosumcor} (and Corollary \ref{Ngamma}), we hence conclude that \eqref{eq:identification-2} agrees with \eqref{eq:identification-1}.
In what regards the generators $W(\lambda)$, we clearly have 
\begin{eqnarray*}
R_\iota(W(\lambda_1\lambda_2))=R_\iota(W(\lambda_1)) R_\iota(W(\lambda_2)), && R_\iota(W(\lambda^{-1}))=R_\iota(W(\lambda)^{-1})
\end{eqnarray*}
and also the following equalities:
\begin{eqnarray*}
 R_\iota(W(\lambda)) R_\iota(s) \epsilon_{\alpha(n),\eta} & = & 
R_\iota(s)  R_\iota(W(\lambda)) \epsilon_{\alpha(n),\eta}\\
R_\iota(W(\lambda)) R_\iota(\mu_m) \epsilon_{\alpha(n),\eta} & = & R_\iota(\mu_m) 
R_\iota(W(\lambda)^{\alpha(m)}) \epsilon_{\alpha(n),\eta} \\
R_\iota(\mu_m^*) R_\iota(W(\lambda)) \epsilon_{\alpha(n),\eta}& = & 
R_\iota(W(\lambda)^{\alpha(m)})
R_\iota(\mu_m^*) \epsilon_{\alpha(n),\eta} \,.
\end{eqnarray*}
This achieves the proof.
\endproof

\begin{definition}\label{def:algebra+embedding}
Given an embedding $\iota \in \mathrm{Emb}(\Sigma,\overline{\bbQ}^\times)$, let us denote by $\cB'_{(\Sigma, \sigma_n,\iota)}$ the $k$-subalgebra of $\cB'_{(\Sigma, \sigma_n)}$ generated by the elements $s\in \Sigma$ with $N_\iota(s)\leq 1$,
by the weight operators $W(\lambda)$ with $|\lambda|\leq 1$, 
and also by the partial isometries $\mu_n,\mu_n^*, n \in \bbN$. We write $\cB'_{(\Sigma, \sigma_n,\iota,\C)}$ for the associated $\bbC$-algebra $\cB'_{(\Sigma, \sigma_n,\iota)}\otimes_k \C$.
\end{definition}

\begin{proposition}\label{prop:new}
\begin{itemize}
\item[(i)] When $N_\iota(\Sigma)=\{1\}$, the $\bbC$-algebra $\cB_{(\Sigma,\sigma_n,\C)}$ acts by bounded operators on the Hilbert space $\cH_\alpha$.
\item[(ii)] In general, the representation $R_\iota$ of Proposition \ref{Hreps} extends to a representation $R_\iota$ of the $\bbC$-algebra $\cB'_{(\Sigma, \sigma_n,\iota,\C)}$ by bounded operators on $\cH^{\leq}_{\alpha, \iota}$.
\end{itemize}
\end{proposition}
\begin{proof}
When $N_\iota(\Sigma)=\{1\}$, we have $\| R_\iota(s) \|\leq \sup_n |\iota(s)^{\alpha(n)}| =1$ for every $s \in \Sigma$. Therefore, item (i) follows from this estimate, together with the fact that $\|R_\iota(\mu_m)\|=1$. 
In what concerns item (ii), the action of Proposition \ref{Hreps} extends to an action of $\cB'_{(\Sigma, \sigma_n,\iota,\C)}$ on $\cV_{\alpha,\iota,\C}$. We claim that this action factors through $\cV^\leq_{\alpha,\iota,\C}$:
\begin{itemize}
\item[(a)] If $\eta\leq 1$ and $N_\iota(s)\leq 1$, then $N_\iota(s)\eta \leq 1$.
\item[(b)] If $\eta\leq 1$ and $\eta=\xi^{\alpha(n)}$, then $\xi\leq 1$.
\item[(c)] The weight operators $R_\iota(W(\lambda)), \lambda \in k^\times$, do not alter $\eta$. \end{itemize}
The above items (a)-(c) imply our claim, \ie that the operators $R_\iota(s)$ with $N_\iota(s)\leq 1$, 
$R_\iota(\mu_n)$ and $R_\iota(W(\lambda))$ preserve $\cV_{\alpha, \iota, \bbC}^\leq$. 
Now, since $\cB'_{(\Sigma, \sigma_n,\iota,\C)}$ only contains elements $s \in \Sigma$ with $N_\iota(s)\leq 1$, $\| R_\iota(s) \| \leq  \sup_n N_\iota(s)^{\alpha(n)} \leq 1$. Clearly, we have also $\|R_\iota(\mu_n)\|\leq 1$. In what concerns the weight operators, since $|\lambda|\leq1$, we have $\| R_\iota(W(\lambda))\|\leq \sup_n |\lambda|^{\alpha(n)} \leq 1$. This implies that the action of $\cB'_{(\Sigma, \sigma_n,\iota,\C)}$ on $\cV^\leq_{\alpha, \iota, \bbC}$ extends to an action on $\cH_{\alpha, \iota}^\leq$ by bounded operators.
\end{proof}

\begin{definition}\label{goodBC}
A pair $((\Sigma_n,\sigma_n), \mathrm{Emb}_0(\Sigma, \overline{\bbQ}^\times))$, consisting of a concrete Bost--Connes datum and of a subset 
$\mathrm{Emb}_0(\Sigma, \overline{\bbQ}^\times) \subset \mathrm{Emb}(\Sigma, \overline{\bbQ}^\times)$
is called {\em good} if the algebras $\cB'_{(\Sigma, \sigma_n,\iota,\C)}$ are independent of the embedding $\iota \in \mathrm{Emb}_0(\Sigma, \overline{\bbQ}^\times)$. The pair is called {\em very good} if both the
algebras $\cB'_{(\Sigma, \sigma_n,\iota,\C)}$ as well as the Hilbert spaces $\cH_{\alpha, \iota}^{\leq}$ are independent of the embedding $\iota$. The representations $R_\iota$ may still depend on the choice of the embedding $\iota$. Consult \S\ref{sec:examples} for several examples.
\end{definition}

\begin{definition}\label{Cstardef}
\begin{itemize}
\item[(i)] The $C^*$-algebra $\cA_{(\Sigma,\sigma_n)}$ is defined as the completion of $\cB_{(\Sigma,\sigma_n,\C)}$ in the norm $ \| a \| := \sup_{\iota \in \mathrm{Emb}(\Sigma, \overline{\bbQ}^\times)} 
\| R_\iota(a) \|_{B(\cH_\iota)}$.
\item[(ii)] When the Bost-Connes datum $(\Sigma,\sigma_n)$ is good, the $C^\ast$-algebra $\cA'_{(\Sigma,\sigma_n)}$ is defined as the completion of
$\cB'_{(\Sigma,\sigma_n,\bbC)}$  in the following norm
$$ \| a \| := \sup_{\iota \in \mathrm{Emb}(\Sigma, \overline{\bbQ}^\times)} \| R_\iota(a) \|_{B(\cH_{\alpha, \iota}^\leq)}\,. $$
\end{itemize}
\end{definition}

\begin{remark}\label{Aprimerem}
When $N_\iota(\Sigma)=\{ 1 \}$, we have a family of representations 
$R_\iota$ of $\cA_{(\Sigma, \sigma_n)}$ on $\cB(\cH_{\alpha})$  indexed by embeddings 
$\iota \in \mathrm{Emb}(\Sigma, \overline{\bbQ}^\times)$. In general, for a good Bost-Connes datum,
we have a family of representations $R_\iota$ of $\cA'_{(\Sigma, \sigma_n)}$ on $B(\cH_{\alpha, \iota}^{\leq})$.
\end{remark}

\medskip

In the case where $N_\iota(\Sigma)=\bigcup_{r\geq 1}\lambda_r^\bbZ$ is a union of {\em infinitely} many geometric progressions, we have the following analogue of Proposition \ref{Hreps} (and hence of Remark \ref{Aprimerem}); recall from Notation \ref{tildeHalphaiota} the definition of the Hilbert space $\widetilde\cH_{\alpha,\iota}^\leq$.

\begin{proposition}\label{reptildeH}
The representation $R_\iota$ of the algebra $\cA'_{(\Sigma, \sigma_n)}$ on $B(\cH^\leq_{\alpha,\iota})$ extend as follows to a representation on $B(\widetilde\cH_{\alpha,\iota}^\leq)$ (let $N_\iota(s)=\prod_r \lambda_r^{a_r(s)}$):
\begin{eqnarray*}
 R_\iota(s) \epsilon_{\alpha(n), \lambda_r^{k_r}} & := & \iota(s)^{\alpha(n)} \epsilon_{\alpha(n), \lambda_r^{k_r+a_r(s)}},  \\[2mm]
 R_\iota(\mu_m) \epsilon_{\alpha(n), \lambda_r^{k_r}} & := &
\left\{ \begin{array}{ll}  \epsilon_{\alpha(nm), \lambda_r^{k_r/\alpha(m)}} & \mathrm{when}\, \alpha(m)| k_r
 \\ 0 & \mathrm{otherwise} \end{array} \right. \\[3mm]
 R_\iota(W(\lambda)) \epsilon_{\alpha(n), \lambda_r^{k_r}} & := & \lambda^{\alpha(n)} \epsilon_{\alpha(n),\lambda_r^{k_r}}\,.
\end{eqnarray*}
\end{proposition}
\begin{proof}
The proof is similar to the one of Proposition \ref{Hreps}.
\end{proof}

\smallskip
\subsection{Time evolution and Hamiltonian}\label{tevol}
The constructions in this subsection depend on the choice of an auxiliary 
semi-group homomorphism $g: \bbN \to \bbR^\times_+$.
We assume always that the pair $((\Sigma,\sigma_n), \mathrm{Emb}_0(\Sigma, \overline{\bbQ}^\times))$ is good.

\begin{proposition}\label{tevlem}
Given $\iota \in \mathrm{Emb}_0(\Sigma, \overline{\bbQ}^\times)$, the following assignments
\begin{eqnarray*}
\sigma_t(s):= W(N_\iota(s))^{-it}s& \sigma_t(\mu_n):=g(n)^{it} \mu_n& \sigma_t(W(\lambda)):=W(\lambda)
\end{eqnarray*}
define a continuous $1$-parameter family of automorphisms $\sigma: \bbR \to 
\mathrm{Aut}(\cA'_{(\Sigma,\sigma_n)})$.
\end{proposition}
\begin{proof}
We need to verify that $\sigma_t(ab)=\sigma_t(a)\sigma_t(b)$ and $\sigma_{t+t'}(a)=\sigma_t(\sigma_{t'}(a))$ 
for every $t,t' \in \bbR$ and $a, b \in \cA'_{(\Sigma, \sigma_n)}$. 
The latter equality is clear, since $\sigma_t(\sigma_{t'}(s))=W(N_\iota(s))^{-i (t+t')} s$
and $\sigma_t(\sigma_{t'}(\mu_n))=g(n)^{i(t+t')} \mu_n$. Let us focus then in the first equality. Since $N_\iota$ is a group homomorphism, we have
$$ \sigma_t(s_1 s_2) = W(N_\iota(s_1 s_2))^{-it} s_1 s_2 = \sigma_t(s_1) \sigma_t(s_2)\,.$$
Similarly, since $g$ is a semi-group homomorphism, we have
$$ \sigma_t(\mu_n\mu_m)=\sigma_t(\mu_{nm})=g(nm)^{it}\mu_{nm}=
\sigma_t(\mu_n)\sigma_t(\mu_m)\,.$$
The action on $\mu_n^*$ is then given by $\sigma_t(\mu_n^*)=g(n)^{-it} \mu_n^*$. Note that it is compatible with all the relations between $\mu_n$ and $\mu_n^*$. In what concerns the weight operators, we clearly have the equality $\sigma_t(W(\lambda_1)W(\lambda_2))=\sigma_t(W(\lambda_1))\sigma_t(W(\lambda_2))$. In order to conclude the proof, it remains then to verify the relations of Definitions \ref{Akalg} and \ref{tildealg}. We will focus ourselves in \eqref{catBCrho} and leave the simple verification of the remaining relations to the reader. On one hand, we have:
\begin{eqnarray*}
\sigma_t (\mu_n s \mu_n^*) & = &  \frac{1}{\alpha(n)}\sum_{s'\in \rho_n(s)} \sigma_t(s') \\ 
& = & W(N_\iota(s'))^{-it} \,\frac{1}{\alpha(n)}\sum_{s'\in \rho_n(s)}  s' \\
 & = &  W(N_\iota(s)^{1/\alpha(n)})^{-it} \,\frac{1}{\alpha(n)}\sum_{s'\in \rho_n(s)}  s' \\
& = & W(N_\iota(s))^{-it/\alpha(n)} \,\frac{1}{\alpha(n)}\sum_{s'\in \rho_n(s)}  s'\,.
\end{eqnarray*}
On the other hand, we have:
\begin{eqnarray*}
\sigma_t(\mu_n) \sigma_t(s) \sigma_t(\mu_n^*)  & = & g(n)^{it} 
\mu_n W(N_\iota(s))^{-it} s \,\, g(n)^{-it}  \mu_n^* \\
& = & W(N_\iota(s))^{-it/\alpha(n)} \mu_n s  \mu_n^* \\
& = & W(N_\iota(s))^{-it/\alpha(n)} \frac{1}{\alpha(n)}\sum_{s'\in \rho_n(s)}  s'\,,
\end{eqnarray*}
where the last equality follows from the relations $\cA'_{(\Sigma,\sigma_n)}$ 
between the generators $W(\lambda)$ and $\mu_n$. This achieves the proof.
\end{proof}

Let $\iota \in \mathrm{Emb}(\Sigma, \overline{\bbQ}^\times)$ be an embedding and $R_\iota:
\cA'_{(\Sigma, \sigma_n)}\to B(\cH^\leq _{\alpha, \iota})$ the 
associated representation of Proposition \ref{prop:new}. Consider the following linear operator:
\begin{eqnarray*}
H_\iota: \cH_{\alpha, \iota}^{\leq} \too \cH_{\alpha, \iota}^{\leq}  && 
\epsilon_{\alpha(n),\eta} \mapsto (-\alpha(n) \log(\eta) + \log (g(n)))\epsilon_{\alpha(n),\eta}\,.
\end{eqnarray*}
\begin{proposition}\label{prop:Hamiltonian}
For every $t \in \bbR$ and $a \in \cA'_{(\Sigma, \sigma_n,\iota)}$, we have the equality:
\begin{equation}\label{eq:equality-Hamiltonian}
R_\iota(\sigma_t(a))= e^{itH_\iota}R_\iota(a) e^{-itH_\iota}\,.
\end{equation}
\end{proposition}
\begin{proof}
Clearly, it suffices to verify the above equality \eqref{eq:equality-Hamiltonian} in the case where $a$ is a generator of $\cA'_{(\epsilon, \sigma_n)}$. In what concerns the generators 
$s \in \Sigma$, we have:
\begin{eqnarray*}
R_\iota(\sigma_t(s)) \epsilon_{\alpha(n),\eta} & = & R_\iota(W(N_\iota(s))^{-it} ) R_\iota(s) \epsilon_{\alpha(n),\eta}\\
& = & R_\iota(W(N_\iota(s))^{-it} ) \iota(s)^{\alpha(n)} \epsilon_{\alpha(n),N_\iota(s) \eta}\\
& = & N_\iota(s)^{-it \alpha(n)} \iota(s)^{\alpha(n)} \epsilon_{\alpha(n),N_\iota(s) \eta} \\
& = & \eta^{-it \alpha(n)} N_\iota(s)^{-it \alpha(n)} g(n)^{it} \iota(s)^{\alpha(n)} \eta^{it \alpha(n)} g(n)^{-it} 
\epsilon_{\alpha(n),N_\iota(s) \eta} \\
& = & e^{itH_\iota} \iota(s)^{\alpha(n)} \eta^{it \alpha(n)} g(n)^{-it} \epsilon_{\alpha(n),N_\iota(s) \eta} \\
& = & e^{itH_\iota}R_\iota(s) \eta^{it \alpha(n)} g(n)^{-it} \epsilon_{\alpha(n),\eta} \\
& = & e^{itH_\iota}R_\iota(s) e^{-itH_\iota} \epsilon_{\alpha(n),\eta}\,.
\end{eqnarray*}
For the generators $\mu_m, m \in \bbN$, the left-hand side of \eqref{eq:equality-Hamiltonian} identifies with
$$ R_\iota(\sigma_t(\mu_m)) \epsilon_{\alpha(n),\eta} =  g(m)^{it} 
R_\iota(\mu_{m})  \epsilon_{\alpha(n),\eta} =  \left\{  \begin{array}{ll}
    g(m)^{it} \epsilon_{\alpha(nm), \xi} & \text{when}\,\, \eta=\xi^{\alpha(m)} \\
   0 & \text{otherwise}\,.
  \end{array} \right.
$$
On the other hand, the right-hand side identifies with
\begin{eqnarray*}
e^{itH_\iota}R_\iota(\mu_{m}) e^{-itH_\iota} 
\epsilon_{\alpha(n),\eta} & = & e^{itH_\iota}R_\iota(\mu_{m})
\eta^{it \alpha(n)}  g(n)^{-it} \epsilon_{\alpha(n),\eta}  \\
(\text{when}\,\,\eta=\xi^{\alpha(m)}; 0\,\,\text{otherwise}) 
& = & e^{itH_\iota} \eta^{it \alpha(n)} g(n)^{-it} \epsilon_{\alpha(nm), \xi} \\
 & = & \xi^{-it \alpha(nm)} g(nm)^{it} \eta^{it \alpha(n)} g(n)^{-it} 
\epsilon_{\alpha(nm), \xi} \\
& = & g(m)^{it} \epsilon_{\alpha(nm), \xi}\,.
\end{eqnarray*}
Finally, in what concerns the generators $W(\lambda), \lambda \in k^\times$, we have:
\begin{eqnarray*}
R_\iota(W(\lambda)) \epsilon_{\alpha(n),\eta} = \lambda^{\alpha(n)} \epsilon_{\alpha(n),\eta} & = &
\eta^{-it \alpha(n)} g(n)^{it} \lambda^{\alpha(n)} \eta^{it \alpha(n)} g(n)^{-it} \epsilon_{\alpha(n),\eta} \\
& =&  e^{itH_\iota} \lambda^{\alpha(n)} \eta^{it \alpha(n)} g(n)^{-it} \epsilon_{\alpha(n),\eta} \\
& = & e^{itH_\iota}R_\iota(W(\lambda))
\eta^{it \alpha(n)} g(n)^{-it} \epsilon_{\alpha(n),\eta}  \\
& = & e^{itH_\iota}R_\iota(W(\lambda)) e^{-itH_\iota} \epsilon_{\alpha(n),\eta}\,.
\end{eqnarray*}
This achieves the proof.
\end{proof}
\begin{remark}
By combining Propositions \ref{tevlem}-\ref{prop:Hamiltonian} with \S\ref{sub:observables}, we hence obtain the QSM-system $(\cA'_{(\Sigma, \sigma_n)}, \sigma_t)$ associated to $(\Sigma, \sigma_n)$ (and to an embedding $\iota \in \mathrm{Emb}(\Sigma, \overline{\bbQ}^\times)$). In the next two subsections we prove that this QSM-system satisfies the extra assumptions (iii)-(iv) of Definition \ref{QSMdef}.
\end{remark}

\subsection{Partition function: finite union of geometric progressions}
Recall from \S\ref{tevol} that $g: \bbN \to \bbR^\times_+$ is an auxiliary semi-group homomorphism. Let us denote by $\beta_0$ the exponent of convergence of the series $\sum_{\alpha(n)\geq 1} g(n)^{-\beta}$. 

In this subsection we assume that for every embedding $\iota \in \mathrm{Emb}(\Sigma,\overline{\bbQ}^\times)$, the countable multiplicative subgroup  $N_\iota(\Sigma)$ of $\R^\times_+$ is a finite union of geometric progressions $\N_\iota(\Sigma)=\cup_{j=1}^M
q_{\iota,j}^\Z$, with $q_{\iota,j} > 1$.

\begin{proposition}\label{prop:partition}
The partition function $Z_\iota(\beta)$ is computed by the series:
\begin{equation}\label{eq:sum-partition}
\sum_{\eta\leq 1 \in N_\iota(\Sigma)} \sum_{\alpha(n) \geq 1}  \eta^{\alpha(n)\beta} 
g(n)^{-\beta} \,.
\end{equation}
Moreover, for every $\beta >\beta_0$, the above series \eqref{eq:sum-partition}
is convergent. Consequently, $e^{-\beta H_\iota}$ is trace class operator. 
\end{proposition}

\begin{proof}
By construction, the operator $H_\iota$ is diagonal in the basis $\epsilon_{\alpha(n),\eta}$
 of $\cH_{\alpha, \iota}^\leq$. 
Hence, the associated partition function agrees with the following series
\begin{eqnarray*}
Z_\iota(\beta):= \mathrm{Tr}(e^{-\beta H_\iota}) & = & \sum_{\eta\leq 1 \in N_\iota(\Sigma)} \sum_{\alpha(n) \geq 1}  \langle \epsilon_{\alpha(n),\eta} , e^{-\beta H_\iota}  \epsilon_{\alpha(n),\eta} \rangle \\
& = &  \sum_{\eta\leq 1 \in N_\iota(\Sigma)} \sum_{\alpha(n) \geq 1}   
e^{-\beta (-\alpha(n) \log(\eta) + \log (g(n)))} \\
&= & \sum_{\eta\leq 1 \in N_\iota(\Sigma)} \sum_{\alpha(n) \geq 1} \eta^{\alpha(n)\beta} g(n)^{-\beta}\,.
\end{eqnarray*}
Under the above assumption on $N_\iota(\Sigma)$, the sum \eqref{eq:sum-partition} can be re-written as 
\begin{equation}\label{eq:series}
\sum_{\alpha(n) \geq 1} \sum_{j=1}^M \sum_{k\geq 0} q_{\iota,j}^{-k\alpha(n)\beta} g(n)^{-\beta} 
= \sum_{\alpha(n) \geq 1} \sum_{j=1}^M \frac{1}{1-q_{\iota,j}^{-\alpha(n)\beta}} g(n)^{-\beta}\,.
\end{equation}
Let $j_{\mathrm{max}}$ (resp. $j_{\mathrm{min}}$) be the index $j$ that realize the maximum (resp. minimum) above. Since $1-q_{\iota,j_{\mathrm{max}}}^{-\alpha(n)\beta} \geq 1-q_{\iota,j_{\mathrm{max}}}^{-\beta}$ when $\alpha(n)\geq 1$ and $1-q_{\iota,j_{\mathrm{min}}}^{-\alpha(n)\beta} \leq 1$, we have 
\begin{eqnarray*}
g(n)^{-\beta} \frac{M}{1-q_{\iota,j_{\mathrm{max}}}^{-\alpha(n)\beta}} \leq g(n)^{-\beta} \frac{M}{1-q_{\iota,j_{\mathrm{max}}^{-\beta}}} & \mathrm{and} & g(n)^{-\beta} \frac{M}{1-q_{\iota,j_{\mathrm{min}}}^{-\alpha(n)\beta}} \geq g(n)^{-\beta} M\,.
\end{eqnarray*}
This implies that \eqref{eq:series} is bounded above and below by series whose convergence and divergence depends only on the series $\sum_{\alpha(n)\geq 1} g(n)^{-\beta}$. 
As a consequence, the above series \eqref{eq:sum-partition} converges for every $\beta > \beta_0$ and diverges for every $\beta \leq \beta_0$. This concludes the proof.
\end{proof}

\subsection{Partition function: infinite union of geometric progressions}\label{Zsec}

In Proposition \ref{prop:partition} we have only treated the convergence of the partition 
function $Z(\beta)$ in the case where the group $N_\iota(\Sigma)\subset \R^\times_+$ has the form 
$N_\iota(\Sigma)=\cup_{j=1}^M q_{\iota,j}^\Z$, for some $q_{\iota,j} > 1$. Here we treat the
case of infinitely many progressions. 

In order to control the convergence properties of the partition function, it is
useful to introduce another choice of a homomorphism, in addition to the choice of
$g: \N \to \R^\times_+$. We modify the time evolution by introducing an additional auxiliary
choice of a group homomorphism $h: N_\iota(\Sigma) \to \R^\times_+$. 

\begin{proposition}\label{tevolprop2}
Given homomorphisms $g: \N \to \R^\times_+$ and $h: N_\iota(\Sigma) \to \R^\times_+$, 
the following assignments
\begin{eqnarray*}
 \sigma_t(s)= W(h(N_\iota(s)))^{-it} s & \sigma_t(\mu_m) = g(m)^{it} \mu_m & \sigma_t(W(\lambda))= W(\lambda)
\end{eqnarray*}
defines a $1$-parameter family of automorphisms 
$\sigma: \R \to {\rm Aut}(\cA'_{(\Sigma, \sigma_n)})$.
\end{proposition}
\proof The proof is similar to the one of Proposition \ref{tevlem}.
\endproof

\begin{proposition}\label{rHam}
The Hamiltonian implementing the time evolution of Proposition \ref{tevolprop2} in the representation 
of Lemma \ref{reptildeH} is given by
$$ H_\iota\, \epsilon_{\alpha(n), \lambda_r^k}= (-\alpha(n) k \log(h(\lambda_r)) + \log(g(n))) \, 
\epsilon_{\alpha(n), \lambda_r^k}. $$
\end{proposition}

\begin{proof}
The proof is analogous to the one of Proposition \ref{prop:Hamiltonian}; 
instead of $\cH_{\alpha,\iota}^{\leq}$ and $\epsilon_{\alpha(n), \eta}$, we use the Hilbert space 
$\widetilde{\cH}_{\alpha,\iota}^{\leq}$ and its standard orthonormal basis 
$\epsilon_{\alpha(n),\lambda_r^{k_r}}$; see Notation \ref{tildeHalphaiota}.
\end{proof}
\smallskip

All the above constructions work with an arbitrary homomorphism $h$ . Let us now focus on the following example:
\begin{example}\label{ex:primes}
Recall that $\{\lambda_r\}_{r \geq 1}$ is a set of generators of $N_\iota(\Sigma)$. Let $h: N_\iota(\Sigma) \to \bbR^\times_+$ be the homomorphism defined as $h(\lambda_r):= p_r$, where $\{ p_r \}_{r\geq 1}$ stands for an enumeration of the
prime numbers (for example the natural one in increasing order).
\end{example}

Recall that $\beta_0$ denotes the exponent of convergence of the series $\sum_{\alpha(n)\geq 1} g(n)^{-\beta}$.
\begin{theorem}\label{thm:Zgeneral}
When $h$ is as in Example \ref{ex:primes}, the partition function $Z_\iota(\beta)$ is computed by the following series:
\begin{equation}\label{eq:series-new}
\sum_{\alpha(n)\geq 1} g(n)^{-\beta} \zeta(\beta \alpha(n))\,.
\end{equation}
Moreover, for every $\beta> \max\{\beta_0, 3/2\}$, the series \eqref{eq:series-new} is convergent. Consequently, $e^{-\beta H_\iota}$ (where $H_\iota$ is as in Proposition \eqref{rHam}) is a trace class operator.
\end{theorem}

\proof Since the operator $H_\iota$ (hence also $e^{-\beta H_\iota}$) is diagonal on the orthonormal
basis $\epsilon_{\alpha(n), \lambda_r^{k_r}}$ of $\widetilde{\cH}_{\alpha,\iota}^{\leq}$, we have the following equality
$$ Z(\beta):={\rm Tr}(e^{-\beta H_\iota}) = \sum_{\alpha(n)\geq 1} g(n)^{-\beta} 
\prod_r \sum_{k_r \geq 0} h(\lambda_r)^{-k_r \alpha(n) \beta}, $$
where the negative sign in the exponent of $h(\lambda_r)$ comes from the fact that we are writing the sum over $k_r \geq 0$
instead of $k_r\leq 0$. For each fixed $r\in \N$, 
we can compute the following series
$$ \sum_{k_r \geq 0} \lambda_r^{-k_r \alpha(n) \beta} = \frac{1}{1- h(\lambda_r)^{-\alpha(n) \beta}} $$
as the sum of a geometric series. For the particular choice of $h(\lambda_r)=p_r$, this is equal to
$(1-p_r^{-\alpha(n) \beta})^{-1}$. Thus, we can rewrite the trace as
\begin{equation}\label{eq:trace+sum}
\sum_{\alpha(n)\geq 1} g(n)^{-\beta} \prod_p (1-p^{-\alpha(n) \beta})^{-1},
\end{equation}
where the product is over the prime numbers. The Euler product converges to the Riemann zeta function,
$\prod_p (1-p^{-s})^{-1}=\zeta(s)$. Therefore, we can rewrite \eqref{eq:trace+sum} as
\begin{equation}\label{eq:series-key}
\sum_{\alpha(n)\geq 1} g(n)^{-\beta} \zeta(\beta \alpha(n))\,.
\end{equation}
In order to understand the convergence of this series, we need to estimate the behavior of the values
$\zeta(\beta \alpha(n))$ of the Riemann zeta function. When $s$ is real and $2(s-1)>1$, we can use the estimate
$$ \zeta(s) = \sum_{n\geq 1} n^{-s} = 1 + \sum_{n\geq 2} n^{-s} \leq 1 + \int_2^\infty \frac{dx}{x^s} =1 - \frac{1}{2 (s-1)} \leq 1. $$
Hence, for $2(\beta \alpha(n)-1) \geq 1$, the terms $\zeta(\beta \alpha(n))$ are all bounded above by
$\zeta(\beta \alpha(n))\leq 1$. This gives rise to the following inequality: 
$$ \sum_{\alpha(n)\geq 1} g(n)^{-\beta} \zeta(\beta \alpha(n)) \leq \sum_{\alpha(n)\geq 1} g(n)^{-\beta}\,. $$
Since $\alpha(n)\geq 1$ for all $n\in \N$, the condition $2(\beta \alpha(n)-1) \geq 1$ is satisfied for all $n\in \N$
if $2(\beta -1) \geq 1$ is satisfied, that is, if $\beta > 3/2$. We hence conclude that the above series \eqref{eq:series-key} converges for every $\beta > \max\{ \beta_0, 3/2\}$. 
\endproof
\begin{remark}
Unlike Proposition \ref{prop:partition}, we are only using an estimate from above. Therefore, we can conclude only that $e^{-\beta H_\iota}$ is trace class for $\beta > \max\{ \beta_0, 3/2\}$. 
\end{remark}

\subsection{Symmetries}\label{SymSec}
Given an embedding $\iota \in \mathrm{Emb}(\Sigma,\overline{\bbQ}^\times)$, consider the (sub)group 
\begin{equation*}\label{Giota}
G_\iota:=\{ \gamma \in G \,|\, N_\iota(\gamma(s))=N_\iota(s), \ \forall s\in \Sigma \}\subset G.
\end{equation*}
\begin{notation}\label{not:Gzero}
Let $G_0$ be the intersection of the subgroups $G_\iota$ with $\iota \in \mathrm{Emb}_0(\Sigma, \overline{\bbQ}^\times)$.
\end{notation}

\begin{proposition}\label{prop:G-action}
The following assignments
\begin{eqnarray*}
s \mapsto \gamma(s) & \mu_n \mapsto \mu_n & W(\lambda) \mapsto W(\lambda)
\end{eqnarray*}
define an action of $G$ on the $k$-algebra $\cB_{(\Sigma, \sigma_n)}$ and  of $G_\iota$ on
the $k$-algebra $\cB'_{(\Sigma, \sigma_n,\iota)}$.
\end{proposition}
\begin{proof}
It follows automatically from the definition of $\cB_{(\Sigma, \sigma_n)}$ and $\cB'_{(\Sigma, \sigma_n,\iota)}$ in terms of generators and relations; see Definitions \ref{Akalg} and \ref{def:algebra+embedding}. Note that in the case of the $k$-algebra $\cB'_{(\Sigma, \sigma_n,\iota)}$ we need to restrict to the subgroup $G_\iota \subset G$ to ensure that
$s \mapsto \gamma(s)$ preserves the subset $\{ s\in \Sigma\,|\, N_\iota(s) \leq 1 \}$.
\end{proof}

\begin{proposition}\label{prop:G-action2}
The actions of Proposition \ref{prop:G-action} extend to actions
\begin{eqnarray}\label{eq:G-actions-new}
\tau:G \to \mathrm{Aut}(\cA_{(\Sigma, \sigma_n)}) &\mathrm{and} & 
\tau:G_\iota \to \mathrm{Aut}(\cA'_{(\Sigma, \sigma_n,\iota)})\,.
\end{eqnarray}
These actions are compatible with the time evolution in the sense that 
$\sigma_t \circ \tau_\gamma = \tau_\gamma \circ \sigma_t$ 
for every $\gamma \in G$ (or $\gamma \in G_\iota$) and $t \in \bbR$.\end{proposition}
\begin{proof}
Since $(\Sigma, \sigma_n)$ is a Bost-Connes datum, the group $G$ acts continuously on $\Sigma$. Therefore, the $G$-action of Proposition \ref{prop:G-action} extends first to the $\bbC$-algebra $\cB_{(\Sigma, \sigma_n,\C)}$ and
then to the $C^*$-algebras $\cA_{(\Sigma, \sigma_n)}$. Similarly, the $G_\iota$-action extends first to 
$\cB'_{(\Sigma, \sigma_n,\iota,\C)}$ and then to $\cA'_{(\Sigma, \sigma_n,\iota)}$. 
We obtain in this way the above actions \eqref{eq:G-actions-new}. The compatibility with the time evolution is given by the following equalities:
$$ \sigma_t(\tau_\gamma(\mu_n))=  \sigma_t(\mu_n) = g(n)^{it}\mu_n= g(n)^{it} \tau_\gamma(\mu_n) = \tau_\gamma (\sigma_t (\mu_n)) \,.$$
Similarly, we have $\tau_\gamma( \sigma_t(W(\lambda)) = \sigma_t (\tau_\gamma(W(\lambda)))$. Note also that when $\gamma \in G_\iota$, we have moreover the following equalities: 
$$ \sigma_t(\tau_\gamma(s))= W(N_\iota(\tau_\gamma(s)))^{-it}\tau_\gamma(s) = \tau_\gamma(W(N_\iota(s))^{-it})\tau_\gamma(s) = \tau_\gamma(\sigma_t(s))\,.$$
This achieves the proof. 
\end{proof}

\subsection{Gibbs states}\label{sec:Gibbs}
Recall from \cite[Vol.~I \S2.3.3]{BR} that a {\em state} on a unital $C^*$-algebra $\cA$ is a 
continuous linear functional $\varphi: \cA \to \C$
that is normalized, \ie $\varphi(1)=1$, and satisfied the positivity condition, \ie $\varphi(a^* a)\geq 0$ for all $a\in \cA$. An {\em equilibrium state} of a quantum statistical mechanical system $(\cA,\sigma_t)$ is a state $\varphi$ that is invariant
with respect to the time evolution, $\varphi(\sigma_t(a))=\varphi(a)$, for all $t\in \R$ and $a\in \cA$; consult \cite[Vol.~II \S5.3]{BR} for further details.
If the QSM-system $(\cA,\sigma_t)$ has a representation
$R_\iota$ on a Hilbert space $\cH_\iota$ with a Hamiltonian $H_\iota$ for which the partition function 
$Z(\beta)={\rm Tr}(e^{-\beta H_\iota})$ is convergent for all $\beta > \beta_\iota$, then we can
define a special class of equilibrium states, namely the {\em Gibbs states} at 
inverse temperature $\beta$:
\begin{equation}\label{Gibbsphi}
\varphi_{\beta,\iota}(a):= \frac{{\rm Tr}(R_\iota(a)\, e^{-\beta H_\iota})}{{\rm Tr}(e^{-\beta H_\iota})}. 
\end{equation}

\begin{definition}\label{def:GSigma}
Let $\Phi: G \to {\rm Aut}(\Sigma)$ be the $G$-action on $\Sigma$ and  
$G_\Sigma$ the quotient $G/{\rm Ker}(\Phi)$. We denote by $\pi_\Sigma: G \to G_\Sigma$ the quotient map, by $Z(G_\Sigma)$ the center of $G_\Sigma$, and by $\tilde Z(G_\Sigma)$ the preimage $\pi_\Sigma^{-1}( Z(G_\Sigma)) \subset G$.
\end{definition}

\begin{example}\label{ex:GSigma}
In the case where $\Sigma=\Q/\Z$, we have $G_\Sigma=G^{\mathrm{ab}}=Z(G_\Sigma)$. In contrast, when $\Sigma=\overline{\Q}^\times$, we have $G_\Sigma=G$ and (by the 
Neukirsch--Uchida theorem) $Z(G)=\{ 1 \}$.
\end{example}
Recall from Notation \ref{not:Gzero} the definition of the group $G_0$. Let $G_{\Sigma,0}$ be its image under the quotient map $\pi_\Sigma: G\to G_\Sigma$ and $\widetilde{G}_{\Sigma, 0}$ the pre-image $\pi_\Sigma^{-1}(Z(G_\Sigma)\cap G_{\Sigma,0})$.

\begin{lemma}\label{GRiota}
Given a good pair $((\Sigma,\sigma_n), \mathrm{Emb}_0(\Sigma, \overline{\bbQ}^\times))$ in the sense of Definition \ref{goodBC}, the $G_0$-action $\tau: G_0 \to {\rm Aut}(\cB'_{(\Sigma, \sigma_n)})$ induces a $\widetilde{G}_{\Sigma, 0}$-action $R_\iota \mapsto R_{\iota\circ \gamma}$ on the set $\{ R_\iota \}$ of representations $R_\iota: \cB'_{(\Sigma, \sigma_n)}
\to B(\cH_\alpha^\leq)$.
\end{lemma}

\proof 
The condition that $\iota\circ \gamma$ is still a $G$-equivariant embedding implies that the image
of $\gamma \in G_0$ in the quotient $G_\Sigma$ commutes with all elements of $G_\Sigma$. Therefore, it belongs to the center and hence to the intersection $Z(G_\Sigma) \cap G_{\Sigma,0}$. We now verify on the generators of $\cB'_{(\Sigma, \sigma_n)}$ that $R_\iota (\tau_\gamma(a))=R_{\iota\circ\gamma}(a)$.
For $s\in \Sigma$ we have 
$$ R_\iota (\tau_\gamma(s)) \epsilon_{\alpha(n),\eta}= 
 \iota(\gamma(s))^{\alpha(n)} \epsilon_{\alpha(n),N_\iota(\gamma(s))\eta} 
= R_{\iota\circ\gamma}(s) \epsilon_{\alpha(n),\eta}. $$
In the remaining cases, since the action of $R_\iota(\mu_n)$ does not depend on $\iota$
and $\mu_n$ is fixed by $G$, we have
$R_\iota (\tau_\gamma(\mu_n)) = R_\iota(\mu_n) = R_{\iota\circ \gamma}(\mu_n)$. Similarly, we have 
$R_\iota(\tau_\gamma(W(\lambda)))= R_\iota(W(\lambda)) = R_{\iota\circ \gamma}(W(\lambda))$.
\endproof

\begin{proposition}\label{GGibbs}
Given a good pair $((\Sigma,\sigma_n), \mathrm{Emb}_0(\Sigma, \overline{\bbQ}^\times))$, 
the $G_0$-action 
$\tau: G_0 \mapsto {\rm Aut}(\cB'_{(\Sigma, \sigma_n)})$ induces a pullback $\tilde G_{\Sigma,0}$-action 
$$ \varphi_{\beta,\iota} \mapsto \tau_\gamma^*(\varphi_{\beta,\iota}) =\varphi_{\beta,\iota}\circ \tau_\gamma =
\varphi_{\beta,\iota\circ \gamma}$$
on the set of Gibbs states at a fixed inverse temperature $\beta$.
\end{proposition}

\proof
Let $(\cA,\sigma_t)$ be a quantum statistical mechanical system. It is well-known (see 
\cite{BR}) that an action $\tau: G_0 \to {\rm Aut}(\cA)$ by 
automorphisms which verifies $\sigma_t \circ \tau_\gamma = \tau_\gamma \circ \sigma_t$,
(for all $\gamma \in G_0$ and $t\in \R$) induces a pullback action on the set of 
Gibbs states of $(\cA,\sigma_t)$ at a fixed inverse temperature~$\beta$:
$$ \varphi_{\beta,\iota} \mapsto \tau_\gamma^*(\varphi_{\beta,\iota}):= \varphi_{\beta,\iota} \circ \tau_\gamma. $$
For elements in $\tilde G_{\Sigma,0}$ we have:
$$ \varphi_{\beta,\iota} \circ \tau_\gamma(a)= \frac{{\rm Tr}(R_\iota(\tau_\gamma(a))\, 
e^{-\beta H_\iota})}{{\rm Tr}(e^{-\beta H_\iota})} 
 = \frac{ {\rm Tr} ( R_{\iota\circ \gamma}(a)\, e^{-\beta H_\iota} ) } { {\rm Tr} ( e^{-\beta H_\iota} ) } = 
\varphi_{\beta,\iota\circ \gamma}(a)\,. $$
Therefore, the proof follows now from Lemma \ref{GRiota}.
\endproof

Recall from \cite{CM} that the ground states (or zero temperature equilibrium states) of a quantum statistical mechanical
system are defined as weak limits when $\beta \to \infty$ 
of the Gibbs states at inverse temperature $\beta$. Concretely, $ \varphi_{\infty,\iota}(a):= \lim_{\beta\to \infty} \varphi_{\beta,\iota}(a)$. These are given by traces of projections onto the Kernel of the Hamiltonian $H_\iota$.

\medskip

The following result rephrases in our setting the ``fabulous states" property 
of the Bost--Connes system (as formulated in \cite{CM}), namely the intertwining of
symmetries of the quantum statistical mechanical system and Galois symmetries.

\begin{proposition}\label{lem:fabulous}
Given a good pair $((\Sigma,\sigma_n), \mathrm{Emb}_0(\Sigma, \overline{\bbQ}^\times))$, there is an induced $\tilde G_{\Sigma,0}$-action $\varphi_{\infty,\iota}\mapsto \varphi_{\infty,\iota\circ \gamma}$ on the ground states.  
In the particular case where $N_\iota(\Sigma)=\{ 1 \}$, the ground states take
 values $\varphi_{\infty,\iota}(s) = \iota(s)$ on the generators $s\in \Sigma$ of the
 algebra of observables. Moreover, this action recovers the Galois action of the 
subgroup $\tilde Z(G_\Sigma) \subset G$ on $\iota(\Sigma)$.
\end{proposition}

\proof Recall from \S\ref{tevol} that the Hamiltonian $H_\iota$ is given by 
$$ H_\iota \epsilon_{\alpha(n),\eta} = 
(-\alpha(n) \log(\eta) + \log (g(n)))\epsilon_{\alpha(n),\eta}\,.$$
Its kernel is one-dimensional and it is spanned by the vector $\epsilon_{1,1}$.
Thus, the ground state $\varphi_{\infty,\iota}$ is given by the projection onto the
kernel $H_\iota$, that is, by 
$$ \varphi_{\infty,\iota}(a) = \langle \epsilon_{1,1} , R_\iota(a)\, \epsilon_{1,1} \rangle. $$
When we evaluate it on a generator $s\in \Sigma$ we obtain
$$ \langle \epsilon_{1,1}, R_\iota(s) \epsilon_{1,1} \rangle = \iota(s) \langle \epsilon_{1,1}, \epsilon_{1,N_\iota(s)} \rangle. $$
This is zero unless $N_\iota(s)=1$, in which case it is equal to $\iota(s)$. 
Note that in the case where $N_\iota(\Sigma)=\{ 1 \}$, we have $G_0=G$. Hence, 
$\tilde G_{\Sigma,0}=\tilde Z(G_\Sigma)$. By the $G$-equivariance of the
embedding $\iota$, this implies that the $\tilde Z(G_\Sigma)$-action is given by $\iota(\gamma(s))=\gamma(\iota(s))$, \ie by the restriction to the subgroup $\tilde Z(G_\Sigma)$ of the Galois action of $G$ on $\iota(\Sigma)$. This achieves the proof.
\endproof

\medskip
\subsection{Bost-Connes data with trivial $\alpha$}\label{alpha1sec}
Let $(\Sigma, \sigma_n)$ be an abstract Bost-Connes datum for which the semi-group homomorphism $\alpha$ is trivial, \ie $\alpha(n)=1$ for every $n \in \bbN$. By Definition \ref{def:BC-data}, $(\Sigma, \sigma_n)$ is {\em not} a concrete Bost-Connes datum. Nevertheless, we explain here how we can still construct a partial QSM-system.

\smallskip

\begin{lemma}\label{rtimesQ}
When $N_\gamma(\Sigma)=\{ 1 \}$, the algebra 
$\cB=\cB_{(\Sigma,\sigma_n)}$ is isomorphic to the group crossed product
algebra $k[\Sigma] \rtimes \Q^\times_+$.
\end{lemma}

\proof When $\alpha(n)=1$ for every $n\in \N$, $\rho_n$ and $\sigma_n$ are inverse of each other. In particular, they are automorphisms of $k[\Sigma]$. The relation \eqref{catBCrho} hence implies that $\mu_n \mu_n^*=1$. 
Consequently, the $\mu_n$'s are not just isometries but rather unitaries, with $\mu_n^*=\mu_n^{-1}=\mu_{1/n}$,
implementing the action of $\Q^\times_+$ on $k[\Sigma]$. Making use of them, we then obtain an isomorphism between $\cB=\cB_{(\Sigma, \sigma_n)}$ and $k[\Sigma] \rtimes \Q^\times_+$.
\endproof

Thanks to Lemma \ref{rtimesQ}, the operators $\mu_n^*$'s are invertible. Therefore, the Hilbert space representations of 
$\cB=\cB_{(\Sigma,\sigma_n)}$ need to be modified accordingly. There is a unique natural way
to proceed: the action of the semi-group $\N$ on the Hilbert space $\ell^2(\N)$
is just the regular representation. In the case where the semi-group
is replaced by the group it generates, we correspondingly consider the 
regular representation of the group. In our case, this means the
action of $\Q^\times_+$ on the Hilbert space $\ell^2(\Q^\times_+)$ is given by multiplication
on its basis elements.

\begin{lemma}\label{QrepsA}
Assume that $N_\gamma(\Sigma)=\{ 1 \}$. Given an embedding $\iota \in \mathrm{Emb}(\Sigma,\overline{\bbQ}^\times)$, the assignments $R_\gamma(\pi) \epsilon_r := \gamma(\pi)^r\, \epsilon_r$ and $R_\gamma(\mu_s) \epsilon_r := \epsilon_{sr}$, with $\pi \in \Sigma$ and $s,r\in \Q^\times_+$, define a representation of $\cB_{(\Sigma,\sigma_n)}=k[\Sigma]\rtimes \Q^\times_+$ 
on the Hilbert space $\ell^2(\Q^\times_+)$.
\end{lemma}

\proof The crossed product relation is satisfied since $R_\gamma(\mu_s) R_\gamma(\pi) R_\gamma(\mu_s)^* \epsilon_r = \gamma(\pi)^{r/s} \epsilon_r$. The remaining arguments are analogous to the semi-group case discussed in Proposition
\ref{Hreps}. 
\endproof

The $C^*$-algebra completion of $\cB_{(\Sigma,\sigma_n)}$ is the crossed product 
$\cA_{(\Sigma, \sigma_n)}=C^*(\Sigma)\rtimes \Q^\times_+$. 
The time evolution discussed in \S\ref{tevol} extends naturally as follows:

\begin{lemma}\label{ggroup}
The choice of a semi-group homomorphism $g: \N \to \R^\times_+$  determines a time evolution
on the $C^*$-algebra $\cA_{(\Sigma, \sigma_n)}=C^*(\Sigma)\rtimes \Q^\times_+$.
\end{lemma}

\proof Note that $g$ extends uniquely to a group homomorphism
$g: \Q^\times_+ \to \R^\times_+$. Therefore, it suffices to set $\sigma_t(\mu_r) := g(r)^{it} \mu_r$ 
for every $r\in \Q^\times_+$ and $\sigma_t(\pi):=\pi$ for every $\pi \in \Sigma$.
\endproof

\begin{remark}
The time evolution of Lemma \ref{ggroup} is the natural generalization of the one of Proposition \ref{tevlem}. However, it is clear that the resulting ``partial QSM-system'' $\cQ=(C^*(\Sigma)\rtimes \Q^\times_+, \sigma_t)$
does not have a convergent partition function $Z(\beta)$ for $\beta\gg 0$, neither low-temperature Gibbs states.
\end{remark}

\section{Examples of QSM-systems}\label{sec:examples}
In this section we describe in detail the QSM-systems associated to our examples of concrete Bost-Connes data (as in Notation \ref{not:embedding}).
\subsection*{Example 1: Original Bost--Connes system}
Let $k=\bbQ$. Recall from \S \ref{sec:BC-systems} the definition of the concrete Bost-Connes datum $(\bbQ/\bbZ,\sigma_n)$. What follows is standard and can be found in the foundational article of Bost and Connes \cite{BC}; consult also \cite[\S3]{CoMa}. The only novelty is that in the construction of the time evolution we consider more general
choices of the auxiliary semi-group homomorphism $g: \N \to \R^\times_+$.
\begin{lemma}
For every embedding $\iota \in \mathrm{Emb}(\bbQ/\bbZ, \overline{\bbQ}^\times)$, the countable multiplicative subgroup $N_\iota(\bbQ/\bbZ)$ of $\bbR^\times_+$ is equal to $\{1\}$.
\end{lemma}
\begin{proof} 
As explained in proof of Proposition \ref{alphaLem}, $\Hom(\Q/\Z,\overline{\Q}^\times)=\Hom(\Q/\Z,\Q/\Z)$. Therefore, $\iota(\bbQ/\bbZ)$ is contained in the subgroup of roots of unit of $\overline{\bbQ}^\times$. This implies that $N_\iota(s)=|\iota(s)|=1$ for every $s\in \Q/\Z$.
\end{proof}
The $\bbQ$-algebra $\cB_{(\bbQ/\bbZ,\sigma_n)}$ agrees with the cross product $\bbQ[\bbQ/\bbZ]\rtimes \bbN$, where the semi-group action of $\bbN$ on $\bbQ[\bbQ/\bbZ]$ is given by $n \mapsto (s \mapsto \sum_{s' \in \rho_n(s)}s')$. On the other hand, since $N_\iota(\bbQ/\bbZ)=\{1\}$, the $\bbQ$-algebra $\cB'_{(\bbQ/\bbZ,\sigma_n)}$ is not necessary for the construction of the QSM-system.

\smallskip 
 
Note that since $\alpha(n)=n$ and $N_\iota(\bbQ/\bbZ)=\{1\}$, the Hilbert spaces $\cH_{\alpha, \iota}^{\leq}$ are all equal to $\cH:=l^2(\bbN)$. Similarly, all the $\bbC$-vector spaces $\cV_{\alpha,\iota}^{\leq}$ are equal to $\cV:=\cV_\bbC$.

\smallskip

Let $\iota \in \mathrm{Emb}(\bbQ/\bbZ, \overline{\bbQ}^\times)\simeq \widehat{\bbZ}^\times$. The  representation $R_\iota$ of $\cB_{(\bbQ/\bbZ, \sigma_n)}$ on $\cV$ is given by $R_\iota(s)(\epsilon_n)= \iota(s)^n \epsilon_n$ and $R_\iota(\mu_m)(\epsilon_n)= \epsilon_{mn}$. Following Proposition \ref{prop:new}, this representation extends to a representation $R_\iota$ of the $\bbC$-algebra $\cB_{(\bbQ/\bbZ, \sigma_n, \bbC)}=\bbC[\bbQ/\bbZ]\rtimes \bbN$ by bounded operators on $\cH$. The $C^\ast$-algebra $\cA_{(\bbQ/\bbZ,\sigma_n)}$ identifies then with the closure $C^\ast(\bbQ/\bbZ)\rtimes \bbN$ of $\bbC[\bbQ/\bbZ]\rtimes \bbN$ inside the $C^\ast$-algebra of bounded operators $B(\cH)$.
\smallskip

Let $g: \bbN \to \bbR^\times_+$ be the standard embedding of $\bbN$ into $\bbR^\times_+$. The associated time evolution is given by $\sigma_t(s)=s$ and $\sigma_t(\mu_n)=n^{it}\mu_n$ and the associated Hamiltonian $H:=H_\iota: l^2(\bbN) \to l^2(\bbN)$ by $\epsilon_n \mapsto \mathrm{log}(n) \epsilon_n$. Consequently, the partition function $Z(\beta)$ agrees with the Riemann zeta function $\zeta(\beta)=\sum_{n\geq 1} n^{-\beta}$. 

\begin{remark}
Since $\bbN$ is the free commutative semi-group generated by the prime numbers, every semi-group homomorphism $g: \bbN \to \bbR^\times_+$ is determined by its values $\lambda_p:=g(p)$ at the prime numbers $p$. 
Therefore, we can write the 
partition function as an Euler product:
$$ Z(\beta):= \sum_{n \geq 1} g(n)^{-\beta} = \prod_p \sum_{k \geq 0} g(p)^{-k\beta} = \prod_p(1-\lambda_p^{-\beta})^{-1}\,.$$
The first equality is obtained using the primary decomposition of $n\in \N$ and the fact that $g$ is a 
semi-group homomorphism, and the second equality follows by summing the resulting geometric series.
\end{remark}

\begin{example}
Let $q$ be a prime power. Given an algebraic variety $X$ defined over $\bbZ$, let $\lambda_p:=\# X_p(\bbF_q)$ where $X_p$ stands for the reduction of $X$ modulo $p$. The corresponding semi-group homomorphism $g$ gives then rise to the partition function
$$ Z_X(\beta)= \prod_p (1-\#X_p(\bbF_q)^{-\beta})^{-1}\,.$$
\end{example}
The absolute Galois group $\mathrm{Gal}(\overline{\bbQ}/\bbQ)$ acts on $\bbQ/\bbZ$ through the quotient group $\mathrm{Gal}(\overline{\bbQ}^{\mathrm{ab}}/\bbQ)\simeq \widehat{\bbZ}^\times$. This action of $\widehat{\bbZ}^\times$ extends to $\bbQ[\bbQ/\bbZ]\rtimes \bbN$ and to the $C^\ast$-algebra $C^\ast(\bbQ/\bbZ) \rtimes \bbN$ as in Propositions \ref{prop:G-action} and \ref{prop:G-action2}.

\smallskip

Let $\beta_0>0$ be the exponent of convergence of the series 
$Z(\beta)=\sum_{n\geq 1} g(n)^{-\beta}$.
For $\beta >\beta_0$, $e^{-\beta H_\iota}$ is a trace class operator and 
we have Gibbs states of the form
$$ \varphi_{\iota,\beta}(a) = \frac{{\rm Tr}(R_\iota(a) e^{-\beta H_\iota})}{{\rm Tr}(e^{-\beta H_\iota})}. $$

\begin{lemma}\label{GibbsBCcase}
We have the following computation
$$ \varphi_{\iota,\beta}(a) = \left\{ \begin{array}{ll}0& \mathrm{when}\,\,a= s  \mu_m \mu_n^*, n \neq m \\
 Z(\beta)^{-1} \sum_{n\geq 1} \sum_j c_j \iota(s_j)^n \, g(n)^{-\beta} & \mathrm{when}\,\, a= \sum_j c_j s_j\in \Q[\Q/\Z]\,.
  \end{array}\right. 
$$
\end{lemma}

\proof
On an additive basis of $\cA_{(\Q/\Z,\sigma_n)}$, given by monomials of the form $s  \mu_m \mu_n^*$, we have 
$\varphi_{\iota,\beta}(s \mu_m \mu_n^*)= Z(\beta)^{-1} \sum_r \langle \epsilon_r ,   R_\iota(s  \mu_m \mu_n^*) e^{-\beta H_\iota} \, \epsilon_r \rangle =0$ when $n\neq m$. Therefore, the proof follows from the fact that $ \varphi_{\iota,\beta}(s) = Z(\beta)^{-1} \sum_{n\geq 1} \iota(s)^n \, g(n)^{-\beta}$ for every $s \in \bbQ/\bbZ$.
\endproof

\begin{remark}\label{GibbsBCpolylogs}
When $g(n)=n$, we have $ \varphi_{\iota,\beta}(s) =\zeta(\beta)^{-1} {\rm Li}_\beta(\iota(s))$, where $\zeta(\beta)$ is the Riemann zeta function and ${\rm Li}_\beta(\iota(s))$ the evaluation at roots of unit of a polylogarithm function.
\end{remark}

\begin{lemma}\label{groundBC}
The ground states are given by $\varphi_{\iota,\infty}(s) =\iota(s)$. 
Moreover, the $G$-action by automorphisms of  $\cA_{(\Q/\Z,\sigma_n)}$ agrees with the Galois $G$-action on $\iota(\Q/\Z)$.
\end{lemma}

\proof
When $\beta\to \infty$ the Gibbs states converge weakly to ground states of the form $ \varphi_{\iota,\infty}(s) = \langle \epsilon_1, R_\iota(s) e^{-\beta H_\iota} \epsilon_1 \rangle = \iota(s)$. The action of  $G:=\mathrm{Gal}(\overline{\bbQ}/\bbQ)$ on $C^\ast(\bbQ/\bbZ) \rtimes \bbN$ 
factors through the abelianization $G^{\mathrm{ab}}$. Hence, using the terminology of Definition \ref{def:GSigma}, we have $Z(G_\Sigma)=G_\Sigma=G^{\mathrm{ab}}$ and $\widetilde Z(G_\Sigma)=G$. Thanks to Proposition \ref{lem:fabulous}, $G$ acts then on the Gibbs states by $\gamma: \varphi_{\iota,\beta} 
\mapsto \varphi_{\iota\circ \gamma, \beta}$, and the 
induced action on the limits $\varphi_{\iota,\infty}(s) \mapsto 
\varphi_{\iota\circ \gamma, \infty}(s)$ agrees with the Galois action on $\iota(\Q/\Z)$.
\endproof

\smallskip

\subsection*{Example 3: Algebraic numbers}
Let $k=\bbQ$. 
Recall from \S \ref{sec:BC-systems} the definition of the concrete Bost-Connes datum $(\overline{\bbQ}^\times, \sigma_n)$. 
We consider a fixed embedding $\overline{\Q}^\times \subset \C^\times$.

\begin{definition}\label{Qbarleq}
Let $\overline{\Q}^\times_{\leq}:=\{ s\in \overline{\Q}^\times \,|\, |s|\leq 1\}$, $\cS:= \{ |s| \,|\, s\in \overline{\Q}^\times \}$ and $\cS_\leq =\{ |s| \,|\, s\in \overline{\Q}^\times_{\leq}\} $. Note that $\cS$ is a multiplicative subgroup of $\bbR^\times_+$ and $\cS_\leq$ a subset of $\cS \cap (0,1]$. 
\end{definition}

\begin{remark}\label{Sgeomprog}
The subgroup $\cS\subset \R^\times_+$ is a union
of infinitely many geometric progressions. Therefore, in the construction of the QSM-system, we will use the Hilbert space $\widetilde\cH_{\alpha,\iota}^\leq$; see Notation \ref{tildeHalphaiota}.
\end{remark}
 Let $\mathrm{Emb}_{\cS}(\overline{\bbQ}^\times, \overline{\bbQ}^\times)$ be the set of embeddings $\{
\iota \in \mathrm{Emb}(\overline{\bbQ}^\times, \overline{\bbQ}^\times)\,|\,
N_\iota(\overline{\Q}^\times)=\cS \}$ and  $\iota(\overline{\Q}^\times)_\leq :=\{ s\in \overline{\Q}^\times \,|\, N_\iota(s) \leq 1 \}$. In what follows we consider the following subset of $G$-equivariant embeddings:
\begin{equation}\label{Emb0barQ}
\mathrm{Emb}_0(\overline{\bbQ}^\times, \overline{\bbQ}^\times): =\{ 
\iota \in \mathrm{Emb}(\overline{\bbQ}^\times, \overline{\bbQ}^\times)\,|\,
\iota(\overline{\Q}^\times_\leq) = \iota(\overline{\Q}^\times)_\leq \}.
\end{equation}
\begin{lemma}\label{Emb0EmbS}
We have an inclusion $ \mathrm{Emb}_0(\overline{\bbQ}^\times, \overline{\bbQ}^\times)\subset
\mathrm{Emb}_{\cS}(\overline{\bbQ}^\times, \overline{\bbQ}^\times)$.
\end{lemma}

\proof Note that every embedding $\iota \in \mathrm{Emb}_0(\overline{\bbQ}^\times, \overline{\bbQ}^\times)$
maps the set $\{ s \in \overline{\Q}^\times \,|\, |s|\leq 1 \}$ isomorphically to the set 
$\{ s \in \overline{\Q}^\times \,|\, N_\iota(s) \leq 1 \}$. Consequently, $N_\iota(\overline{\Q}^\times_\leq)$
surjects onto $\cS_\leq$. This implies that $N_\iota(\overline{\Q}^\times)$ surjects onto $\cS$ since
any element of $\cS_> :=\cS\smallsetminus \cS_\leq$ is the absolute value $|s^{-1}|=|s|^{-1}$ of some
$s\in \overline{\Q}^\times_\leq$. Therefore, if the absolute values $N_\iota(s)=|\iota(s)|$, with $s\in \overline{\Q}^\times_\leq$, fill up the set $\cS_\leq$, $N_\iota(s^{-1})$ 
fill up the complementary subset $\cS_>$. We hence conclude that $N_\iota(\overline{\Q}^\times)$ 
 surjects onto~$\cS$.
\endproof

\medskip

The $\bbQ$-algebra $\cB_{(\overline{\bbQ}^\times, \sigma_n)}$ is generated by the elements $s \in \overline{\bbQ}^\times$ and by the partial symmetries $\mu_n, \mu_n^\ast, n \in \bbN$, as in Definition \ref{Akalg}. Similarly, the $\bbQ$-algebra $\cB'_{(\overline{\bbQ}^\times, \sigma_n)}$ is generated by $s, \mu_n, \mu_n^\ast$ and by the weight operators $W(\lambda), \lambda \in \bbQ^\times$; see Definition \ref{tildealg}. 
Given $\iota \in \mathrm{Emb}_0(\overline{\bbQ}^\times, \overline{\bbQ}^\times)$, 
we obtain a representation $R_\iota$ of $\cB'_{(\overline{\bbQ}^\times, \sigma_n)}$ on the
Hilbert space $\widetilde\cH^\leq_{\alpha,\iota}$ of Notation \ref{tildeHalphaiota}.

\begin{lemma}\label{Qn1lem}
For every embedding $\iota \in \mathrm{Emb}_{\cS}(\overline{\bbQ}^\times, \overline{\bbQ}^\times)$, we have
$R_\iota(\mu_n^\ast \mu_n)=1$.
\end{lemma}

\proof The operator $R_\iota(\mu_n^\ast \mu_n)$ is the projection onto the
subspace spanned by the vectors $\epsilon_{n,\eta}$ such that $\eta =\xi^n$ for some $\xi \in N_\iota(\overline{\Q}^\times)$.
Given an embedding $\iota \in \mathrm{Emb}_{\cS}(\overline{\bbQ}^\times, \overline{\bbQ}^\times)$, we have
$N_\iota(\overline{\Q}^\times)=\cS$. This implies that for every $\eta \in N_\iota(\overline{\Q}^\times)$ we have 
$\eta=|s|$ for some $s\in \overline{\Q}^\times$. Hence, we can always find an $n^{\mathrm{th}}$ root in 
$\cS$ by taking $\xi =|s^{1/n}|$ for some $n^{\mathrm{th}}$ root of $s$ in 
$\overline{\Q}^\times$. In conclusion, $R_\iota(\mu_n^* \mu_n)=1$.
\endproof

\begin{remark}
Lemma \ref{Qn1lem} implies that if $\iota \in \mathrm{Emb}_{\cS}(\overline{\bbQ}^\times, \overline{\bbQ}^\times)$, we can then work with the algebra $\cB'_{(\overline{\bbQ}^\times, \sigma_n)}$ with the additional relation 
$\mu_n^* \mu_n=1$.
\end{remark}

\smallskip

\begin{lemma}\label{goodQbardatum}
The pair $((\overline{\Q}^\times,\sigma_n), \mathrm{Emb}_0(\overline{\bbQ}^\times, \overline{\bbQ}^\times))$
is a very good concrete Bost--Connes datum, in the sense of Definition \ref{goodBC}.
\end{lemma}

\proof Recall from the proof of Lemma  \ref{Emb0EmbS} that for every
$\iota \in \mathrm{Emb}_0(\overline{\bbQ}^\times, \overline{\bbQ}^\times)$,
have $N_\iota(\overline{\Q}^\times)\cap (0,1]=\cS_\leq$. Hence, the Hilbert space
$\widetilde\cH_{\alpha,\iota}^\leq$ of Notation \ref{tildeHalphaiota} is independent of
$\iota$.
Moreover, $\iota$ maps isomorphically the set $\overline{\Q}^\times_\leq$ to $\iota(\overline{\Q}^\times)_\leq$. Now, recall that the algebra $\cB'_{(\overline{\bbQ}^\times, \sigma_n,\iota)}$ and its $C^*$-completion
$\cA'_{(\overline{\bbQ}^\times, \sigma_n,\iota)}$ have generators $s \in \iota(\overline{\Q}^\times)_\leq$.
The isomorphisms $\iota(\overline{\Q}^\times)_\leq \simeq \overline{\Q}^\times_\leq$ induce isomorphisms
between the algebras for different choices of $\iota$. Therefore, we can conclude that the algebras are also
independent of $\iota$.
\endproof

\smallskip

\begin{notation} We write $\widetilde\cH^\leq$ for the Hilbert space $\widetilde\cH^\leq_{\alpha,\iota}$ with $\alpha=\id$. Note that thanks to Lemma \ref{goodQbardatum}, $\widetilde\cH^\leq$ is independent of $\iota$. Similarly, we write $\cA'_{(\overline{\bbQ}^\times, \sigma_n)}$
for the $C^*$-algebra acting by bounded operators on $\widetilde\cH^\leq$ through the
representations $R_\iota$, with $\iota\in \mathrm{Emb}_0(\overline{\bbQ}^\times, \overline{\bbQ}^\times)$.
\end{notation}

\smallskip

Given a semi-group homomorphism $g: \bbN \to \bbR^\times_+$ and the homomorphism
$h: N_\iota(\overline{\Q}^\times) \to \R^\times_+, \lambda_r \mapsto p_r$, we obtain a 
time evolution $\sigma_t$ on $\cA'_{(\overline{\bbQ}^\times, \sigma_n)}$ as in Proposition
\ref{tevolprop2}, with a Hamiltonian as in Proposition \ref{rHam}, and a partition function as in 
Theorem \ref{thm:Zgeneral}. This gives rise to the following result:

\begin{proposition}\label{GibbsbarQ}
Let $\iota(\overline{\Q}^\times)_1:=\{ s\in \overline{\Q}^\times \,|\, N_\iota(s)=1 \}$.
When $g(n)=n$, and for all $\beta > 3/2$, the Gibbs states evaluated on elements $s\in \iota(\overline{\Q}^\times)_1$ are given by the following convergent series ($Z(\beta):=\sum_{n\geq 1} \zeta(\beta n) \, n^{-\beta}$):
\begin{equation}\label{GibbsbarQseries}
\varphi_{\iota,\beta}(s) = Z(\beta)^{-1} \sum_{n\geq 1} \iota(s)^n \, \zeta(\beta n) \, n^{-\beta}\,.
\end{equation}
\end{proposition}

\proof The partition function of the QSM-system was obtained in Theorem \ref{thm:Zgeneral}. Its trace ${\rm Tr}(R_\iota(s) e^{-\beta H_\iota})$ identifies with
$$ \sum  \langle \epsilon_{n,\lambda_r^{k_r}}, 
R_\iota(s) e^{-\beta H_\iota} \epsilon_{n,\lambda_r^{k_r}} \rangle = \sum \iota(s)^n g(n)^{-\beta} h(\lambda_r)^{-k_r n \beta}\,  \langle \epsilon_{n,\lambda_r^{k_r}}, 
\epsilon_{n,\lambda_r^{k_r +a_r(s)}}\rangle\,,$$
where the sum is taken over the elements of the orthonormal basis. The inner products vanish unless $a_r(s)=0$ for all $r$, \ie unless $N_\iota(s)=1$. In the
case $N_\iota(s)=1$, for the choice of $h: N_\iota(\overline{\Q}^\times) \to \R^\times_+,\lambda_r \mapsto p_r$,
we obtain 
$$ \sum_n \iota(s)^n g(n)^{-\beta}   \prod_r \sum_{k_r}  h(\lambda_r)^{-k_r n \beta} 
= \sum_n \iota(s)^n g(n)^{-\beta}  \prod_p (1-p^{-n \beta})^{-1} \,.$$
The product over $r$ reflects the fact that the Hilbert space $\widetilde\cH^\leq$ is a tensor product over $r$; see Theorem \ref{thm:Zgeneral}.
\endproof

\smallskip

In contrast with the original Bost--Connes case, the following result shows that the ``fabulous states" property is not satisfied! Intuitively speaking, $\overline{\Q}^\times$ is ``too large" to give rise to a well behaved Bost--Connes system.

\begin{proposition}\label{nofabprop}
In the limit $\beta \to \infty$ the ground states take values $\varphi_{\iota,\infty}(s)=\iota(s)$
on all $s\in \iota(\overline{\Q}^\times)_1$. The $G$-action
on $\overline{\Q}^\times$ induces the trivial action on the values of the ground states.
\end{proposition}

\proof In the limit $\beta\to \infty$, the Gibbs states $\varphi_{\iota,\beta}$ converge weakly
 to the ground states $\varphi_{\iota,\infty}$, which are given by the projection onto the kernel 
 $\epsilon_{1,1}$ of the Hamiltonian. The basis element $\epsilon_{1,1}$ is
 given by the vector $\epsilon_n \otimes \bigotimes_r \epsilon_{\lambda_r^{k_r}}$ with $n=1$
 and $k_r=0$ for all $r$, so that $\eta=\prod_r \lambda_r^{k_r}=1$. 
 Thus, we have 
 $$ \varphi_{\iota,\infty}(s) =\langle \epsilon_{1,1}, R_\iota(s) e^{-\beta H} \epsilon_{1,1} \rangle
 = \langle \epsilon_{1,1}, \iota(s) \epsilon_{1,N_\iota(s)} \rangle 
 = \left\{ \begin{array}{ll} \iota(s) & \mathrm{when}\,\,N_\iota(s)=1 \\
 0 & \mathrm{when}\,\,N_\iota(s)\neq 1\,.  \end{array}\right.  $$
 The absolute Galois group $G=\mathrm{Gal}(\overline{\bbQ}/\bbQ)$ acts on $\overline{\bbQ}^\times$.
 In this case $G_\Sigma =G$, hence we have $Z(G_\Sigma)=\{ 1 \}$, see Example \ref{ex:GSigma}. 
 Thus, the group $Z(G)\cap G_0$ is also trivial, hence the induced action on Gibbs states is trivial.
\endproof

\medskip

\begin{remark}
It is also possible to construct a Bost--Connes type QSM-system for $\overline{\Q}^\times$ in a different way: using the logarithmic height function as a specialization of a more general construction for 
toric varieties; see \cite[\S4]{JinMar}.
\end{remark}

\subsection*{Example 4 - Weil numbers of weight zero}\label{W0qQSM:sec}
Let $k=\bbQ$. Recall from Example \ref{ex:4} the definition of the concrete Bost-Connes datum $(\cW_0(q),\sigma_n)$.

\medskip

Let us denote by $\iota_0$ the fixed embedding $\cW(q)\subset \overline{\Q}^\times$, which restricts
to an embedding $\cW_0(q) \subset \overline{\Q}^\times$ of the Weil numbers of weight zero.
Let $G_{\cW(q)}$ and $G_{\cW_0(q)}$ be the quotients, as in Definition \ref{def:GSigma},
of the $G$-action on $\cW(q)$ and $\cW_0(q)$, respectively. Thanks to Proposition \ref{lem:aux1}(iii), the
$G$-action on $\cW(q)$ preserves weights. Therefore, $G_{\cW(q)}=G_{\cW_0(q)}$. 

\medskip

In what follows, we consider
the following subset of $G$-equivariant embeddings $\mathrm{Emb}_0(\cW_0(q),\overline{\bbQ}^\times)= \widetilde Z(G_{\cW(q)}) \cdot \iota_0$, where $\gamma \cdot \iota_0:=\iota_0\circ \gamma$.  

\begin{lemma} \label{W0N1}
For every embedding $\iota \in \mathrm{Emb}_0(\cW_0(q),\overline{\bbQ}^\times)$, 
the countable multiplicative subgroup $N_\iota(\cW_0(q))$ of $\bbR^\times_+$ is equal to $\{1\}$.
\end{lemma}

\proof Thanks to Proposition \ref{lem:aux1}(iii), the $G$-action on $\cW(q)$ preserves weights. Therefore,  for every $s \in \cW_0(q)$, we have $N_{\iota_0\circ \gamma}(s)=N_\iota(\gamma(s))=q^{w(\gamma(s))}=q^{w(s)}=N_{\iota_0}(s)=1$.
\endproof

\begin{lemma} \label{W0algcross}
The $\bbQ$-algebra $\cB_{(\cW_0(q),\sigma_n)}$ agrees with the cross product $\bbQ[\cW_0(q)] \rtimes \bbN$, 
where the semi-group action of $\bbN$ on $\bbQ[\cW_0(q)]$ is given by $ n \mapsto (s \mapsto \sum_{s' \in \rho_n(s)} s')$.
\end{lemma}

\proof 
The two algebras have the same sets of generators and relations. In fact, the crossed product algebra $\bbQ[\cW_0(q)] \rtimes \bbN$ is generated by the elements $s\in \cW_0(q)$ and by
isometries $\mu_p$, for the generators $p$ of the semi-group $\N$, and 
their adjoints $\mu_p^*$, with the semi-group action implemented by $s\mapsto \mu_p s \mu_p^*$.
The $\mu_p$ satisfy the relations $\mu_p\mu_{p'}=\mu_{pp'}$, $\mu_p^*\mu_p=1$, and
$\mu_{p'}^*\mu_p=\mu_p\mu_{p'}^*$ for $p\neq p'$. The subalgebra
generated by the $\mu_p$ and $\mu_p^*$ is clearly isomorphic to the subalgebra of $\cB_{(\cW_0(q),\sigma_n)}$
generated by the $\mu_n$ and $\mu_n^*$, by writing $\mu_n$ as a product of $\mu_p$'s according to the
primary decomposition of $n$. The subalgebras 
$\bbQ[\cW_0(q)]$ of $\cB_{(\cW_0(q),\sigma_n)}$ 
and $\bbQ[\cW_0(q)] \rtimes \bbN$ also match, and
the semi-group action $\mu_n s \mu_n^* =\sum_{s' \in \rho_n(s)} s'$
gives the remaining relation of $\cB_{(\cW_0(q),\sigma_n)}$.
\endproof

\begin{remark}
Since $N_\iota(\cW_0(q))=\{1\}$, the $\bbQ$-algebra $\cB'_{(\cW_0(q),\sigma_n)}$ is not necessary for the construction of the QSM-system.
\end{remark}

Note that since $\alpha(n)=n$ and $N_\iota(\cW_0(q))=\{1\}$, the Hilbert spaces $\cH_{\alpha, \iota}^{\leq}$ are all equal to 
$\cH:=l^2(\bbN)$. Similarly, all the $\bbC$-vector spaces $\cV_{\alpha,\iota}^{\leq}$ are equal $\cV:=\cV_\bbC$.
Given $\iota \in \mathrm{Emb}_0(\cW_0(q),\overline{\bbQ}^\times)$, 
the representation $R_\iota$ of $\cB_{(\cW_0(q),\sigma_n)}$ on $\cV$ is given by $R_\iota(s)(\epsilon)= \iota(s)^n \epsilon_n$ and $R_\iota(\mu_m)(\epsilon_n)=\epsilon_{mn}$. Following Proposition \ref{prop:new}, this representation extends to a representation $R_\iota$ of the $k$-algebra $\cB_{(\cW_0(q),\sigma_n,\bbC)}=\bbC[\cW_0(q)]\rtimes \bbN$ by bounded operators on $\cH$. The $C^\ast$-algebra $\cA_{(\cW_0(q),\sigma_n)}$ identifies then with the closure $C^\ast(\cW_0(q)) \rtimes \bbN$ of $\bbC[\cW_0(q)]\rtimes \bbN$ inside the $C^\ast$-algebra of bounded operators $B(\cH)$.

\smallskip

We construct the time evolution as in the case of the original Bost--Connes system. 
Let $g: \bbN \to \bbR^\times_+$ be the standard embedding of $\bbN$ into $\bbR^\times_+$. The associated time evolution is given by $\sigma_t(s)=s$ and $\sigma_t(\mu_n)=n^{it}\mu_n$ and the associated Hamiltonian $H:=H_\iota: l^2(\bbN) \to l^2(\bbN)$ by $\epsilon_n \mapsto \mathrm{log}(n) \epsilon_n$. Consequently, the partition function $Z(\beta)$ agrees with the Riemann zeta function $\zeta(\beta)=\sum_{n\geq 1} n^{-\beta}$. The series converges for $\beta>1$.

\smallskip

The absolute Galois group $G=\mathrm{Gal}(\overline{\bbQ}/\bbQ)$ acts on $\cW_0(q)$ through
the quotient $G_{\cW(q)}$. This action extends to $\bbQ[\cW_0(q)]\rtimes \bbN$ and to the 
$C^\ast$-algebra $C^\ast(\cW_0(q)) \rtimes \bbN$ as in Propositions \ref{prop:G-action} 
and~\ref{prop:G-action2}.

\smallskip

\begin{proposition}\label{KMSWeil0}
With the time evolution determined by $g(n)=n$,
the low temperature ($\beta>1$) Gibbs states of the quantum statistical mechanical system for
the concrete datum $(\cW_0(q),\sigma_n)$ are polylogarithms evaluated at numbers
$\pi\in \cW_0(q)$, normalized by the Riemann zeta function. The action of the Galois group $G$
as symmetries of the system induces an action of the subgroup $\widetilde Z(G_{\cW(q)})$ on
the zero temperature Gibbs states, which agrees with the restriction to $\widetilde Z(G_{\cW(q)})$
of the Galois action on $\iota(\cW_0(q))$.
\end{proposition}

\proof For $g(n)=n$, the Hamiltonian is $H\, \epsilon_n =\log(n)\, \epsilon_n$, with
partition function the Riemann zeta function, as in the original Bost--Connes case.
The low temperature Gibbs states, evaluated on $s\in \cW_0(q)$, are of the form
$$ \varphi_{\iota_0\circ \gamma,\beta}(s) = \zeta(\beta)^{-1}\, \sum_{n\geq 1} \frac{\gamma(\pi)^n}{n^\beta} =
\frac{{\rm Li}_\beta(\gamma(s))}{\zeta(\beta)}. $$
For $\beta\to \infty$ the weak limits of these Gibbs states define the zero temperature ground states.
These are given by the projection onto the kernel of the Hamiltonian. Evaluated on elements $s\in \cW_0(q)$, they give
$$ \varphi_{\iota_0\circ \gamma,\infty}(s) = \langle \epsilon_1 , R_\gamma(\pi) e^{-\beta H} \epsilon_1 \rangle  =\gamma(s). $$
The action of $G$ by automorphisms of the algebra determines an induced action of $\widetilde Z(G_{\cW(q)})$ 
on the ground states states, as in Proposition \ref{GGibbs}.
\endproof

\subsection*{Example 4 - Weil numbers}
Let $k=\bbQ$. Recall from Example \ref{ex:5} the definition of the concrete Bost-Connes datum $(\cW(q),\sigma_n)$. 
As in the preceding Example, we will make use of the notations $G_{\cW(q)}, \iota_0$, and $\mathrm{Emb}_0(\cW(q),\overline{\bbQ}^\times)=\widetilde Z(G_{\cW(q)}) \cdot \iota_0$.

\begin{lemma}
For every embedding $\iota \in \mathrm{Emb}_0(\cW(q),\overline{\bbQ}^\times)$, 
the countable multiplicative subgroup $N_\iota(\cW(q))$ of $\bbR_+^\times$ 
is equal to $\{q^r\,|\, r \in \frac{\bbZ}{2}\}$.
\end{lemma}
\begin{proof}
Thanks to Proposition \ref{lem:aux1}(iii), the $G$-action on $\cW(q)$ preserves weights. Therefore, 
 for every $s\in \cW_0(q)$, we have $N_{\iota_0\circ \gamma}(s)=N_\iota(\gamma(s))=q^{w(\gamma(s))}=q^{w(s)}$. This implies that $N_{\iota_0\circ \gamma}(\cW(q))=N_{\iota_0}(\cW(q))=q^{\frac{1}{2}\Z}$.
\end{proof}

\begin{remark}\label{oneprog}
Note that the semi-group $N_\iota(\cW(q))\subset\bbR_+^\times$ 
is given by a single geometric progression generated by the element $q^{1/2}$.
\end{remark}

The $\bbQ$-algebra $\cB_{(\cW(q),\sigma_n)}$ is generated by the elements $s \in \cW(q)$ and by the partial symmetries $\mu_n, \mu_n^\ast, n \in \bbN$, as in Definition \ref{Akalg}. Similarly, the $\bbQ$-algebra $\cB'_{(\cW(q),\sigma_n)}$ is generated by $s, \mu_n, \mu_n^\ast$ and by the weight operators $W(\lambda), \lambda \in \cW(q)$; see Definition \ref{tildealg}. 

\smallskip

\begin{lemma}\label{WqgoodBC}
The pair $((\cW(q),\sigma_n),\mathrm{Emb}_0(\cW(q),\overline{\bbQ}^\times))$ is a very good concrete
Bost--Connes datum, in the sense of Definition \ref{goodBC}.
\end{lemma}

\proof
Since $\alpha(n)=n$ and $N_\iota(\cW(q))$ is independent of $\iota$, the Hilbert spaces $\cH_{\alpha, \iota}^\leq$, resp. the $\bbC$-linear subspaces $\cV_{\alpha, \iota, \bbC}^\leq$, are all equal to the Hilbert subspace $\cH^\leq$ of $\cH:= l^2(\bbN) \otimes l^2(\{q^r\,|\, r \in \frac{\bbZ}{2}\})$, resp. to the $\bbC$-linear subspace $\cV^\leq$ of $\cV:=\cV_\bbC\otimes \cV_{\iota,\bbC}$, spanned by the elements $\epsilon_{n,q^r}$ with $r\leq 0$. 
Moreover, since by Proposition \ref{lem:aux1}(iii) the $G$-action on $\cW(q)$ preserves weights, the set $\{ s\in \cW(q)\,|\,
N_\iota(s)\leq 1 \}$ is independent of the embedding $\iota \in \mathrm{Emb}_0(\cW(q),\overline{\bbQ}^\times)$. As a consequence, the $\Q$-algebra $\cB'_{(\cW(q), \sigma_n,\iota)}$ and its $C^*$-completion
$\cA'_{(\cW(q), \sigma_n,\iota)}$ are independent of 
$\iota$.
\endproof

\begin{notation} Let $\cA'_{(\cW(q), \sigma_n)}$ be the resulting $C^*$-algebra acting on the Hilbert space $\cH^\leq$
through the representations $R_\iota$, with $\iota \in \mathrm{Emb}_0(\cW(q),\overline{\bbQ}^\times)$.
\end{notation}

\begin{remark}\label{QnprojWq}
Let $\iota \in \mathrm{Emb}_0(\cW(q),\overline{\bbQ}^\times)$. In contrast with the case of algebraic numbers, the operator $R_\iota(\mu_n^\ast \mu_n)$ is {\em not} the identity but rather the projection onto the subspace spanned by the vectors $\epsilon_{n,q^{r/2}}$ such that $n|r$.
\end{remark}

\smallskip

Given a semi-group homomorphism $g: \bbN \to \bbR^\times_+$, we obtain the time evolution 
\begin{eqnarray*}
\sigma_t(\pi)=W(\omega(\pi))^{-it}\pi & \sigma_t(\mu_n)=g(n)^{it} \mu_n & \sigma_t(W(\lambda))= W(\lambda)\,,
\end{eqnarray*}
where $\omega(\pi)$ stands for the weight of $\pi$. This gives rise to the Hamiltonian 
\begin{eqnarray*}
H:=H_\iota:\cH^\leq \too \cH^\leq && \epsilon_{n,q^r}\mapsto \mathrm{log}(q^{-nr}g(n))\epsilon_{n,q^r}
\end{eqnarray*}
and consequently to the partition function 
\begin{equation}\label{eq:zeta1}
Z(\beta):= \mathrm{Tr}(e^{-\beta H})=\sum_{n \in \bbZ} \sum_{r\leq 0  \in \frac{\bbZ}{2}} q^{nr\beta} g(n)^{-\beta}\,.
\end{equation}
We denote by $\beta_0$ the exponent of convergence of the series $\sum_n g(n)^{-\beta}$.

\begin{proposition}\label{thm:poly}
The partition function \eqref{eq:zeta1} is computed by the series
\begin{equation}\label{eq:function}
Z(\beta)=\sum_{n \geq 1} \frac{g(n)^{-\beta}}{1-q^{-n\frac{\beta}{2}}}\,,
\end{equation}
which converges for $\beta>\beta_0$ and diverges for $\beta\leq \beta_0$.
In the case where $g(n)=n$, 
\eqref{eq:function} can be written as the series of polylogaritms $Z(\beta)= \sum_{k\geq 0} {\rm Li}_\beta(q^{-k\frac{\beta}{2}})$.
\end{proposition}

\begin{proof}
The first claim follows directly from Proposition \ref{prop:partition} (with 
a single geometric progression with generator $q^{1/2}$). Concretely, we have:
$$
\sum_{n \in \bbZ} \sum_{r \leq 0 \in \frac{\bbZ}{2}} q^{nr\beta} g(n)^{-\beta} =  \sum_{n \geq 1} \sum_{k \geq 0} q^{-kn \frac{\beta}{2}} g(n)^{-\beta} = \sum_{n \geq 1} \frac{g(n)^{-\beta}}{1-q^{-n\frac{\beta}{2}}}\,,
$$
with the estimate
$$g(n)^{-\beta} \leq  \frac{g(n)^{-\beta}}{1- q^{-n\frac{\beta}{2}}} \leq \frac{g(n)^{-\beta}}{1- q^{-\frac{\beta}{2}}}\,. $$
We now assume that $g(n)=n$. For $\beta>\beta_0=1$, and after exchanging the order of summation, the above series can be re-written as 
$$\sum_{k\geq 0} \sum_{n\geq 1} q^{-kn \frac{\beta}{2}} n^{-\beta} =  
\sum_{k\geq 0} {\rm Li}_{\beta}(q^{-k \frac{\beta}{2}})\,.$$
\end{proof}

\smallskip

\begin{lemma}\label{GibbsWq}
When $g(n)=n$, for $\beta>1$, the Gibbs states, evaluated on elements $s\in \cW(q)$ are
zero for weight $w(s)\neq 0$, while for $s\in \cW_0(q)$ they are given by
\begin{equation}\label{eqGibbsWq}
\varphi_{\iota,\beta}(s)=\frac{\sum_{k\geq 0} {\rm Li}_\beta(\iota(s) q^{-k\beta/2})}
{\sum_{k\geq 0} {\rm Li}_{\beta}(q^{-k \beta/2})}.
\end{equation}
\end{lemma}

\proof The partition function is provided by Theorem \ref{thm:poly}. Thus, we just need to
compute the trace 
\begin{eqnarray*}
 {\rm Tr}(R_\iota(s) e^{-\beta H_\iota}) & = & \sum_{n\geq 1}\sum_{k\geq 0} \langle \epsilon_{n,q^{-k/2}}, 
R_\iota(s) e^{-\beta H_\iota} \epsilon_{n,q^{-k/2}} \rangle \\[2mm]
& =  & \sum_{n\geq 1}\sum_{k\geq 0} \iota(s)^n \, q^{-k n \beta/2} \, g(n)^{-\beta}\,  \langle \epsilon_{n,q^{-k/2}}, 
\epsilon_{n,q^{-k/2+w(s)}} \rangle \\[2mm]
& = & \left\{ \begin{array}{ll} \sum_{n\geq 1}\sum_{k\geq 0} \iota(s)^n \, q^{-k n \beta/2} \, g(n)^{-\beta} &
\mathrm{when}\,\,w(s)=0 \\ 0 & \mathrm{when}\,\,w(s)\neq 0. \end{array}\right. 
\end{eqnarray*}
This implies that $\varphi_{\iota,\beta}(s)=0$ when $w(s)\neq 0$. For $s\in \cW_0(q)$, and after
exchanging the order of summation, the above expression can be re-written as follows:
$$ \sum_{k\geq 0} \sum_{n\geq 1} \iota(s)^n \, q^{-k n \beta/2} \, n^{-\beta} =
\sum_{k\geq 0} {\rm Li}_\beta(\iota(s) q^{-k\beta/2})\,.$$
As a consequence, we obtain the above equality \eqref{eqGibbsWq}:
$$ \varphi_{\iota,\beta}(s)=  \frac{{\rm Tr}(R_\iota(s) e^{-\beta H_\iota})}{Z(\beta)}
=\frac{\sum_{k\geq 0} {\rm Li}_\beta(\iota(s) q^{-k\beta/2})}
{\sum_{k\geq 0} {\rm Li}_{\beta}(q^{-k \beta/2})}. $$
\endproof

Recall that the absolute Galois group $G=\mathrm{Gal}(\overline{\bbQ}/\bbQ)$ acts on $\cW(q)$ through the
quotient $G_{\cW(q)}$. This action extends to $\cB'_{(\cW(q), \sigma_n)}$ and to the 
$C^\ast$-algebra $\cA'_{(\cW(q), \sigma_n)}$ as in Propositions \ref{prop:G-action} and 
\ref{prop:G-action2}.

\smallskip

\begin{proposition}\label{groundWq}
In the limit $\beta\to \infty$ the ground states are given by
$$ \varphi_{\iota,\infty}(s) = \left\{ \begin{array}{ll} \iota(s) & \mathrm{when}\,\,w(s)=0 \\ 0 & \mathrm{when}\,\,w(s)\neq 0. \end{array} \right. $$
The $G$-action on $\cW(q)$ induces an action of the subgroup $\widetilde Z(G_{\cW(q)})\subset G$
on the ground states, which agrees with the Galois action on the values $\iota(\cW_0(q))$. 
\end{proposition}

\proof The ground states are given by projections onto the kernel of the Hamiltonian $H_\iota$. Therefore, we obtain the following equalities:
$$ \varphi_{\iota,\infty}(s) = \langle \epsilon_{1,1}, R_\iota(s) e^{-\beta H_\iota}\epsilon_{1,1}\rangle  = \iota(s) \, \langle \epsilon_{1,1},\epsilon_{1,q^{w(s)}} \rangle =\left\{ \begin{array}{ll} \iota(s) & \mathrm{when}\,\,w(s)=0 \\
0 & \mathrm{when}\,\,w(s)\neq 0 \,.\end{array}\right. $$
The $G$-action by automorphisms of the algebra $\cA'_{(\cW(q), \sigma_n)}$
induces an action of $\widetilde Z(G_{\cW(q)})\subset G$ on 
the Gibbs states and on the ground states by $\varphi_{\iota,\beta}\mapsto \varphi_{\iota\circ \gamma, \beta}$. This follows from the fact that for every $\iota \in \mathrm{Emb}_0(\cW(q),\overline{\bbQ}^\times)$ and
$\gamma \in \tilde Z(G_{\cW(q)})$, we have $\iota\circ\gamma \in \mathrm{Emb}_0(\cW(q),\overline{\bbQ}^\times)$.
This action on ground states agrees with the Galois action on the values, since
$\varphi_{\iota\circ \gamma, \infty}(s)=\iota(\gamma(s))=\gamma(\iota(s))$, by $G$-equivariance of
the embeddings.
\endproof

\smallskip

Examples 5 and 6 of \S\ref{sec:BC-systems} are only abstract
Bost--Connes data. Example 5 is not a concrete datum because $\alpha$
is not a semi-group homomorphism. The case of Germs in Example 6
is also not a concrete datum, because alpha is the trivial homomorphism $\alpha(n)=1$.
Moreover, we do not have an embedding of $\cW_0(p^\infty)$ in $\overline{\Q}^\times$, so
even the partial construction for $\alpha(n)=1$ discussed in \S\ref{alpha1sec} 
does not apply.  The case of the completion $\widehat{\cW}^L(q)$ in Example 7 of 
\S \ref{sec:BC-systems} is also not a concrete Bost--Connes datum.
We consider the remaining cases in \S\ref{sec:Weil+completion}.
\section{Weil restriction and completion}\label{sec:Weil+completion}

The concrete Bost-Connes data of Examples 2 (=Weil restriction) and 7 (=Completion) do not satisfy the assumption of Notation \ref{not:embedding}. Nevertheless, we explain briefly in this section how they still give rise to QSM-systems. The key idea is to consider them as ``diagonal subsystems'' of larger QSM-systems. The latter are related to the higher rank Bost-Connes systems introduced in \cite{Mar,JinMar}. We only give an outline of the constructions. The details, along with a general treatment of ``high rank Bost-Connes data'', will appear in a forthcoming article.

\subsection*{Example 2: Weil restriction}
Let $k=\bbR$. Recall from \S\ref{sec:BC-systems} the definition of the concrete Bost-Connes datum $(\bbQ/\bbZ \times \bbQ/\bbZ, \sigma_n)$. In this case, $\alpha(n)=n^2$.

We now construct a large QSM-system (which is not associated to a Bost-Connes datum) and an involution on it. The QSM-system associated to $(\bbQ/\bbZ \times \bbQ/\bbZ, \sigma_n)$ will be defined as a subsystem.
Given a pair $(n,m) \in \bbN^2$, let $\sigma_{n,m}$ be the homomorphism $(n\cdot-, m\cdot-): \bbQ/\bbZ \times \bbQ/\bbZ \to \bbQ/\bbZ \times \bbQ/\bbZ$ and $\rho_{n,m}$ the associated map $\bbQ/\bbZ \times \bbQ/\bbZ \to \cP(\bbQ/\bbZ) \times \cP(\bbQ/\bbZ)$ that sends an element $(s_1,s_2)$ to its pre-image under $\sigma_{n,m}$. Let $\cB_{(\Sigma,\sigma_{n,m})}$ be the $\bbR$-algebra generated by the elements $(s_1,s_2) \in \bbQ/\bbZ \times \bbQ/\bbZ$ and by the isometries $\mu_{n,m}, \mu_{n,m}^*$ 
with $(n,m)\in \N^2$. We assume that $\mu_{n,m} \mu_{k,l}=\mu_{nk, ml}$, that $\mu_{n,m}^* \mu_{n,m}=1$, and
\begin{eqnarray*}
\mu_{n,m}\mu_{k,l}^* = \mu_{k,l}^* \mu_{n,m} &\mathrm{when}& (n,k)=(m,l)=1
\end{eqnarray*}
$$\mu_{n,m} (s_1,s_2) \mu_{n,m}^* = \frac{1}{n m} \sum_{(s_1',s_2')\in \rho_{n,m}(s_1,s_2)} (s_1', s_2')\,.$$
\begin{remark}
Intuitively speaking, the $\bbR$-algebra $\cB_{(\Sigma,\sigma_{n,m})}$ is a ``higher rank" generalization of the one of Definition \ref{Akalg}, where the semi-group homomorphism $\alpha$ is now given by $\bbN^2 \to \bbN, (n,m) \mapsto nm$. The generalizations include the higher rank Bost-Connes algebras \cite{Mar,JinMar} and will be discussed in a future work.  
\end{remark}

We have isomorphisms of $\R$-algebras 
\begin{equation}\label{eq:tensor-algebras}
\cB_{(\Sigma,\sigma_{n,m})}  \simeq  \R[\Q/\Z]^{\otimes 2} \rtimes \N^2 \simeq
\cB_{(\Q/\Z,\sigma_n)} \otimes_\R  \cB_{(\Q/\Z,\sigma_m)}
\end{equation}
where $\cB_{(\Q/\Z,\sigma_n)}=\R[\Q/\Z]\rtimes \N$ is the $\bbR$-algebra of the original Bost--Connes system. We have a $\bbZ/2$-action on \eqref{eq:tensor-algebras} which switches the two copies of $\cB_{(\Q/\Z,\sigma_n)}$. Consider the Hilbert space $\ell^2(\N \times \N)$ equipped with the standard orthonormal basis $\{ \epsilon_{n,m} \}$. Via the identification between $\epsilon_{n^2}$ and 
$\epsilon_{n,n}$, the Hilbert space $\ell^2(\alpha(\N))$ can be regarded as a subspace of $\ell^2(\N \times \N)$. Given an embedding $\iota=(u_1,u_2)\in \hat\Z^\times\times \hat\Z^\times=\
\mathrm{Emb}(\Q/\Z, \overline{\bbQ}^\times) \times \mathrm{Emb}(\bbQ/\bbZ, \overline{\bbQ}^\times)$, the assignments
\begin{eqnarray*}
 R_\iota (\mu_{n,m}) \epsilon_{k,l} := \epsilon_{nk, ml} &&
 R_\iota(s_1,s_2) \epsilon_{k,l} := u_1(s_1)^k \, u_2(s_2)^l\, \epsilon_{k,l}
\end{eqnarray*} 
define a representation $R_\iota$ of the $\bbR$-algebra $\cB_{(\Sigma,\sigma_{n,m})}$ on $\ell^2(\N^2)$. The $C^*$-completion $\cA_{(\Sigma,\sigma_{n,m})}$ of $\cB_{(\Sigma,\sigma_{n,m})}\otimes_\R \C$ is isomorphic to $C^*(\Q/\Z)^{\otimes 2} \rtimes \N^2$. This is a particular case of a higher rank Bost--Connes algebra \cite{Mar}. Given a semi-group homomorphism $\hat g: \N\times \N \to \R^\times_+$, the assignments $\sigma_t(s_1,s_2):=(s_1,s_2)$ and
$\sigma_t(\mu_{n,m}):=\hat g(n,m)^{it}\mu_{n,m}$ define a time evolution on $\cA_{(\Sigma,\sigma_{n,m})}$.

\smallskip

We now construct the QSM-system associated to $(\Q/\Z\times \Q/\Z, \sigma_n)$. 
Let us write $\cB_{(\Q/\Z \times \Q/\Z,\sigma_n)}$ for the subalgebra of 
$\cB_{(\Q/\Z \times \Q/\Z,\sigma_{n,m})}$ generated by the elements
$(s_1,s_2)\in \Q/\Z \times \Q/\Z$ and by the isometries $\mu_{n,n}$ with $n\in \N$.
Under the representations $R_\iota$, $\cB_{(\Q/\Z \times \Q/\Z,\sigma_n)}$
preserves the subspace $\ell^2(\alpha(\N)) \subset \ell^2(\N^2)$. 
Therefore, the $C^*$-completion $\cA_{(\Q/\Z \times \Q/\Z,\sigma_n)}$ of 
$\cB_{(\Q/\Z \times \Q/\Z,\sigma_n)} \otimes_\R \C$ can be identified with
$C^*(\Q/\Z \times \Q/\Z)\rtimes \N$, with the semi-group action given by
\begin{eqnarray*}
(s_1,s_2) & \mapsto & \mu_{n,n}(s_1,s_2) \mu_{n,n}^*:=\frac{1}{n^2} \sum_{(s_1',s_2') \in \rho_{n,n}(s_1, s_2)}(s_1',s_2')\,.
\end{eqnarray*}
Note that the assignments $(s_1, s_2) \mapsto (s_2,s_1)$ and $\mu_{n,n} \mapsto \mu_{n,n}$ define a $\bbZ/2$-action by automorphisms of $\cA_{(\Q/\Z \times \Q/\Z,\sigma_n)}$. Let $\hat g: \N^2 \to \R^\times_+$ be a semi-group homomorphism of the form
$\hat g(n,m):=g(n)g(m)$, with $g: \N \to \R^\times_+$ a semi-group homomorphism.
The time evolution determined by $\hat g$ preserves the $C^\ast$-algebra 
$\cA_{(\Q/\Z \times \Q/\Z,\sigma_n)}$. Moreover, the Hamiltonian generating the restriction of the time evolution to $\cA_{(\Q/\Z \times \Q/\Z,\sigma_n)}$,
in the representations $R_\iota$ on the Hilbert space $\ell^2(\alpha(\N))$ is given by
$H \epsilon_{n,n}=2\log(g(n))\, \epsilon_{n,n}$. Hence, the partition function
$Z(\beta)$ agrees with $\sum_n g(n)^{-2\beta}$. Finally, the $\bbZ/2$-action on $\cA_{(\Q/\Z \times \Q/\Z,\sigma_n)}$ is compatible with
the time evolution.

\subsection*{Example 7: Completion} 
Let $k=\bbQ$. Recall from Example \ref{ex:9} the definition of the concrete Bost-Connes datum 
$(\widehat{\cW}(q),\sigma_n)$. In this case, $\alpha(n)=n^2$. As proved in 
Proposition \ref{lem:aux3}, we have a natural identification
$$ (\widehat{\cW}(q),\sigma_n) = (\cW_0(q), \sigma_n) \times (\Q/2\Z, \sigma_n)\,. $$
The QSM-system associated to $(\widehat{\cW}(q),\sigma_n)$ will be obtained as a subsystem.

Let $\cB_{(\cW_0(q)\times \Q/2\Z ,\sigma_{n,m})}$ be the $\Q$-algebra
generated by the elements $s=(\pi,r)\in \cW_0(q)\times \Q/2\Z$ and by the isometries $\mu_{n,m}$ with $(n,m)\in \N^2$. We assume that $ \mu_{n,m}\mu_{k,l}=\mu_{nk,ml}$, that $\mu_{n,m}^*\mu_{n,m}=1$, and that
\begin{eqnarray*}
\mu_{n,m}\mu_{k,l}^*= \mu_{k,l}^*\mu_{n,m} & \mathrm{when} & (n,k)=(m,l)=1
\end{eqnarray*}
$$ \mu_{n,m} (\pi,r)  \mu_{n,m}^* =\frac{1}{nm} \sum_{(\xi^{n}, mr')=(\pi,r)} (\xi, r')\,. $$
The $\bbQ$-algebra $\cB_{(\cW_0(q)\times \Q/2\Z ,\sigma_{n,m})}$ identifies with the semi-group crossed 
product $\Q[\cW_0(q)\times \Q/2\Z] \rtimes \N^2$. As in the preceding example, the Hilbert spaces 
$\ell^2(\alpha(\bbN))$ can be regarded as a subspace of $\ell^2(\bbN^2)$. Now, consider the following set of embeddings 
\begin{equation}\label{Emb0compl}
 \mathrm{Emb}_0(\cW_0(q) \times \Q/2\Z, \overline{\bbQ}^\times\times \overline{\bbQ}^\times) =
 \tilde Z(G_{\cW(q)}) \cdot \iota_0 \times  \mathrm{Emb}(\Q/2\Z,\Q/\Z), 
 \end{equation}
where $\tilde Z(G_{\cW(q)}) \cdot \iota_0 \subset \mathrm{Emb}(\cW_0(q),\overline{\Q}^\times)$ is
defined as in \S \ref{W0qQSM:sec}. 
Given an embedding $\iota=(\iota_0\circ\gamma, u) \in \eqref{Emb0compl}$, the assignments 
\begin{eqnarray*}
R_\iota(\mu_{n,m}) \epsilon_{k,l} :=  \epsilon_{nk,ml} && R_\iota(\pi,r) \epsilon_{k,l} :=  \iota_0(\gamma(\pi))^k u(r)^l \epsilon_{k,l}
\end{eqnarray*}
defines a representation $R_\iota$ of the $\bbQ$-algebra $\cB_{(\cW_0(q)\times \Q/2\Z ,\sigma_{n,m})}$ on $\ell^2(\N^2)$. 
The $C^\ast$-completion $\cA_{(\cW_0(q)\times \Q/2\Z ,\sigma_{n,m})}$ of
$\cB_{(\cW_0(q)\times \Q/2\Z ,\sigma_{n,m})}\otimes_\Q \C$ 
can be identified with $C^*(\cW_0(q)\times \Q/2\Z) \rtimes \N^2$. 
The choice of $\hat g: \N^2 \to \R^*_+$ determines a time evolution on 
$\cA_{(\cW_0(q)\times \Q/2\Z ,\sigma_{n,m})}$, with 
$\sigma_t(\mu_{m,n})=\hat g(m,n)^{it} \mu_{m,n}$ and $\sigma_t(\pi,r)=(\pi,r)$.

\smallskip

We now construct the QSM-system associated to $(\cW_0(q)\times \Q/2\Z ,\sigma_n)$. 
Let $\cB_{(\cW_0(q)\times \Q/2\Z ,\sigma_n)}$ be the subalgebra of
$\cB_{(\cW_0(q)\times \Q/2\Z ,\sigma_{n,m})}$ generated by the elements $s=(\pi,r)\in \cW_0(q)\times \Q/2\Z$ and by the isometries $\mu_{n,n}$. Under the representations $R_\iota$, $\cB_{(\cW_0(q)\times \Q/2\Z ,\sigma_n)}$
preserves the subspace $\ell^2(\alpha(\N))\subset \ell^2(\N^2)$. Therefore, the $C^*$-completion
$\cA_{(\cW_0(q)\times \Q/2\Z ,\sigma_n)}$ can be identified with the $C^\ast$-algebra $C^*(\cW_0(q)\times \Q/2\Z)\rtimes \N$, with the semi-group acting diagonally by
\begin{eqnarray*}
(\pi,r) & \mapsto & \mu_{n,n} (\pi, r) \mu_{n,n}^* = \frac{1}{n^2} \sum_{(\xi^n,nr')=(\pi,r)} (\xi,r')\,.
\end{eqnarray*}
Let $\hat g: \N^2 \to \R^\times_+$ be a semi-group homomorphism of the form $\hat g(n,m):=g(n) g(m)$, with $g:\N \to \R^\times_+$ a semi-group
homomorphism. The time evolution determined
by $\hat g$ preserves the $C^\ast$-algebra $\cA_{(\cW_0(q)\times \Q/2\Z ,\sigma_n)}$. Moreover, the Hamiltonian implementing it in
the representation $R_\iota$ on $\ell^2(\alpha(\N))$ is given by $H \epsilon_{n,n}= 2 \log(g(n)) \epsilon_{n,n}$. Hence, the partition function $Z(\beta)$ agrees with $\sum_n g(n)^{-2\beta}$. In the particular case where $g(n)=n$ the
Gibbs states are of the following form:
\begin{equation*}
\varphi_{\beta,\iota}(\pi,r)= \frac{{\rm Li}_{2\beta}( \iota_0(\gamma(\pi)) u(r) )}{\zeta(2\beta)}.
\end{equation*}
The group $\widetilde Z(G_{\cW(q)})$ acts by automorphisms of the algebras 
$\cB_{(\cW_0(q)\times \Q/2\Z ,\sigma_{n,m})}$ and $\cA_{(\cW_0(q)\times \Q/2\Z ,\sigma_{n,m})}$.
This action preserves the subalgebras $\cB_{(\cW_0(q)\times \Q/2\Z ,\sigma_n)}$
and $\cA_{(\cW_0(q)\times \Q/2\Z ,\sigma_n)}$ and
induces an action on the set of representations by $R_\iota \mapsto R_{\iota \circ \gamma}$.
Finally, the action on the subalgebras is compatible with the time evolution and agrees with the Galois 
action on the values of ground states at elements $s=(\pi,r)\in \cW_0(q)\times \Q/2\Z$.

  \end{document}